\documentclass[12pt]{amsart}
\usepackage{amsmath,amsfonts,amssymb}

\theoremstyle{plain}
\newtheorem{Lem}{Lemma}[section]
\newtheorem{Th}[Lem]{Theorem}

\newtheorem{Prop}[Lem]{Proposition}

\theoremstyle{definition}

\numberwithin{equation}{section}

\def\beq{\begin{equation}}
\def\eeq{\end{equation}}
\def\beqa{\begin{eqnarray}}
\def\eeqa{\end{eqnarray}}

\DeclareMathOperator{\Supp}{Supp}

\def\dist{\mathop{\rm dist}}
\def\diam{\mathop{\rm diam}}

\font\eightrm=cmr8

\def\be{\beta}
\def\ga{\gamma}
\def\Ga{\Gamma}
\def\del{\delta}

\def\ep{\varepsilon}

\def\La{\Lambda}
\def\om{\omega}

\def\ti{\tilde}
\def\vphi{\varphi}
\def\C{\mathcal{C}}
\def\scr{\scriptstyle}

\def\p{\partial}

\newcommand{\olg}{\bar{\gamma}}
\newcommand{\olp}{\bar{\phi}}
\newcommand{\ola}{\bar{\alpha}}
\newcommand{\ZZ}{\mathbb{Z}}
\newcommand{\tga}{\tilde{\gamma}}
\newcommand{\bs}{\backslash}

\newcommand{\hatp}{\hat{\varphi}}
\newcommand{\hatm}{\hat{\mu}}

\newcommand{\mL}{\mathcal{L}}
\newcommand{\mT}{\mathcal{T}}
\newcommand{\mI}{\mathcal{I}}
\newcommand{\mM}{\mathcal{M}}
\newcommand{\mG}{\mathcal{G}}

\newcommand{\mS}{\mathcal{S}}

\newcommand{\mD}{\mathcal{D}}
\newcommand{\mP}{\mathcal{P}}

\begin{document}
\title[Layering and Wetting Transitions]{Layering and wetting transitions\break
 for an SOS interface}
\author[K.S. Alexander]{Kenneth S. Alexander}
\address{Department of Mathematics KAP 108\\
University of Southern California\\
Los Angeles, CA  90089-2532 USA}
\email{alexandr@usc.edu}
%\thanks{Research supported by NSF grant DMS-0405915.}
\author[F. Dunlop]{Fran\c cois Dunlop}
\address{Laboratoire de Physique Th{\'e}orique et Modelisation (CNRS,
  UMR 8089)\\ 
Universit{\'e} de Cergy-Pontoise, 95302 Cergy-Pontoise\\
France}
\email{Francois.Dunlop@u-cergy.fr}
\author[S. Miracle-Sol\'e]{Salvador Miracle-Sol\'e}
\address{Centre de Physique Th{\'e}orique, CNRS, Case 907 \\
13288   Marseille cedex 9, France}
\email{miracle@cpt.univ-mrs.fr}


%\date%{July 27, 2008} 

\begin{abstract}
We study the solid-on-solid interface model above a 
horizontal wall in three dimensional space, 
with an attractive interaction when the 
interface is in contact with the wall, at low temperatures. 
There is no bulk external field.
The system presents a sequence of layering transitions, 
whose levels increase with the temperature, 
before reaching the wetting transition.
%and complete wetting above a certain value of this quantity.
\end{abstract}
\maketitle

%%%%%%%%%%%%%%%%%%%%%%%%%%%%%%%%%%%%%%%%%%%%%%%%%%%%%%%%
\newpage

\section{Introduction and results}

We consider the square lattice ${\bf Z}^2$, 
and to each site $x=(x_1,x_2)\in{\bf Z}^2$ we associate 
an integer variable $\phi_x\ge 0$ 
which represents the height of the interface at this site. 
The system is first considered in a finite box $\Lambda\subset{\bf Z}^2$
with fixed values of the heights outside.
Each interface configuration on $\Lambda$: $\{\phi_x\}$, $x\in \Lambda$,
denoted $\phi_\Lambda$, has
an energy defined by the Hamiltonian
\beqa 
H_\Lambda^{\{WAB\}}(\phi_\Lambda\mid\olp)=
2J_{AB}\sum_{\langle x,x'\rangle\cap \Lambda\ne\emptyset}|\phi_x-\phi_{x'}|
+2(J_{WA}+J_{AB})\sum_{x\in \Lambda}(1-\delta(\phi_x))\cr
+2J_{WB}\sum_{x\in \Lambda}\delta(\phi_x)%\cr
%=2(J_{WA}+J_{AB})|\La|
%+2J_{AB}\sum_{\langle x,x'\rangle\cap \Lambda\ne\emptyset}|\phi_x-\phi_{x'}|
%\hskip2cm\cr
%+2(J_{WB}-J_{WA}-J_{AB})\sum_{x\in \Lambda}\delta(\phi_x)
%\nonumber%\label{H}
\eeqa
or equivalently, simplifying the notation to $J=J_{AB}$ and 
$K=J_{WB}-J_{WA}$, by the Hamiltonian
\beq 
H_\Lambda(\phi_\Lambda\mid\olp)=2J
\sum_{\langle x,x'\rangle\cap \Lambda\ne\emptyset}|\phi_x-\phi_{x'}|
+2(K-J)\sum_{x\in \Lambda}\delta(\phi_x)
\label{H}\eeq
where $J>0$ and $K\in {\bf R}$,
the function $\delta$ equals $1$ when $\phi_x=0$, and $0$ otherwise, 
and $|\Lambda|$ is the number of sites in $\Lambda$. 
The first sum is taken over all nearest neighbors pairs 
$\langle x,x'\rangle\subset{\bf Z}^2$ such that at least one
of the sites belongs to $\Lambda$, and one takes $\phi_{x}=\olp_{x}$
when $x\not\in \Lambda$,
the configuration $\olp$ being the boundary condition,
assumed to be uniformly bounded. 

In the space ${\bf R}^3$, the region 
obtained as the union of all unit cubes centered at the sites 
of the three dimensional lattice $\Lambda\times\left( \frac{1}{2}+{\bf Z} \right)$,  
which satisfy $0< x_3<\phi(x_1,x_2)$, 
is supposed to be occupied by fluid $A$,  
the union of all unit cubes centered at the sites 
which satisfy $x_3>\phi(x_1,x_2)$ 
is supposed to be occupied by fluid $B$.
The common boundary between these regions
is a surface in ${\bf R}^3$, 
the microscopic interface $\mathcal{I}$.
The union of all unit cubes centered at the sites 
which satisfy $x_3<0$ 
is considered as the substrate, also called the
wall $W$. 

The considered system differs from the usual SOS model
by the restriction to non-negative height
variables and the introduction of the second and third sums in
the Hamiltonian.  
This term describes the interaction with the substrate.

The probability of the configuration $\phi_\Lambda$
at the inverse temperature $\beta=1/kT$
is given by the finite volume Gibbs measure
\beq
\mu_\Lambda(\phi_\Lambda\mid\olp)=\Xi(\Lambda,{\olp})^{-1} 
\exp\big(-\beta H_\Lambda(\phi_\Lambda\mid \olp)\big),
\label{mu}\eeq
where $\Xi(\Lambda,{\olp})$ is the partition function
\beq
\Xi(\Lambda,\olp)=\sum_{\phi_\Lambda} 
\exp\big(-\beta H_\Lambda(\phi_\Lambda\mid\olp)\big).
\label{Xi}\eeq
Local properties at equilibrium can be described
by correlation functions between the heights on
finite sets of sites, such as,
\beq
\langle f(\phi_{x_1},\dots,\phi_{x_k})\rangle_\Lambda^{\olp}=
\sum_{\phi_\Lambda}
\mu_\Lambda(\phi_\Lambda\mid\olp)f(\phi_{x_1},\dots,\phi_{x_k}). 
\label{rho}\eeq

We next briefly 
%discuss some general results, 
%which are an adaptation to our case
recall the setting of the wetting transition, as developed
%of analogous results established 
by Fr\"ohlich and Pfister
(refs. \cite{FPa}, \cite{FPb}, \cite{FPc})
for the semi-infinite Ising model. 
We do not formally translate these results to the SOS context, so
they should be taken as descriptive background.

Let $\Lambda\subset{\bf Z}^2$ be a rectangular box 
of sides parallel to the axes. 
Consider the boundary condition
$\olp_x=0$, for all $x\not\in\Lambda$,
and write $\Xi(\Lambda,0)$ for the corresponding 
partition function.
The associated free energy per site,   
\beq 
\tau^{\scriptscriptstyle WB}=2(J_{WA}+J_{AB})
-\lim_{\Lambda\to\infty}(1/\beta|\Lambda|)
\ln \Xi(\Lambda,0),  
\label{tau}\eeq
represents the surface tension between the medium $B$ and the substrate $W$. 

This limit (\ref{tau}) exists and $\tau^{\scriptscriptstyle WB}\le
2\min\{J_{WB},J_{WA}+J_{AB}\}$.  
One can introduce the densities  
\beq
\rho_z=\lim_{\Lambda\to\infty}\sum_{z'=0}^z
\langle \delta(\phi_x-z')\rangle^{(0)}_\Lambda,
\quad 
\rho_0=\lim_{\Lambda\to\infty}
\langle \delta(\phi_x)\rangle^{(0)}_\Lambda,
\eeq
where $(0)$ denotes the height-0 boundary condition.
Their connection with the surface free energy is given by the formula 
\beq
\tau^{\scriptscriptstyle WB}(\beta,K)=\tau^{\scriptscriptstyle WB}
(\beta,0)+2\int_0^K \rho_0(\beta,K')dK' .
\label{efe}\eeq

The surface tension $\tau^{\scriptscriptstyle WA}$
between the fluid $A$ and the substrate is 
$\tau^{\scriptscriptstyle WA}=J_{WA}$. 
%since, according to Hamiltonian (\ref{H}), 
%there is no interaction between them.
In order to define the surface tension $\tau^{\scriptscriptstyle AB}$
associated to a horizontal interface between the fluids $A$ and $B$ 
we consider the ordinary SOS model, with height-0 boundary condition, and
Hamiltonian %equal to (\ref{H}) with $\del(\phi_x)$ replaced by 0.
\beq 
H_\Lambda^{SOS}(\phi_\Lambda\mid\olp)=2J_{AB}
\sum_{\langle x,x'\rangle\cap \Lambda\ne\emptyset}(1+|\phi_x-\phi_{x'}|)
%\label{H}
\eeq
The corresponding free energy gives $\tau^{\scriptscriptstyle AB}$,
obeying $\tau^{\scriptscriptstyle AB}\le 2J_{AB}$.
With the above definitions, we have
\beq
\tau^{\scriptscriptstyle WA}(\beta)+ 
\tau^{\scriptscriptstyle AB}(\beta)\ge  
\tau^{\scriptscriptstyle WB}(\beta,K) .
\label{t} 
\eeq 
and the right hand side in (\ref{t}) is a monotone increasing and 
concave (and hence continuous) function of the parameter $K$.
This follows from relation (\ref{efe}) where the integrand
is a positive decreasing function of $K$.
Moreover, when $K\ge J$ equality is satisfied in (\ref{t}).

In the thermodynamic description of wetting, the partial wetting situation 
is characterized by the strict inequality in equation (\ref{t}),
which can occur only if $K<J$, as assumed henceforth. 
We must have then $\rho_0>0.$
The complete wetting situation is characterized by the equality in
(\ref{t}).
If this occurs for some $K$, say $K'<J$, then equation (\ref{efe})
tells us that this condition is equivalent to $\rho_0=0$.
Then both conditions, the equality and $\rho_0=0,$ hold 
for any value of $K$ in the interval $(K',J)$. 

On the other hand, we expect that $\rho_0=0$ implies also that 
$\rho_z=0$, for any positive integer $z$.
This indicates that, in the limit $\Lambda\to\infty$,
we are in the $A$ phase of the system, despite the height-0 
boundary condition, so that the medium $B$
cannot reach anymore the wall. 
This means also that the Gibbs state of the SOS model does not exist
in this case.  

That such a situation of complete wetting is present for some $K<J$
does not follow, however, from the above results.
This statement, as far as we know, remains an open problem
for the semi-infinite Ising model in 3 dimensions.
For the model \eqref{H} an answer to this problem has
been given by Chalker \cite{Ch}:

\bigskip\noindent{\bf Chalker's theorem: }
The following propositions hold
\beqa 
&\hbox{if}\quad 2\beta(J-K)>-\ln{{1-e^{-2\beta J}}\over{16(1+e^{-2\beta J})}}, 
&\hbox{ then}\quad \rho_0>0 , \label{chalker1}\\
&\hbox{if}\quad 2\beta(J-K)<-\ln(1-e^{-8\beta J}), &\hbox{ then}\quad\rho_0=0. 
\label{chalker2}
\eeqa

Thus, for any given values of $J$ and $K$, 
there is a temperature below which the interface is almost surely bound and 
another higher one above which it is almost surely unbound and  
complete wetting occurs. 
An illustration of these results, in the plane of the parameters
$(K,\beta^{-1})$, is given in Figure 1.

Here we investigate the intermediate region not 
covered by this theorem, when the temperature is low enough.
We will prove that 
a sequence of layering transitions
occurs before the system attains complete wetting.
%More precisely our main results can be summarized as follows.
We shall use the following notation: 
\beq 
u=2\beta(J-K),\quad t=e^{-4\beta J}.
\label{ut}\eeq
The variable $t$ may be viewed as the cost of each pair of plaquettes comprising the interface, and
$u$ is the gain per site of contact of the interface with the substrate.  

\begin{Th} 
Let the integer $n\ge0$ be given. 
For each $\epsilon>0$ there exists a value $t_0(n,\epsilon)>0$ 
such that, if the parameters $t,u,$ satisfy 
$0<t<t_0(n,\epsilon)$ and
\begin{align}
-\ln(1-t^2)+(2+\epsilon)t^{n+3}&<u<-\ln(1-t^2)+(2-\epsilon)t^{n+2} \quad &\text{if } n \geq 1, \notag \\
-\ln(1-t^2)+(2+\epsilon)t^3&<u<\sqrt{t} \quad &\text{if } n=0,
\label{Ass1}\end{align}
then the following statements hold:
(1) 
The free energy $\tau^{\scriptscriptstyle WB}$ is an
analytic function of the parameters $t,u$.
(2) 
There is a unique 
translation invariant
Gibbs state $\mu_n$ %a pure phase associated to the level $n$, 
satisfying %\break
$\mu_n(\{\phi_x\ne n\})=O(t^2)$ for all $x$.
(3) 
The density is $\rho_0>0$. 
\label{Th1}\end{Th}

An illustration for this theorem, in the plane $(K,\beta^{-1})$, 
is given in Figure~2.
It reflects the principle that if
the parameter $K$ is kept fixed, which seems natural 
since it depends on the properties of the substrate,
then the level $n$ of the unique translation invariant Gibbs state
increases when the temperature is increased. 
Rigorously, this principle would be in part a consequence of
\eqref{Ass1} if we could take $\epsilon=0$. 

%\newpage
%%%%%%%%%%%%%%%%%%%%%%%%%%%%%%%%%%%%%%%%%%%%%%%%%%%%%%%%%%%%%%%%%%%%
%%%                    Fig  1, 2
%%%%%%%%%%%%%%%%%%%%%%%%%%%%%%%%%%%%%%%%%%%%%%%%%%%%%%%%%%%%%%%%%%%%%
\vskip-2cm
\setlength{\unitlength}{.7mm}
\begin{center}
\begin{picture}(0,80)
\bezier{200}(-60,60)(40,30)(40,0)
%\bezier{200}(-20,70)(15,50)(40,30)(40,0)
\put(-30,-10){\vector(0,1){70}}
\put(-60,0){\vector(1,0){120}}

\bezier{200}(-60,20)(5,14)(40,0)
\put(-27,62){$\beta^{-1}$}
\put(64,-2){$K$}

\put(40,-2){\line(0,1){2}}
\put(38,-8){$J$}

\put(-31,20){\line(1,0){2}}
\put(-28,20){$\beta_R^{-1}$}

\put(-23,5){{\eightrm partial wetting}}
\put(25,30){{\eightrm wetting}}
\put(22,35){{\eightrm complete}}
\end{picture}

\bigskip%\bigskip
\bigskip\bigskip

\begin{minipage}{12cm}
{\footnotesize 
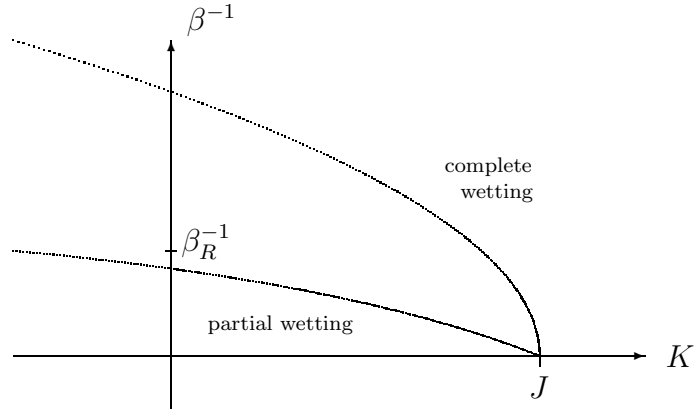
Figure 1. Illustration of Chalker's results.
The slope of the ``partial wetting curve'' at $(K=J, \be^{-1}=0)$
is $-2\ln2$. The approximate value of the roughening temperature for
the SOS model with values in {\bf Z}, expected near
$\beta_RJ\simeq0.4$, indicates the scale on the vertical axis.}
\end{minipage}

\end{center}

\bigskip%\bigskip

%\bigskip\bigskip

\begin{center}
\begin{picture}(0,80)

\bezier{400}(30,70)(40,35)(40,0)

\bezier{400}(10,60)(40,30)(40,0)
\put(-13,60){\line(1,0){23}}
\put(3,53){\footnotesize $\mu_0$}

\bezier{400}(20,56)(40,28)(40,0)
\put(14,56){\line(1,0){6}}
\put(15,58){\footnotesize $\mu_1$}

\bezier{400}(28,53)(40,26)(40,0)
\put(22,53){\line(1,0){6}}
\put(22,55){\footnotesize $\mu_2$}

\bezier{400}(32,51)(40,25)(40,0)
\put(29,51){\line(1,0){3}}
\put(28.5,53){\footnotesize $\mu_3$}

\put(-30,-10){\vector(0,1){77}}
\put(-20,0){\vector(1,0){80}}

%\bezier{400}(-20,20)(5,16)(40,0)
\put(-27,72){$\beta^{-1}$}
\put(64,-2){$K$}

\put(40,-2){\line(0,1){2}}
\put(38,-8){$J$}

\put(44,25){{\eightrm wetting}}
\put(41,30){{\eightrm complete}}
\end{picture}
\bigskip%\bigskip\bigskip
\vskip1cm
\begin{minipage}{12cm}
{\footnotesize
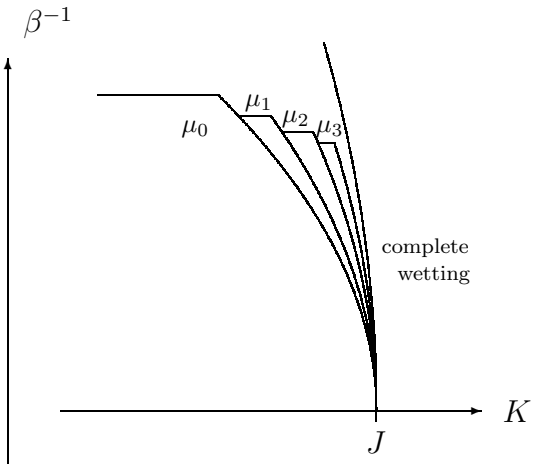
Figure 2. Qualitative picture of the unicity regions established in Theorem \ref{Th1}.  There are small
gaps between the regions, due to $\epsilon>0$. Actual unicity regions should be larger.
Magnified view
around the point $(K=J, \be^{-1}=0)$. }
\end{minipage}

\end{center}

%%%%%%%%%%%%%%%%%%%%%%%%%%%%%%%%%%%%%
%\newpage
\bigskip

\noindent{\bf Remarks: }
\begin{enumerate}
\item
The analyticity of the free energy comes from the existence
of a convergent cluster expansion for this system.
This implies the analyticity, in a direct way, of some
correlation functions and, in particular, of the density $\rho_0$.

%\item
%The unicity of the Gibbs state means that the 
%correlation functions (\ref{rho}) converge, when
%$\Lambda\to\infty$, to a limit that does not depend on the chosen 
%(uniformly bounded) boundary condition $\olp_x$.
%Being unique and translation invariant, this state represents
%a pure phase.

%In what sense it is associated to a level $n$ will become clearer
%later when some of its properties will be discussed.
%Let us say that for the typical configurations of the state
%large portions of the interface are near to the level $n$.

\item
By unicity of the translation invariant Gibbs state we mean that the
average over translations of local observables in finite volume Gibbs
states with any uniformly bounded boundary conditions lead to the same
infinite volume measure, independent of the boundary conditions.  

\item
The condition $\rho_0>0$ means that the interface remains at a 
finite distance from the wall and hence, we have partial wetting.
We can see from Theorem \ref{Th1} that the region where this condition
holds is much larger,  at low temperatures,
than the region initially proved by Chalker.
It comes very close to the line above which it is known 
that complete wetting occurs. 

\item
The values $t_0(n,\epsilon)$ for which we are able to prove Theorem \ref{Th1} satisfy
$t_0(n,\epsilon)\to 0$ when $n\to\infty$ or $\epsilon\to0$.
\end{enumerate}

The dependence of  $t_0(n,\epsilon)$ on $\epsilon$,
satisfying Remark 4, may be understood as follows.
One believes that the regions of unicity of the state
extend in such a way that two neighboring regions,
say those corresponding to the levels $n$ and $n+1$,
will have a common boundary where the two states $\mu_n$
and $\mu_{n+1}$ coexist.
At this boundary there will be a first order phase transition, since
the two Gibbs states are different.
The curve of coexistence does not necessarily exactly coincide with the curve 
$u=-\ln(1-t^2)+2t^{n+3}$.
Theorem \ref{Th1} says that it is however very near to it,
if the temperature is sufficiently low.

Let us formulate in the following statement the kind of theorem
that we expect, though we are not able to prove it.
We think that such a statement could be proved using
an extension of the Pirogov-Sinai theory \cite{Sin}.
This would certainly require some additional work and, 
in particular, a refinement of the notion of contours.

\medskip

\noindent{\bf Statement: } 
For each integer $n\ge0$, there exists $t_0(n)>0$
and a continuous function $u=\psi_{n+1}(t)$ on the 
interval $0<t<t_0(n)$ such that 
the statements of Theorem \ref{Th1} hold for $t$ in this interval,   
in the region where $\psi_{n+1}(t)<u<\psi_n(t)$ and for $n=0$, 
in the region $\psi_{1}(t)<u$. 
When $u=\psi_{n+1}(t)$ the two Gibbs states, $\mu_n$ and $\mu_{n+1}$, 
coexist. 

\medskip

The existence of such a sequence of layering transitions has been proved for
the SOS model with a bulk external field. See the works by 
Dinaburg, Mazel \cite{DM}, Cesi, Martinelli \cite{CM} and
Lebowitz, Mazel \cite{LM}.
This model has the same set of configurations as the model considered
here, but a different energy:
The second term in (\ref{H}) has to be
replaced by the term $+h\sum_{x\in\Lambda}\phi_x$ to obtain the
Hamiltonian of the model with an external magnetic field.
Effectively, each model incorporates a potential $V(\phi_x)$, with $V(n)=hn$ in the earlier work
and $V(n) = 2(J-K)\delta(n)$ here.  The purely local nature of the latter potential makes the layering transitions
perhaps \emph{a priori} less natural, and more difficult to prove.
Nonetheless, the method that we follow for the proof of 
Theorem \ref{Th1} is inspired by the method
developed for the study of the model with external field and we
shall have occasion to refer at various points
to the works mentioned above.
The most important difference between the two systems
concerns the restricted ensembles and the computation
of the associated free energies,
a point that will be discussed in Sections \ref{RE} and 
\ref{ProofProp}.
In ref. \cite{DM}
an analogous result to that of Theorem \ref{Th1} of the present work
has been proved for the model with external field.
In ref. \cite{CM} these results were extended  
to a proof of a theorem analogous to the Statement above.
Ref. \cite{LM} contains new inequalities for the model, allows 
$t_0>0$ to be independent of $n$ 
(so that, for small $t$, there is an infinite sequence of layering transitions)
and strengthens the form of uniqueness established for the Gibbs state.

References \cite{CM} and \cite{LM} make use of an ``infinite-height'' boundary condition which 
dominates all other boundary conditions, in the FKG sense.  The fact that the measure given by this 
boundary condition is well-defined depends on the fact that the interaction $h\phi_x \to \infty$ as the 
height $\phi_x \to \infty$.  The analogous statement fails in our context here, and the finite-volume measure with
infinite-height boundary condition does not exist.  Lacking this tool, we are restricted to results 
analogous to \cite{DM}.

Without bulk external field, the problem of layering transitions has been considered before for the Ising model, heuristically and numerically by several authors, and mathematically mostly by Basuev (\cite{Ba} and references therein). Ref. \cite{Ba} deals mostly with the case of a positive bulk field, coexisting with a surface field, but the case of zero bulk field is considered in section 9 of \cite{Ba}.

Theorem \ref{Th1} indicates that infinitely many phase transition
lines start from the point $K=J$, $\beta^{-1}=0$ in the plane $(K,\beta^{-1})$.
Where and how do these lines end? Five possible scenarios are shown in
Fig. 3 (a-e). Cases (a-b-c) are inspired by the analog diagrams of
Binder and Landau for the Ising model (Fig. 1 p. 2 in \cite{BL}).

\newpage
%%%%%%%%%%%%%%%%%%%%%%%%%%%%%%%%%%%%%%%%%%%%%%%%%%%%%%%%%%%%%%%%%%%%
%%%                    Fig  1, 2
%%%%%%%%%%%%%%%%%%%%%%%%%%%%%%%%%%%%%%%%%%%%%%%%%%%%%%%%%%%%%%%%%%%%%

\newpage
\font\eightrm=cmr8
%\nopagenumbers
%\begin{document}
\begin{center}
\setlength{\unitlength}{.7mm}
\begin{picture}(0,70)

% case (a)
\put(-80,0){\vector(0,1){60}}
\put(-100,0){\vector(1,0){100}}
\put(-77,55){$\beta^{-1}$}
\put(-6,2){$K$}
\put(-10,-2){\line(0,1){2}}
\put(-12,-8){$J$}
\put(-81,30){\line(1,0){2}}
\put(-78,30){$\beta_R^{-1}$}
\bezier{50}(-100,30)(-50,30)(0,30)
\bezier{200}(-50,30)(-10,25)(-10,0)
\bezier{200}(-50,30)(-10,10)(-10,0)
\bezier{200}(-30,19)(-10,13)(-10,0)
\bezier{200}(-38,24)(-10,17)(-10,0)
\bezier{200}(-44,27)(-10,20)(-10,0)
\bezier{200}(-46.5,28.5)(-10,22)(-10,0)
\bezier{200}(-100,59)(-70,50)(-50,30)
\bezier{30}(-100,46)(-70,40)(-50,30)
\put(-45,42){{\eightrm complete}}
\put(-42,37){{\eightrm wetting}}
\put(-95,35){\eightrm 0}
\put(-50,15){\eightrm 0}
\put(-21,8){\eightrm 1}
\put(-35,18){\eightrm 2}
\put(-43,23){\eightrm 3}
\put(-100,-6){\eightrm (a): triple points accumulate at the}
\put(-100,-11){\eightrm roughening temperature from below}

% case (b)
\put(30,0){\vector(0,1){60}}
\put(10,0){\vector(1,0){100}}
\put(33,56){$\beta^{-1}$}
\put(104,2){$K$}
\put(100,-2){\line(0,1){2}}
\put(98,-8){$J$}
\put(29,30){\line(1,0){2}}
\put(32,30){$\beta_R^{-1}$}
\bezier{50}(10,30)(60,30)(110,30)
\bezier{400}(10,60)(100,47)(100,0)
\bezier{200}(47,52)(100,13)(100,0)
\bezier{200}(58,48)(100,13)(100,0)
\bezier{200}(70,43)(100,13)(100,0)
\bezier{200}(79,37.5)(100,13)(100,0)
\bezier{200}(84,34)(100,13)(100,0)
\bezier{200}(87,31)(100,15)(100,0)
\bezier{20}(10,57)(20,56)(40,54)
\put(80,42){{\eightrm complete}}
\put(83,37){{\eightrm wetting}}
\put(15,45){\eightrm 0}
\put(60,15){\eightrm 0}
\put(57,44){\eightrm 1}
\put(68,40){\eightrm 2}
\put(78,35){\eightrm 3}
\put(10,-6){\eightrm (b): triple points accumulate at the}
\put(10,-11){\eightrm roughening temperature from above}
\end{picture}
\end{center}
%\smallskip

\begin{center}
\setlength{\unitlength}{.7mm}
\begin{picture}(0,80)

% case (d)
\put(-80,0){\vector(0,1){60}}
\put(-100,0){\vector(1,0){100}}
\put(-77,55){$\beta^{-1}$}
\put(-6,2){$K$}
\put(-10,-2){\line(0,1){2}}
\put(-12,-8){$J$}
\put(-81,30){\line(1,0){2}}
\put(-78,30){$\beta_R^{-1}$}
\bezier{50}(-100,30)(-50,30)(0,30)
\bezier{200}(-50,30)(-10,25)(-10,0)
\bezier{200}(-50,30)(-10,10)(-10,0)
\bezier{200}(-50,30)(-10,14)(-10,0)
\bezier{200}(-50,30)(-10,17)(-10,0)
\bezier{200}(-50,30)(-10,20)(-10,0)
\bezier{200}(-50,30)(-10,22)(-10,0)
\bezier{200}(-100,61)(-70,50)(-50,30)
\bezier{30}(-100,42)(-70,40)(-50,30)
\put(-45,42){{\eightrm complete}}
\put(-42,37){{\eightrm wetting}}
\put(-95,35){\eightrm 0}
\put(-50,15){\eightrm 0}
\put(-100,-6){\eightrm (d): multiphase point at the}
\put(-100,-11){\eightrm roughening temperature}

% case (c)
\put(30,0){\vector(0,1){60}}
\put(10,0){\vector(1,0){100}}
\put(33,55){$\beta^{-1}$}
\put(104,2){$K$}
\put(100,-2){\line(0,1){2}}
\put(98,-8){$J$}
\put(29,30){\line(1,0){2}}
\put(32,30){$\beta_R^{-1}$}
\bezier{50}(10,30)(60,30)(110,30)
\bezier{400}(10,60)(100,47)(100,0)
\bezier{50}(10,50)(50,47)(90,29)
\bezier{200}(40,46)(100,13)(100,0)
\bezier{200}(57,41)(100,13)(100,0)
\bezier{200}(70,37)(100,13)(100,0)
\bezier{200}(78,34)(100,13)(100,0)
\bezier{200}(83,32)(100,13)(100,0)
\bezier{200}(86,31)(100,15)(100,0)
\put(80,42){{\eightrm complete}}
\put(83,37){{\eightrm wetting}}
\put(20,40){\eightrm 0}
\put(60,15){\eightrm 0}
\put(53,39){\eightrm 1}
\put(68,34){\eightrm 2}
\put(77,31){\eightrm 3}
\put(10,-6){\eightrm (c): critical points accumulate at the}
\put(10,-11){\eightrm roughening temperature from above}
\end{picture}
\end{center}
%\smallskip

\begin{center}
\setlength{\unitlength}{.7mm}
\begin{picture}(0,80)

% case (e)
\put(-80,0){\vector(0,1){60}}
\put(-100,0){\vector(1,0){100}}
\put(-77,55){$\beta^{-1}$}
\put(-6,2){$K$}
\put(-10,-2){\line(0,1){2}}
\put(-12,-8){$J$}
\put(-81,35){\line(1,0){2}}
\put(-78,32){$\beta_R^{-1}$}
\bezier{50}(-100,35)(-50,35)(0,35) % dotted line
\bezier{400}(-100,50)(-10,35)(-10,0)
\bezier{200}(-39,33)(-10,5)(-10,0)
\bezier{200}(-32,29.5)(-10,7)(-10,0)
\bezier{200}(-27,26.5)(-10,8)(-10,0)
\bezier{200}(-23,23.5)(-10,9)(-10,0)
\bezier{200}(-19,20)(-10,9)(-10,0)
\bezier{30}(-100,42)(-70,40)(-43,35)
\put(-26,37){{\eightrm complete}}
\put(-25,32){{\eightrm wetting}}
\put(-95,37){\eightrm 0}
\put(-50,15){\eightrm 0}
\put(-35,31){\eightrm 1}
\put(-30,28.5){\eightrm 2}
\put(-25,25){\eightrm 3}
\put(-100,-6){\eightrm (e): triple points accumulate}
\put(-100,-11){\eightrm at zero temperature}
\end{picture}
\end{center}

\vskip 1cm
\centerline{\footnotesize 
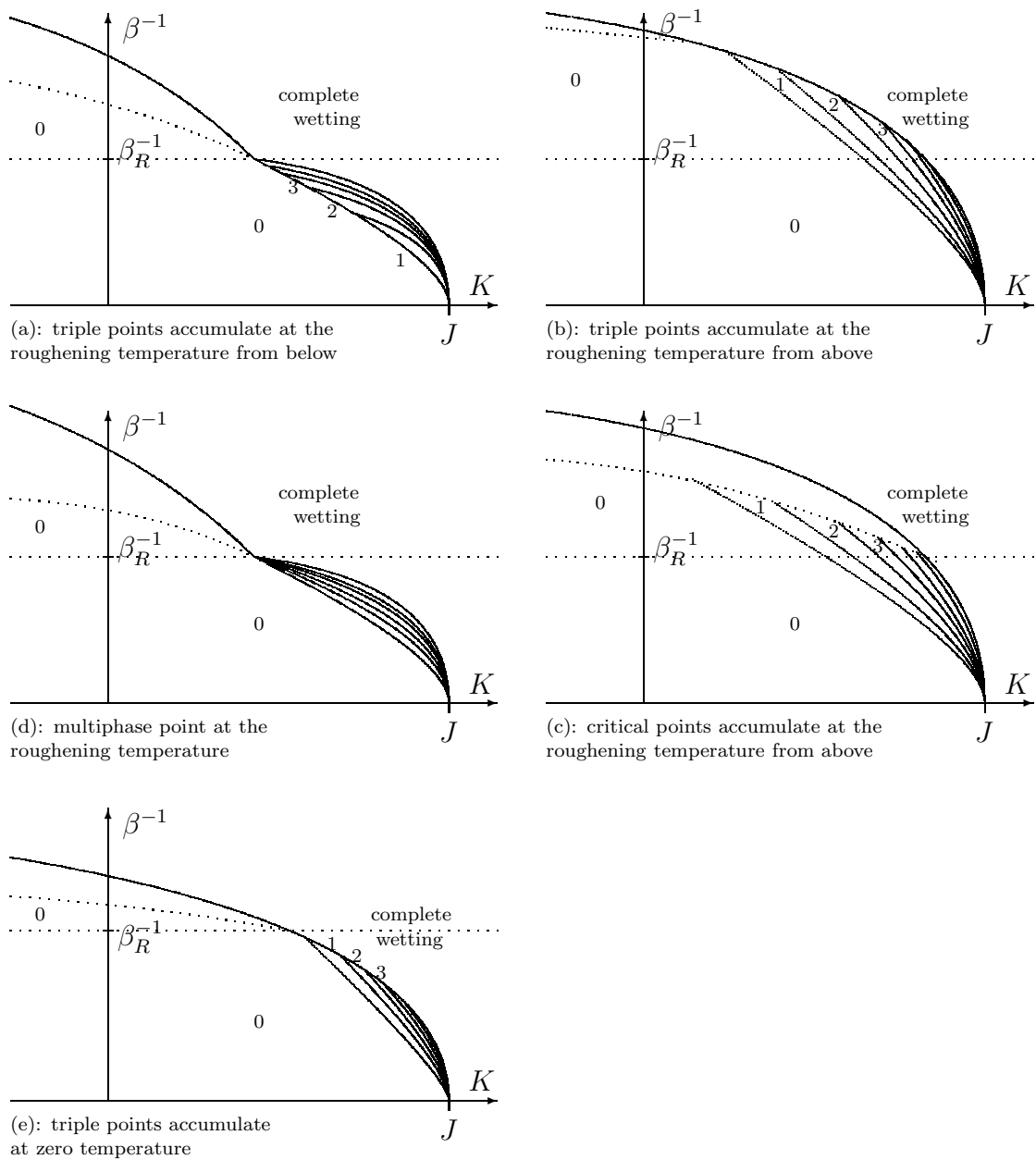
Fig. 3. Tentative phase diagrams (a)-(e)}

%\bigskip\bigskip

%\bigskip\bigskip

\newpage
%\section{Cylinder models, restricted ensembles}\label{CM}
\section{Cylinder models}\label{CM}
We consider the model in a box $\Lambda$ under the constant boundary 
condition $\olp_x=n$, for any given integer $n\ge0$. 
As we have seen, 
every configuration $\phi_\Lambda$ can naturally be considered
as a surface $\mathcal{I}$ imbeded into ${\bf R}^3$. 
Under the given boundary condition all points of the 
boundary $\partial{\mathcal{I}}$ are at height $n$, and 
we can write
\beq
\beta H_\La(\phi_\Lambda|n)=2\beta J\bigl(|{\mathcal{I}}|-|\La|\bigr)
-u|{\mathcal{I}}\cap W| ,
\label{HH}\eeq
where $W$ is the horizontal plane at height 0. 
The possible interfaces are connected sets of unit squares, 
also called plaquettes, and $|\cdot|$ denotes the number of these plaquettes.
%The set of vertical plaquettes of a given interface $\mathcal{I}$ 
%splits into maximally connected components which are called 
%{\it Dobrushin walls}.

Dinaburg and Mazel \cite{DM} have shown that such an interface can also be 
obtained in a unique way from the horizontal plane plane $\mathcal{P}_n$,
at height $n$, 
by adding ``positive cylinders'' and digging ``negative cylinders''.
The order of operations is not specified, 
but one may decide that larger ones are placed first.
%The cylinders must not intersect, 
%but may be in contact at wall angles or at basement and ceiling only.
%Walls of previous cylinders cannot be used nor erased. 

Formally, a {\it cylinder}, $\gamma=({\tga},E,I)$, is defined by 
its base perimeter ${\tga}$, 
% a connected set of bonds in the dual lattice ${\bf Z}^2+(1/2,1/2)$ 
% such that only an even number of bonds passes through every site
% in ${\tga}$, in the dual lattice ${\bf Z}^2+(1/2,1/2)$, 
a closed path of bonds in the dual lattice ${\bf Z}^2+(1/2,1/2)$,   
and two different nonnegative integers: 
$E$, the exterior level, and 
$I$, the interior level. 
%We know that the number of paths of length $l$ can be bounded by $3^l$.
The closed path $\tga$ must remain a single closed path when
self-intersections are removed by connecting south to west and north
to east as in Figure~4. Then $\tga$ can be slightly
deformed into a simple path by rounding corners.\\[10 pt]

\begin{center}
\begin{picture}(20,20)

\bezier{150}(10,12)(10,10)(12,10)
\bezier{150}(8,10)(10,10)(10,8)

\put(0,10){\line(1,0){8}}
\put(11,10){\line(1,0){8}}
\put(10,12){\line(0,1){8}}
\put(10,0){\line(0,1){8}}

\end{picture}

{\footnotesize
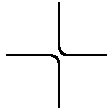
Figure 4. Rounding corners at self-intersections to make a simple path}\\[10 pt]
\end{center}

The sign of $\gamma$ is defined as 
$S(\gamma)={\rm sign}(I(\gamma)-E(\gamma))$.    
The length of $\gamma$ is defined as $L(\gamma)=|I(\gamma)-E(\gamma)|$.  
The interior, $\olg$, 
is defined to be the set of sites $x\in{\bf Z}^2$ enclosed by $\tga$.

Next, one defines the notion of compatibility
of two cylinders in such a way as to have a one--to--one correspondence
between the set of configurations $\phi_\Lambda$ and the set of all
compatible sets of cylinders.
We shall not use here the notion of weak compatibility considered 
in ref.~\cite{DM}.

Two cylinders $\gamma,\gamma'$ are 
{\it compatible}, written $\ga\sim\ga'$, if either condition (1) or condition (2), 
together with condition (3), hold: 

\smallskip

\begin{tabular}{ll}
(1)   
&$S(\gamma) = S(\gamma')$, ${\olg}\ne{\olg}'$ and either \\   
&${\olg}\cap{\olg}' = \emptyset$ and 
${\tga}\cap{\tga}' = \emptyset$, 
or ${\olg}\subset{\olg}'$, 
or ${\olg}'\subset{\olg}$, \\
\end{tabular} 

\smallskip

\begin{tabular}{ll}
(2)   
&$S(\gamma) = -S(\gamma')$, ${\olg}\ne{\olg}'$ and either \\
&${\olg}\cap{\olg}' = \emptyset$,  
or ${\olg}\subset{\olg}'$ and 
${\tga}\cap{\tga}' =\emptyset$, 
or ${\olg}'\subset{\olg}$ 
and ${\tga}\cap{\tga}' = \emptyset$.
\end{tabular}

\smallskip

\begin{tabular}{llll}
(3)   
&$E(\gamma) = E(\gamma')$ if ${\olg}\cap{\olg}'=\emptyset$,    
$E(\gamma) = I(\gamma')$ if ${\olg}\subset{\olg}'$, \\ 
&$I(\gamma) = E(\gamma')$ if ${\olg}'\subset{\olg}$, 
\end{tabular}\\

\noindent
where $\ti\ga\cap\ti\ga'=\emptyset$ is decided after rounding SW and
NE corners as in Figure~4. Examples of compatible cylinders are shown
on Figure~5.
\bigskip
\begin{center}
\setlength{\unitlength}{2000sp}%
\begin{picture}(7224,2724)(439,-2323)
\thinlines
{\put(1126,-1636){\oval(450,450)[tr]}
}%
{\put(1576,-1186){\oval(450,450)[bl]}
}%
{\put(6976,-1186){\oval(450,450)[bl]}
}%
{\put(6526,-1636){\oval(450,450)[tr]}
}%
{\put(1126,-1411){\line(-1, 0){675}}
\put(451,-1411){\line( 0, 1){1800}}
\put(451,389){\line( 1, 0){2700}}
\put(3151,389){\line( 0,-1){2700}}
\put(3151,-2311){\line(-1, 0){1800}}
\put(1351,-2311){\line( 0, 1){675}}
}%
{\put(1351,-1186){\line( 0, 1){675}}
\put(1351,-511){\line( 1, 0){900}}
\put(2251,-511){\line( 0,-1){900}}
\put(2251,-1411){\line(-1, 0){675}}
}%
{\put(6751,-1636){\line( 0,-1){675}}
\put(6751,-2311){\line(-1, 0){1800}}
\put(4951,-2311){\line( 0, 1){2700}}
\put(4951,389){\line( 1, 0){2700}}
\put(7651,389){\line( 0,-1){1800}}
\put(7651,-1411){\line(-1, 0){675}}
}%
{\put(6526,-1411){\line(-1, 0){675}}
\put(5851,-1411){\line( 0, 1){900}}
\put(5851,-511){\line( 1, 0){900}}
\put(6751,-511){\line( 0,-1){675}}
}%
{\put(5896,-1366){\framebox(810,810){}}
}%
\put(+700,-150){$\scr n+1$}
\put(+5200,-150){$\scr n+1$}
\put(+1500,-900){$\scr n-1$}
\put(+6000,-900){$\scr n-1$}
\put(+700,-1900){$\scr n$}
\put(+7200,-1900){$\scr n$}
\put(+1700,-1300){$\scr \ga_2$}
\put(+1700,-2200){$\scr \ga_1$}
\put(+6200,-1300){$\scr \ga_4$}
\put(+6200,-2200){$\scr \ga_3$}
\end{picture}%
\end{center}
\smallskip
\centerline{\footnotesize
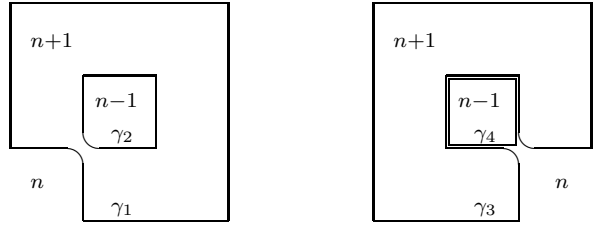
Figure 5. Compatible cylinders $\ga_1\sim\ga_2$ with
$E(\ga_2)=I(\ga_1)=n+1$ and $\bar\ga_2\subset\bar\ga_1$,  }
\centerline{\footnotesize and $\ga_3\sim\ga_4$ with
  $E(\ga_3)=E(\ga_4)=n$ and $\bar\ga_3\cap\bar\ga_4=\emptyset$   }

\bigskip

Two cylinders $\gamma',\gamma''$ are 
{\it separated\/} by a cylinder $\gamma$,
if ${\olg}'\ne{\olg}\ne{\olg}''$ and

\smallskip

\begin{tabular}{lll}
\hphantom{(1)} &either 
${\olg}'\subset{\olg}\subset{\olg}''$, 
&or ${\olg}''\subset{\olg}\subset{\olg}'$, \\
&or ${\olg}'\subset{\olg}$ 
and ${\olg}''\subset{\olg}^c$, 
&or ${\olg}''\subset{\olg}$ 
and ${\olg}'\subset{\olg}^c$. 
\end{tabular}

\smallskip

Let $\Gamma=\{\gamma_i\}$ be a set of cylinders. 
We say that this set $\Gamma$ 
is a {\it compatible} set of cylinders if
any two  of its cylinders not separated by a third one are compatible.
We denote by $\Gamma_{ext}$ the set of all {\it external} cylinders in 
$\Gamma$,
i.e., the set of all $\gamma$ such that $\olg$ is not contained in the
interior of any other cylinder in $\Gamma$.
We write $E(\Gamma)=n$ if $E(\gamma)=n$ for all $\gamma\in\Gamma_{ext}$.
The partition function with
constant boundary conditions can now be expressed %, apart from
%a trivial ground state energy term, 
as a sum over
compatible sets of cylinders with suitable weights.
To any cylinder $\gamma=({\tga},E,I)$ we assign the statistical weight
\beqa
\varphi(\gamma) = \varphi_{t,u}(\ga) &=&\exp\Big(-2\beta J L(\gamma)|{\tga}|
+u|{\olg}|
\big(\delta(I)-\delta(E)\big)\Big) \nonumber \\
&=&t^{{1\over2}L(\gamma)|{\tga}|}
\exp\Big(u|{\olg}|
\big(\delta(I)-\delta(E)\big)\Big) \label{phicyl}
\eeqa 
Then, taking formula (\ref{HH}) into account, we obtain  
\beq
\Xi(\Lambda,n)=e^{u\delta(n)|\Lambda|}
\sum_{\Gamma\in\mathcal{C}(\Lambda,n)}\prod_{\gamma\in\Gamma}\varphi(\gamma) , 
\label{Xi1}\eeq
where the sum runs over the set $\mathcal{C}(\Lambda,n)$ of all  
compatible sets of cylinders,
$\Gamma$, on $\Lambda$, such that $E(\Gamma)=n$. 

%%%%%%%%%%%%%%%%%%%%%%%%%%%%%%%%%%%%%

%\newpage
\section{Restricted ensembles}\label{RE}

We come back to the problem of phase transitions in the SOS model with
a wall. 
It has been recognized that for 
an interesting class of systems, including our model, 
one needs some extension of the Pirogov-Sinai theory of phase transitions, 
though a general theory of the concerned systems does not exist. 
In such an  extension certain states, called the {\it restricted ensembles}, 
play the role of the ground states in the usual theory. 
They can be defined as a Gibbs probability measure
on certain subsets of configurations. 

In the present case we shall consider, for each $n=0,1,2,\dots$,   
subsets of configurations which are in some
sense  near to the constant configurations $\phi_x \equiv n$. 
The precise definition is as follows.  
A cylinder $\gamma$ is called {\it elementary} if
\beq 
\diam{\tga}\le 3k+3 ,
\label{elementary}\eeq 
where $k$ is a given positive integer. 
Since we are dealing with integer numbers we shall use $\ell^1$ distance 
for this diameter, and the $\ell^1$ norm $|x| = |x_1| + |x_2|$ on ${\bf Z}^2$.
We use the notation $\C_k^{el}(\Lambda,n)$ for the set of finite 
compatible sets of elementary cylinders, that is, the set of all 
$\Gamma\in\mathcal{C}(\Lambda,n)$ that contain only elementary cylinders. 

The Gibbs measure defined on the subset $\C_k^{el}(\Lambda,n)$ 
is the {\it restricted ensemble} corresponding to level $n$.
The associated free energy (times $\beta$) per unit area is 
\beq
f_k(n) = -\lim_{\Lambda\to\infty} \frac{1}{|\Lambda|} \ln Z_k(\Lambda,n) ,
\label{limZeta}
\eeq
where
\begin{align}
Z_k(\Lambda,n) &=
\sum_{\phi_\Lambda\in\C_k^{el}(\Lambda,n)}
\exp(-\beta H_\La(\phi_\Lambda|n)) \notag \\
&= e^{u\delta(n)|\Lambda|}
\sum_{\Ga\in\C_k^{el}(\Lambda,n)}\prod_{\ga\in\Ga}\varphi(\gamma).
\label{Zeta}
\end{align}

Using the cluster expansion technique we shall prove,
in Section \ref{ProofProp}, that one is able 
to compute the free energies $f_k(n)$  
as convergent  power series in the variable $t$.  
The radius of convergence  
depends on the choice of the restricted set configurations 
$\C_k^{el}(\Lambda,n)$, that is on the value of $k$ used in its definition.
It happens, as will be seen in the proof,  
that $f_k(n)$ differs from the other free energies $f_k(n')$
at least by an order $t^{3n+3}$. 
Therefore, when we want to study the level $n$, 
we have to consider in the definition of $\C_k^{el}(\Lambda,n)$
all the %Dobrushin walls 
cylinders that can contribute with 
a weight at least of order $t^{3n+3}$. 
A consistent choice for this purpose will be to take $k=\max(2n,8)$. 

At this point one is able to study the phase diagram of the 
restricted ensembles. The restricted ensemble at level $n$
is said to be {\it dominant} 
for some given values of the parameters $u,t,k$, 
if $f_k(n)=\min_h f_k(h)$. 
In the next proposition we summarize the results
that will be proved  
concerning the regions in the plane $t,u$, where, for each $n$,
the associated restricted ensemble at level $n$ 
is the dominant restricted ensemble.

\begin{Prop} 
Let the integer $n\ge0$ be given and choose $k\ge 8$.
Let $a,b\ge0$ and $0<t\le t_1(k)=(3k+3)^{-4}$. If $n \geq 1$ and
\beq
-\ln(1-t^2)+(2+a)t^{n+3}\le u\le-\ln(1-t^2)+(2-b)t^{n+2} , 
\label{Ass}\eeq
or if $n=0$ and 
\beq
-\ln(1-t^2)+(2+a)t^3 \le u \le t^{1/2} , 
\label{Ass2} \eeq
then we have
\begin{align}
f_k&(n) \le f_k(h)-at^{3n+3}+O(t^{3n+4}), \notag \\
&\hbox{uniformly in } 
n \geq 0, k \geq \max(8,n), h\ge n+1, 
\label{Pr1} 
\end{align}
\begin{align} \label{Pr2} 
&f_k(n) \le f_k(n-1)-bt^{3n}+O(t^{3n+1}) \quad \text{if $n \geq 1$, and} \notag \\
&f_k(n) \le f_k(h) - 2t^{3h+3} + O(t^{3h+4}) \quad \text{for } h \leq n-2, \\
&\quad \hbox{both uniformly in } n \geq 1, k \geq \max(8,n), 0\le h\le n-1. \notag
\end{align}
\label{Prop1}\end{Prop}

``Uniformly'' in (\ref{Pr1}) means that $|O(t^{3n+4})|$ is bounded
from above by $t^{3n+4}$ times a constant
independent of $k,n,h$ in the given ranges, and similarly for (\ref{Pr2}).

%The definition of $t_2$ takes $t_1(k)/2$ instead of $t_1(k)$
%for later purposes outside Proposition \ref{Prop1}.

The bound of $t^{1/2}$ in \eqref{Ass2} is somewhat arbitrary, mainly present to ensure convergence of
the cluster expansion in the form that we use.

Proposition \ref{Prop1} is proved in Section \ref{ProofProp}.
Then the proof of Theorem \ref{Th1} will consist of showing that the 
phase diagram of the pure phases at low temperature is close
to the phase diagram of the dominant restricted ensembles.

%\newpage

\section{Proof of Proposition 3.1}\label{ProofProp}

Throughout the paper, $K_i$ are constants which do not depend on any
of the parameters in the problem; in particular they are independent of heights $h,n$.

In order to discuss the cluster expansion for the free energies $f_k(n)$
associated to the restricted ensembles, some definitions are needed.  
We first introduce the notion of {\it elementary perturbation}, defined  
as a compatible set of elementary cylinders $\omega\in\C_k^{el}(\Lambda,n)$
that contains a unique cylinder which is external in the set.

Thus, if $\omega$ is an elementary perturbation
we may write $\omega=(\gamma^{ext},\{\gamma_i\})$ as a compatible set of
elementary cylinders such that ${\olg}_i\subset{\olg^{ext}}$ 
for all $i$ and $E(\omega)=E(\gamma^{ext})=n$. 
The set $\olg^{ext}$ is called the {\it support} of
the elementary perturbation and is denoted by $\Supp\omega$.
To each elementary perturbation the following statistical weight is assigned: 
\beq
\varphi(\omega)=\varphi(\gamma^{ext})\prod_i\varphi(\gamma_i). 
\eeq
When computing these weights the following notation will be useful:
\beq
2\Vert\omega\Vert=L(\gamma^{ext})|{\tga}^{ext}|+
\sum_i L(\gamma_i)|{\tga}_i|, 
\eeq
so that $\Vert\omega\Vert$ represents half the number of vertical plaquettes.

Let now $\Gamma\in\C_k^{el}(\Lambda,n)$ be any compatible set of 
elementary cylinders. It is then possible to write
$\Gamma$ as the disjoint union of elementary perturbations 
\beq
\Gamma = \omega_1\cup\ldots\cup\omega_r
\label{decompo1}\eeq
in a unique way. 
This shows that there is a one-to-one correspondence between the 
restricted set of configurations $\C_k^{el}(\Lambda,n)$ and 
the set of all compatible sets of elementary perturbations, 
with the following definitions.
 
Two elementary perturbations $\om$ and $\om'$ are {\it compatible} if
their supports do not intersect 
and, moreover, $\om\cup\om'\in\C(\La,n)$. 
This last condition enters only when the boundaries of the 
supports of $\om$ and $\om'$ have a common part 
and is satisfied as soon as $S(\gamma^{ext})=-S(\gamma'^{ext})$.
A set of elementary perturbations is a {\it compatible set} if 
any two perturbations in the set are compatible.  The following is clear.

\begin{Lem}
The partition function (\ref{Zeta}) may be written as
\beq
Z_k(\Lambda,n)= e^{u\delta(n)|\Lambda|}
\sum_{\{\omega_i\}\in\C_k^{el}(\La,n)}\prod_{i}\varphi(\omega_i) ,
\label{Zeta3}\eeq
where the sum runs over all compatible sets of elementary perturbations
contained in $\Lambda$ such that $E(\omega)=n$.
\end{Lem}

This lemma and the compatibility definitions 
tell us that the restricted ensemble can be considered 
as a gas of elementary perturbations, in technical terms as a polymer system, 
for which a cluster expansion theory can be applied.
See, for instance, refs. \cite{GMM}, \cite{KP}, \cite{M}. 

A {\it cluster} $X$ from $\C_k^{el}(\La,n)$ can be defined as a finite
sequence of elementary  
perturbations $X=(\omega_1,\dots,\omega_r)$, where
some perturbations may be repeated, such that the set 
$\{\omega_1,\dots,\omega_r\}$
is connected by incompatibility relations $\om\not\sim\om'$.
(This notion of connection is formalized below.)
Supp $X$ will denote the union of the supports of the elementary
perturbations in $X$. 

Then there is a function $\varphi_u^{\rm T}(X)$, the {\it truncated function}, 
defined on the set of clusters 
(independent of $\Lambda$ and of the order of perturbations in $X$), 
such that one can write
\beq
-\ln Z_k(\Lambda,n) = -u\delta(n)|\Lambda|-\sum_{X:\Supp X\subset\Lambda}
\varphi_u^{\rm T}(X) \label{lnZ} 
\eeq
whenever the series converges, where
\beq
\varphi_u^{\rm T}(X) = a^{\rm T}(X)\prod_{i=1}^r \varphi(\omega_i) . \label{phiT}
\eeq
%Note that from the definition of $\varphi(\om_i)$, $\varphi^{\rm T}(X)$
%actually depends on the height $n$, though this is suppressed in the notation.
%The coefficients $a^{\rm T}$ depend only on the compatibility relations 
%between the elementary perturbations in $X$, and  
%for the computations below, it will be enough to know that $a^{\rm T}(X)=1$ 
%when the cluster $X$ has only one element.
and $a^{\rm T}(X)$ is a signed combinatoric factor, defined in
(\ref{aT}) below. 

More formally,
one considers a countable set $\mathcal{P}$; the elements of this set  
are the abstract polymers, 
the notation $\alpha\not\sim\alpha'$ means that the
two polymers $\{\alpha,\alpha'\}\subset\mathcal{P}$ are not compatible,
and the weight $\varphi(\alpha)$ is a complex valued
function on $\mathcal{P}$. 
To any finite sequence of polymers
$(\alpha_1,\dots,\alpha_r)$ there corresponds a function $X=X(\alpha)$ on
$\mathcal{P}$,  
with non-negative integer values, such that $\sum_\alpha X(\alpha)=r\ge1$.
$X(\alpha)$ is the multiplicity of $\alpha$ in the sequence.
Associated to the same $X$ we have $r!/X!$ ordered sequences, 
where $X!=\prod_\alpha X(\alpha)!$.
Given $X$ we construct the graph $g(X)$ with vertices ${1,\dots,r}$ and edges
$\{ i,j\}$ corresponding to pairs such that $\alpha_i\not\sim\alpha_j$ or
$\alpha_i=\alpha_j$. 
If $g(X)$ is a connected graph we say that $X$ is a cluster, and then 
\beq
a^{\rm T}(X)=(X!)^{-1}\sum_g (-1)^{|g|}\,,
\label{aT}\eeq
where the sum runs over all spanning connnected subgraphs of $g(X)$, and
$|g|$ is the number of edges of $g$.
If $r=1$, then $a^{\rm T}=1$.

Let us recall the following theorem on the convergence
of cluster expansions (see \cite{M}.)

\medskip

% \noindent{\bf Proposition: }
\noindent{\bf Convergence theorem: }
For a polymer system $\mathcal{P}$, assume that there is a positive function $\mu(\alpha)$, 
$\alpha\in\mathcal{P}$, such that for all $\alpha\in\mathcal{P}$, 
\beq
|\varphi(\alpha)|\le\mu(\alpha)
\exp\Big(-\sum_{\alpha':\,\alpha'\not\sim\alpha}\mu(\alpha')\Big).
\label{i1}\eeq
Then for all $\alpha\in\mathcal{P}$ we have
\beq
\sum_{X:\,\alpha\in X}
|\varphi_u^{\rm T}(X)|\le\mu(\alpha)
\label{i2}\eeq
and
\beq
\sum_X
X(\alpha)|\varphi_u^{\rm T}(X)| \le\vphi(\om)e^{\sum_{\om':\,\om'\not\sim\om}\mu(\om')} 
\le e^{\mu(\alpha)}-1.
\label{Xfactor}\eeq

\medskip\noindent
We now apply this theorem to the free energies of the 
restricted ensembles.  As a consequence of (\ref{limZeta}) and (\ref{lnZ})
we have the formal expansion
\beq
f_k(h)= -u\delta(h)-\sum_{X:\,\Supp X\ni 0}
\frac{1}{|\Supp X|} \varphi_u^{\rm T}(X).
\label{fk}\eeq

\begin{Lem}\label{fkLemma}%\label{fkLemma}
Let $k\ge8$, $t\le t_1(k)=(3k+3)^{-4}$, $u\le t^{1/2}$. Then the
expansion \eqref{fk} of the free energy is an absolutely convergent
power series in $t$. Moreover, with $s=te^{t^{1/4}}$ and
$\mu(\omega)=\varphi_{s,0}(\omega)$,
the following bounds are satisfied for all elementary perturbations $\om$:
\beq
\sum_{\omega':\,\omega'\not\sim\omega}\mu(\omega')
< 3000s^2 \cdot 9|\olg^{ext}_\om|
<  s^{1/2} |\olg^{ext}_\om|, \label{omegaA} 
\eeq
\beq
|\varphi(\om)|\le\mu(\om)
\exp\Big(-\sum_{\om':\,\om'\not\sim\om}\mu(\om')\Big),
\label{i1A}\eeq
\beq
\sum_{X:\,\om\in X}
|\varphi_u^{\rm T}(X)|\le\mu(\om),
\label{i2A}\eeq
\beq
\sum_X
X(\om)|\vphi_u^{\rm T}(X)|\le\vphi(\om)e^{\sum_{\om':\,\om'\not\sim\om}\mu(\om')}
\le e^{\mu(\om)}-1.
\label{i3A}\eeq
Here the sums are over elementary perturbations $\om'$ or clusters $X$, of an arbitrary fixed external height.
\end{Lem}

\begin{proof}
In our case, in the Convergence Theorem we can take $\mu(\omega)=\varphi_{s,0}(\omega)$,
the same as the weight $\varphi(\omega)=\varphi_{t,u}(\omega)$
on the set of elementary perturbations, but for some $s>t$ in place of
$t$, and $u$ replaced by 0.
Then if $\ga_\om^{ext}$ is the exterior cylinder of $\omega$, we have
\begin{align} \label{geom}
\sum_{\om: \ga_\om^{ext}=\ga} \mu(\om) 
  &<\mu(\ga)\Big(1+2\sum_{h=1}^\infty s^{h/2}\Big)^{4|\olg|} \\
&<\mu(\ga)(1-2s^{1/2})^{-4|\olg|} \notag \\ 
&<\mu(\ga)\exp(9s^{1/2}|\olg|),
\nonumber 
\end{align}
and then
\beqa \label{geom2}
&&\sum_{\om: \ti\ga_\om^{ext}=\ti\ga} \mu(\om) 
<{{2s^{|\tga|/2} }\over{ 1-s^{|\tga|/2}}}
\exp(9s^{1/2}|\olg|)
<3s^{|\tga|/2}\exp(9s^{1/2}|\olg|) .
%\nonumber 
\eeqa
We have first bounded the contributions of all other cylinders of $\omega$ by 
the contributions of all possible sets of vertical plaquettes that project 
on bonds of the dual lattice inside $\tga$ or on $\tga$; we take $4|\olg|$
as a bound for the number of such bonds.
We have used $(1-2s^{1/2})^{-1}<e^{9s^{1/2}/4}$, assuming $2s^{1/2}<\,0.21$. 
%The last inequality follows from $u\le 2s^2$, 
%according to inequality (\ref{Ass}) in Proposition \ref{Prop1} with $n\ge1$.

%We consider $h \geq 1$; the case $h=0$ is similar.  
A useful inequality is
\beq \label{Peierlssum}
  \sum_{l=m}^\infty l^2x^l \leq \frac{5}{4}m^2x^m \quad \text{for all } 
m \geq 4\,, x \leq \frac{1}{8}. 
  \eeq

Now $2\cdot3^{\ell-1}$ is a bound on the number of 
$\tga$ of length $\ell$ passing through a given dual lattice site.
Further, $\ga_\om^{ext}$, being an elementary cylinder, 
has interior containing fewer than $(3k+3)^2/4$ lattice sites, or also fewer than
$(\ell/4)^2$ lattice sites, so using \eqref{geom2} and \eqref{Peierlssum} we obtain for all $k \geq 1$:
\beqa
\sum_{\om: \olg_\om^{ext}\ni 0} \mu(\om) &<&
2\exp\Bigl(9s^{1/2} (3k+3)^2/4\Bigr)
\sum_{\ell=4}^\infty \left( \frac{\ell}{4} \right)^2(3s^{1/2})^\ell \nonumber\\ 
&<&\frac{5}{2} (3^4s^2)\exp\Bigl(9s^{1/2}  (3k+3)^2/4\Bigr)\nonumber\\ 
&<&3000s^2 \quad\hbox{if}\quad s\le2(3k+3)^{-4},
\label{8k}
\eeqa
where we used \eqref{Peierlssum} and $\exp(9s^{1/2}  (3k+3)^2/4) \leq \exp(9e^{1/54}/4)<10$.

For $s$ as above with $k \geq 2$, the argument of the 
exponential function in condition (\ref{i1}) is therefore
bounded as
\begin{align}
\sum_{\omega':\,\omega'\not\sim\omega}\mu(\omega')
  &< (|\olg^{ext}_\om|+|\ti\ga^{ext}_\om|+4)
  \sum_{\om':\,\Supp\om'\ni0}\mu(\om') \notag \\
&< 3000s^2 \cdot 9|\olg^{ext}_\om| \notag \\
&<  \frac{1}{20}s^{1/2} |\olg^{ext}_\om|, \label{omega} 
\end{align}
where the factor in front in the first inequality 
reflects the fact that given $\om$ there is a set of at most
$|\olg^{ext}_\om|+|\ti\ga^{ext}_\om|+4$
sites such that every $\om'$ incompatible with $\om$ must contain one of these sites in its support.
This proves \eqref{omegaA}.
%%Indeed the support of such an $\om'$ must contain at least one
%%point in the support of $\om$ or at distance one from this support.

By \eqref{omega}, in order to obtain \eqref{i1A}, which corresponds to (\ref{i1}), it suffices that
\beq \label{ts1}
  \varphi_{t,u}(\om) \leq \varphi_{s,0}(\om) e^{-s^{1/2} |\olg^{ext}_\om|}.
  \eeq
Now 
\begin{equation} \label{ts4}
  \frac{ \varphi_{t,u}(\om) }{ \varphi_{s,0}(\om) } 
    \leq \frac{ \varphi_{t,0}(\ga^{ext}_\om) }{ \varphi_{s,0}(\ga^{ext}_\om) } e^{u |\olg^{ext}_\om|}
    \leq \left( \frac{t}{s} \right)^{|\ti\ga^{ext}_\om|/2} e^{s^{1/2} |\olg^{ext}_\om|},
  \end{equation}
so (\ref{ts1}) reduces to
\beq \label{ts}
t^{{1\over2}|\ti\ga_\om^{ext}|} \leq s^{{1\over2}|\ti\ga_\om^{ext}|}
e^{-2s^{1/2} |\olg^{ext}_\om|}.
\eeq
Using $s\le2(3k+3)^{-4}$, $|\olg^{ext}_\om| \leq (3k+3)^2/4$ and the isoperimetric inequality 
$|\ti\ga_\om^{ext}| \geq 4|\olg_\om^{ext}|^{1/2}$, we see that
\beq\label{ts3}
  2s^{1/4} |\olg^{ext}_\om| \leq (3k+3)s^{1/4} |\olg^{ext}_\om|^{1/2} <  2^{-7/4} |\ti\ga_\om^{ext}|,
\eeq
so it suffices for \eqref{ts} that
\beq \label{s1}
  t \leq se^{-(s/8)^{1/4}}.
  \eeq
For \eqref{s1}, in turn, it suffices that 
\beq
  te^{t^{1/4}} \leq s \leq 2t.
\eeq
Thus for $k\ge2$, $t \leq (3k+3)^{-4}$ and $s(t) = te^{t^{1/4}}$, \eqref{i1A} is satisfied,
and \eqref{i2A}, \eqref{i3A} follow by the Convergence Theorem.  

We can use \eqref{i2A} to establish convergence of \eqref{fk}, as follows.  
Let $n(X) = |\{\om: X(\om) \geq 1\}|$ 
be the number of distinct elementary perturbations in the cluster $X$.  Then 
by \eqref{i2A} and \eqref{8k} we have
\begin{align} \label{newsum}
  \sum_{X:\Supp X\ni 0} \frac{1}{|\Supp X|} |\varphi_u^{\rm T}(X)|
    &= \sum_{\om:0\in \bar\ga_\om^{ext}}
    {1\over|\Supp\om|}\sum_{X\ni\om}{1\over n(X)} |\varphi_u^{\rm T}(X)| \notag \\
  &\leq \sum_{\om:0\in \bar\ga_\om^{ext}} {1\over|\Supp\om|} \mu(\om) \notag \\
  &< \infty.
  \end{align}
\end{proof}

%\newpage

Henceforth, dealing with $t \leq (3k+3)^{-4}$, 
we use the function $\mu(\om) = \vphi_{s,0}(\om)$ defined in the last proof,
with $s = s(t) \equiv te^{t^{1/4}}$, so 
\begin{equation} \label{sk}
  s \leq s_k \equiv 2(3k+3)^{-4}\leq 2/3^{12} \quad \text{for } k \geq 8.
  \end{equation}
From \eqref{i1A} we have $\varphi = \varphi_{t,u} \leq \mu$.  In view of \eqref{Pr1}
we observe that
\beq \label{svst}
  s^{3h+4} = O(t^{3h+4}) \quad \text{uniformly in $h \leq k, t \leq (3k+3)^{-4}$
    and } k \geq 8.
  \eeq
One may replace $3h+4$ here with $3h+r$ for any fixed $r$.

Next, we are going to compute the relevant terms in the expansion \eqref{lnZ}.

In Figure 6 are represented the clusters that
contribute to the difference $f_k(h+1)-f_k(h)$, 
%until the order $3h+3$, 
classified according to the number of contacts with the wall, with
leading terms singled out,
when $h\ge 2$. 
The picture is a schematic representation as  
the elementary perturbations are actually three-dimensional objects.
Either they touch the wall when placed on the interface at level $h$,
cases (a) to (d), or when placed on the interface at level 
$h+1$, cases (e) to i). 

The elementary perturbations in (a) and (e) are cylinders with one 
plaquette as base, with $\Vert\omega\Vert=2h$ in case (a)
and $\Vert\omega\Vert=2h+2$ in case (e).
Those in (c) and i) are cylinders with two plaquettes as base and we 
have $\Vert\omega\Vert=3h$ and $\Vert\omega\Vert=3h+3$, respectively.
The perturbations in (b) touch the wall only with one plaquette
and satisfy $\Vert\omega\Vert\ge2h+1$.
In case (d) they touch with two plaquettes and $\Vert\omega\Vert\ge3h+1$.
Finally, cases (f), (g) and (h) correspond to perturbations 
that touch the wall with one plaquette, obtained as a continuation
of those in cases (b), (c), and (d), and satisfy 
$\Vert\omega\Vert\ge2h+3$ in case (f),
$\Vert\omega\Vert=3h+2$ in case (g) and
$\Vert\omega\Vert\ge3h+3$ in case (h).

%\newpage

%%%%%%%%%%%%%%%%%%%%%%%%%%%%%%%%%%%%%%%%%%%%%%%%%%%%% FIG

\begin{center}
\setlength{\unitlength}{5mm}
\begin{picture}(20,7)(0,0)
\def\H{\line(1,0){1}}
\def\V{\line(0,1){1}}
\def\I{\line(1,1){0.5}}
{\thicklines
\put(-2,4){\line(1,0){22}}}
\put(21,3.8){\footnotesize level $h$}

\multiput(-1,0)(0,1){5}{\H}
\multiput(-1,0)(0,1){4}{\V}
\multiput(0,0)(0,1){4}{\V}
\multiput(3.5,0)(0,1){2}{\H}
\put(3.5,0){\V}
\put(4.5,0){\V}
\multiput(4.5,0)(0,1){2}{\I}
\bezier{300}(3.5,1)(2.7,2)(2.7,4)
\bezier{300}(4.5,1)(5.3,2)(5.3,4)
\multiput(8,0)(0,1){5}{\H}
\multiput(8,0)(0,1){4}{\V}
\multiput(9,0)(0,1){4}{\V}
\multiput(9,0)(0,1){5}{\H}
\multiput(10,0)(0,1){4}{\V}
%\multiput(10.5,0)(0,1){5}{\H}
%\multiput(10.5,0)(0,1){4}{\V}
%\multiput(11.5,0)(0,1){4}{\V}
%\multiput(11.5,0)(0,1){5}{\H}
%\multiput(12.5,0)(0,1){4}{\V}
%\multiput(9.5,3)(0,1){2}{\H}
%\put(9.5,3){\V}
\put(-1,-1.5){(a)}
\put(3.5,-1.5){(b)}
\put(8.5,-1.5){(c)}

\bezier{300}(14,1)(13.2,2)(13.2,4)
\bezier{300}(16,1)(16.8,2)(16.8,4)
\multiput(14,0)(0,1){2}{\H}
\multiput(14,0)(0,1){1}{\V}
\multiput(15,0)(0,1){1}{\V}
%\multiput(16,0)(0,1){1}{\V}
\put(16,0){\V}
\multiput(15,0)(0,1){2}{\H}
\put(14.5,-1.5){(d)}

\end{picture}
\end{center}

\begin{center}
\setlength{\unitlength}{5mm}
\begin{picture}(20,7)(0,0)
\def\H{\line(1,0){1}}
\def\V{\line(0,1){1}}
\def\I{\line(1,1){0.5}}
{\thicklines
\put(-2,4){\line(1,0){22.5}} }
\put(21,3.9){\footnotesize level }
\put(21,3.3){\footnotesize $h+1$}

\multiput(-1,0)(0,1){5}{\H}
\multiput(-1,0)(0,1){4}{\V}
\multiput(0,0)(0,1){4}{\V}
\multiput(3.5,0)(0,1){2}{\H}
\put(3.5,0){\V}
\put(4.5,0){\V}
\bezier{300}(3.5,1)(2.7,2)(2.7,4)
\bezier{300}(4.5,1)(5.3,2)(5.3,4)
\multiput(8,0)(0,1){5}{\H}
\multiput(8,0)(0,1){4}{\V}
\multiput(9,0)(0,1){4}{\V}
\multiput(8,4)(1,0){2}{\I}
\multiput(9,0)(0,1){5}{\H}
\multiput(10,0)(0,1){4}{\V}
%\multiput(10.5,0)(0,1){5}{\H}
%\multiput(10.5,0)(0,1){4}{\V}
%\multiput(11.5,0)(0,1){4}{\V}
%\multiput(11.5,0)(0,1){5}{\H}
%\multiput(12.5,0)(0,1){4}{\V}
%\multiput(9.5,3)(0,1){2}{\H}
%\put(9.5,3){\V}
\multiput(18,0)(0,1){5}{\H}
\multiput(18,0)(0,1){4}{\V}
\multiput(19,0)(0,1){4}{\V}
\multiput(19,0)(0,1){5}{\H}
\multiput(20,0)(0,1){4}{\V}
\put(-1,-2.5){(e)}
\put(3.5,-2.5){(f)}
\put(8.5,-2.5){(g)}
%\put(10.5,-2.5){(h)}
\put(18.5,-2.5){(i)}
\put(-1,-1){\H}
\put(3.5,-1){\H}
\put(8,-1){\H}
%\put(10.5,-1){\H}
\put(18,-1){\H}
\put(19,-1){\H}
\multiput(-1,-1)(1,0){2}{\V}
\multiput(3.5,-1)(1,0){2}{\V}
\multiput(8,-1)(1,0){2}{\V}
%\multiput(10.5,-1)(1,0){2}{\V}
\multiput(18,-1)(1,0){3}{\V}

\bezier{300}(14,1)(13.2,2)(13.2,4)
\bezier{300}(16,1)(16.8,2)(16.8,4)
\multiput(14,0)(0,1){2}{\H}
\multiput(14,0)(0,1){1}{\V}
\multiput(15,0)(0,1){1}{\V}
\put(16,0){\V}
\multiput(15,0)(0,1){2}{\H}
%\put(14,-1.5){(d')}
\multiput(14,-1)(0,1){1}{\H}
\multiput(14,-1)(0,1){1}{\V}
\multiput(15,-1)(0,1){1}{\V}
\put(14.5,-2.5){(h)}

\end{picture}
\vskip1.8cm
\begin{minipage}{13cm}\centerline
{\footnotesize
Figure 6. Clusters as in (\ref{Fig3A})
%up to order   $t^{3n+3}$. Case (d) also includes at the same order a double column hanging from a 2x2 square, and case (h) similarly. Cases (d') and (h')
%have at least 4 cubes at the top level, not in the shape of a 2x2 square.
}
\end{minipage}
\end{center}
\vskip2.5cm
%%%%%%%%%%%%%%%%%%%%%%%%%%%%%%%%%%%%%%%%%%%%%%%%% END FIG
%\newpage

From these considerations and \eqref{fk} we obtain, for $k\ge 1, h \geq 4$,
\begin{align}
f_k(h&+1)-f_k(h) \nonumber\\ 
&=At^{2h}e^u+P_h(t)e^u+Ct^{3h}e^{2u}+Q_h(t)e^{2u} \nonumber\\
&\quad-At^{2h}-P_h(t)-Ct^{3h}-Q_h(t) \nonumber\\
&\quad-Et^{2h+2}e^u-P_h(t)t^2e^u-Gt^{3h+2}e^u
-2t^2Q_h(t)e^u-It^{3h+3}e^{2u}\nonumber\\
&\quad+V_h(t,u),\label{Fig3A}
\end{align}
where $V_h(t,u)$ contains only terms of order $t^{3h+4}$, or smaller, in $t$.  
In the right side of \eqref{Fig3A}, the contributions to 
$f_k(h)$ from the relevant clusters are in the first line, and the contributions to $f_k(h+1)$ are
in the second and third lines. 
%the dots at the end stand for the terms of higher order in $t$.
The coefficients $A,C,E,$ etc. are the number of elementary 
perturbations per site that belong to cases (a), (c), (e), etc., so that
\beq
A=E=1,\  C=I=2,\  G=2C.
\label{A}
\eeq
The terms containing $P_h(t)$ correspond to the contributions
coming from clusters of types (b) and (f). $P_h(t)$ is a convergent series
\beq \label{Phdef}
P_h(t)=B_1t^{2h+1}+B_2t^{2h+2}+\dots+B_{h+3}t^{3h+3}+\dots
\eeq
%for $h \geq 1$, 
where $B_j$, $j=1,,\dots$, is the number, at each site, of 
the perturbations in case (b) such that $\Vert\omega\Vert=2h+j$,
plus the number (times their coefficient $a^T(\cdot)$) of clusters of height $h$ of
order $t^{2h+j}$.  In particular, the leading coefficient $B_1=4$ for $h \geq 2$, and $B_1=0$ for $h=1$.
The terms containing $Q_h(t)$ correspond to the contributions coming from cases
(d) and (h), and analogously
\beq \label{Qhdef}
Q_h(t) = D_1t^{3h+1} + D_2t^{3h+2} + D_3t^{3h+3}+\dots,
\eeq
with leading coefficient $D_1=16$ for $h \geq 2$ and $D_1=0$ for $h=1$.
Here $D_1=16$ comes from perturbations with 2 cubes in the top and bottom layers,
with each of the cubes attached to any of the 4 sides of the main $1 \times 1$ column.
The dependence of $P_h, Q_h, V_h$ on $k$ is suppressed in the notation.
The error term $V_h(t,u)$ excludes those terms of order $t^{3h+4}$ or smaller
which are already accounted for in $P_h(t)$ or $Q_h(t)$.
Regrouping terms in \eqref{Fig3A} we obtain for $h \geq 4$:
\begin{align}
f_k&(h+1)-f_k(h) \nonumber\\ 
&\quad=(t^{2h}+P_h(t))\big(e^u-1-t^2e^u\big) + (2t^{3h} + Q_h(t))\big(e^{2u}-1-2t^2e^u\big) \notag \\
&\quad \qquad -2t^{3h+3}e^{2u} + V_h(t,u). 
\label{PP}
\end{align}

We also need the analog of (\ref{PP}) for $1 \leq h \leq 3$, which is
similar but includes additonal terms. For $j = 3h+2,3h+3$ and $m \geq 2$ let $L_{jm}(h)$ be the number 
(per site, and incorporating the combinatorial factors $a^T(X)$) 
of clusters of external level $h$ incorporating $m$ horizontal plaquettes touching the wall and $2j$ 
vertical plaquettes, which are not counted in the top row of Figure 6.  
From the isoperimetric inequality, $L_{jm}(h)=0$ for $m>j^2/4h^2$.  Define
\[
  R_h(t,u) = \sum_{j=3h+2}^{3h+3}\ \sum_{2 \leq m \leq j^2/4h^2} L_{jm}(h) t^j (e^{mu} - 1).
  \]
Then for $h = 2,3$,
\begin{align} \label{Fig3A23}
f_k&(h+1) - f_k(h) \notag \\
&= (t^{2h} + P_h(t))(e^u-1-t^2e^u)+ (2t^{3h} + Q_h(t))(e^{2u}-1-2t^2e^u)\notag \\
&\qquad -2t^{3h+3}e^{2u} +R_h(t,u) + V_h(t,u).\end{align}
For $h=1$ we have additional terms with $j=3h+1=4$:
\begin{align} \label{Fig3A1}
f_k(2)-f_k(1)&= (t^2 + P_1(t))(e^u-1-t^2e^u)+ (2t^3 + Q_1(t))(e^{2u}-1-2t^2e^u)\notag \\
&\quad + L_{43}(1)t^4(e^{3u}-1-3t^2e^u)+ L_{44}(1)t^4(e^{4u}-1-4t^2e^u) \notag \\
&\quad + L_{42}(1)t^4(e^{2u} - 1 - 2t^2e^u) -2t^6e^{2u} +R_1(t,u) + V_1(t,u).
\end{align}
Note that $L_{43}(1)=6$ is the number (per site) of downward perturbations of height 1
consisting of 3 cubes, adjacent via common faces, and $L_{44}(1)=1$ is the 
number (per site) consisting of a $2\times 2$ 
block of cubes.  $L_{42}(1)=-5/2$ incorporates contributions from (i) single perturbations which consist of a pair of downward cubes with bases touching by a northeast-southwest corner, and (ii) clusters consisting of two single-downward-cube perturbations which may coincide, or may be adjacent with bases touching by either an edge or a northeast-southwest corner.

Finally, for $h=0$ we have
\beq 
f_k(1)-f_k(0)= u-t^2(1+e^{u}-e^{-u})-2t^3(1+e^{2u}-e^{-2u})+ V_0(t,u).
\label{10}\eeq 
The elementary perturbations that contribute to this difference outside of $V_0(t,u)$ are only 
the cylinders of type (a) and (c) in Figure~6 with height equal to 1, 
i.e., having as base one or two plaquettes, and $\Vert\omega\Vert=2$ or $3$,
respectively. 
However, one has now to consider the upward as well as the downward
cylinders. 

In order that a given $n \geq 4$ be the optimum interface height, $u$ should be chosen so that \eqref{PP} is negative for $h<n$ and positive for $h \geq n$.  At the crossover point at which the right side of \eqref{PP} is 0, at least to order $t^{3h+3}$, it is easily seen that $e^u - 1 - t^2e^u$ is approximately $2t^{h+3}$, so that the term $t^{2h}(e^u - 1- t^2 e^u)$ from perturbations of types (a) and (e) in Figure 6 cancels the term $2t^{3h+3}e^{2u}$ from perturbations of type i); all other terms are then of smaller order.  To make the given $n$ be optimal, therefore, ignoring error terms one should have 
\begin{equation}
  e^u - 1 - t^2e^u \begin{cases} < 2t^{h+3} \quad &\text{for all } h<n,\\ \geq 2t^{h+3} \quad &\text{for all } h \geq n, \end{cases}
  \end{equation}
which (after allowing room for error terms) yields the interval of $u$ values given by \eqref{Ass1}.  The essential cancellation of type (a), (e) and i) terms here contrasts with the external-field case in \cite{DM}, where the balancing is between type (e) and the external field, with all other perturbation types contributing only to smaller order.

The major difficulty in making this idea rigorous is that to establish an 
optimal $n$ we need control of error terms uniformly in $h$.  Thus
we need bounds on $P_h$, $Q_h$ and $V_h$ which are uniform in $h,k$.
This requires a preliminary lemma.  
Given a cluster $X$, an \emph{incompatibility path} in $X$ is a finite sequence 
$(\alpha_0,\dots,\alpha_m)$ of elementary perturbations in $X$ satisfying
$\alpha_{i-1} \not\sim \alpha_i$ or $\alpha_{i-1}=\alpha_i$
for all $1 \leq i \leq m$; such a sequence may be viewed as
a path in the incompatibility graph $g(X)$.  The \emph{length} of the path is $m$; we say the path
is \emph{minimal} if there is no strictly shorter path from $\alpha_0$ to $\alpha_m$ in $g(X)$.  
Given a cluster $X$ containing a designated $\alpha_0$ which we call the \emph{root},
we say a perturbation $\om$ in $X$ is \emph{beyond} another perturbation $\alpha$ in $X$
(relative to the root $\alpha_0$) if 
$\om \notin \{\alpha_0,\alpha\}$ and every path from $\alpha_0$ to $\om$
in $g(X)$ contains $\alpha$.  
The \emph{outer leaf} in $X$ of a perturbation $\alpha \in X$ is 
\[
  \mL_{out}(\alpha) = \{\om \in X: \om \text{ is beyond } \alpha \},
 \]
the \emph{leaf} is 
\[
  \mL(\alpha) = \mL_{out}(\alpha) \cup \{\alpha\},
  \]
and the \emph{stem} of $\alpha$ in $X$ is
\[
  \mS(\alpha) = X \bs \mL_{out}(\alpha).
  \]
Note $\{\alpha\} =  \mL(\alpha) \cap \mS(\alpha)$, and the leaf and stem are both clusters.
The multiplicity of $\alpha$ in $\mS(\alpha)$ is then $X(\alpha)$, and the multiplicity of $\alpha$ in $\mL(\alpha)$
is 1.  The leaf $\mL(\alpha)$ may be just $\{\alpha\}$, meaning the outer leaf is empty; this is always the case if 
$X(\alpha) \geq 2$.
For $X$ containing a minimal incompatibility path 
$(\alpha_0,\dots,\alpha_m)$,
we can then describe the structure of $g(X)$ as follows.
For $1 \leq i\leq m-1$ we say $\alpha_i$ is \emph{critical}
(for $\alpha_0 \to \alpha_m$) if $\alpha_m$ is 
beyond $\alpha_i$.  Let $i_1 < \dots < i_l$ be the indices of the critical perturbations $\alpha_i$.  
The 0th \emph{bead} consists of those $\om \in X$ which are not beyond $\alpha_{i_1}$, the 
$l$th \emph{bead} consists of $\alpha_{i_l}$ and those $\om \in X$ which are beyond $\alpha_{i_l}$, and for
$1 \leq j \leq l-1$, the $j$th \emph{bead} consists of $\alpha_{i_j}$ and those $\om \in X$ which are beyond 
$\alpha_{i_j}$ but not beyond $\alpha_{i_{j+1}}$.  Note that the intersection of the $(j-1)$st and $j$th 
beads is $\{\alpha_{i_j}\}$.

Let $\mG(\alpha_0,\dots,\alpha_m)$
denote the set of clusters $X$ for which $(\alpha_0,\dots,\alpha_m)$ is a minimal 
incompatibility path in $g(X)$.  

\begin{Lem} \label{pertcost}
(i) Let $m \geq 0$ and let $\alpha_0,\dots,\alpha_m$ be perturbations with the same external height $h \geq 0$.
Then for all $u,t$ as in Lemma \ref{fkLemma}
and $\mu=\vphi_{s,0}$ with $s=s(t)=te^{t^{1/4}}$,
\beq \label{factor}
  \sum_{X \in \mG(\alpha_0,\dots,\alpha_m)} \left| \vphi_u^{\rm T}(X) \right| \leq 3^{3m+1}\mu(\alpha_0)
    \prod_{i=1}^m \vphi(\alpha_i)
\eeq
and
\beq \label{factorexp}
  \sum_{X \in \mG(\alpha_0,\dots,\alpha_m)} X(\alpha_0) \left| \vphi_u^{\rm T}(X) \right| \leq 3^{3m+2}\mu(\alpha_0)
    \prod_{i=1}^m \vphi(\alpha_i).
\eeq

(ii) For $m=1$ the path assumption $X \in \mG(\alpha_0,\alpha_1)$ in \eqref{factor} can be removed:  for any two 
distinct elementary perturbations 
$\alpha_0,\alpha_1$ with the same external height $h \geq 0$,
\beq \label{factor2}
  \sum_{X:\alpha_0,\alpha_1 \in X} \left| \vphi_u^{\rm T}(X) \right| \leq 3^4\mu(\alpha_0) \vphi(\alpha_1)
\eeq
and
\beq \label{factor3}
  \sum_{X:\alpha_0,\alpha_1 \in X} X(\alpha_0) \left| \vphi_u^{\rm T}(X) \right| \leq 3^5\mu(\alpha_0) \vphi(\alpha_1).
\eeq

(iii) For each elementary perturbation $\alpha$,
\beq \label{factor4}
  \sum_{X:X(\alpha)\geq 2} \left| \vphi_u^{\rm T}(X) \right| \leq 3^4\mu(\alpha)\varphi(\alpha)
\eeq
and
\beq \label{factorexp2}
  \sum_{X:X(\alpha)\geq 2} X(\alpha) \left| \vphi_u^{\rm T}(X) \right| \leq 2\cdot 3^5 \mu(\alpha)\varphi(\alpha)
\eeq
\end{Lem}
\begin{proof}
(i) For $m=0$ this is a consequence of the Convergence Theorem, so consider $m \geq 1$.
Let $X$ be a cluster in which $(\alpha_0,\dots,\alpha_m)$ is a minimal incompatibility path.
Inductively we first define $X_m,T_m$ to be the leaf and stem, respectively,
of $\alpha_m$ in $X$, with $\alpha_0$ as root, then define $X_{m-1}, T_{m-1}$ to
be the leaf and stem of $\alpha_{m-1}$ in $T_m$, continuing this way until $X_1,T_1$ 
are the leaf and stem of $\alpha_1$ in $T_2$.  Then we define $X_0=T_1$.

An alternative description is as follows.  If $\alpha_i$ is non-critical 
for $\alpha_0 \to \alpha_m$ and is part of the $j$th bead, then 
the leaf $X_i$ is part of the $j$th bead as well, and the stem $T_i$ consists of beads 0 through $j$, 
with the outer leaves of $\alpha_i,\alpha_{i+1},\dots$ removed from the $j$th bead.  If $\alpha_i$ is
critical, then (i) $X(\alpha_i)=1$, (ii) $\alpha_i$ is the 
intersection of the $(j-1)$st and $j$th beads for some $1 \leq j \leq m$, 
(iii) $X_i$ consists of the $j$th bead with the outer leaves of $\alpha_{i+1},\alpha_{i+2},\dots$ removed, 
and (iv) the stem $T_i$ consists of beads 0 through $j-1$.

Thus each bond in $g(X)$ is in exactly one graph $g(X_i)$, and each spanning graph of $X$
is uniquely obtained as a union of spanning graphs of each cluster $X_i$.  
Further, from the bead description,
each $\alpha_i, i \geq 1,$ appears with multiplicity 1 in $X_i$, and with multiplicity $X(\alpha_i)$ in one 
$X_j$ with $j<i$; in other $X_l$ it appears with multiplicity 0.  It follows that
\[
  a^T(X) = \prod_{i=0}^m a^T(X_i),
  \]
\beq \label{Xksum}
  \sum_{j=0}^m X_j(\alpha_i) = X(\alpha_i) + 1, \quad i=0,\dots,k,
  \eeq
and
\beq \label{split}
  \left| \vphi_u^{\rm T}(X) \right| = \frac{ \prod_{i=0}^m \left| \vphi_u^{\rm T}(X_i) \right| }
    { \prod_{i=1}^m \vphi(\alpha_i) }.
\eeq
We now define modifications of the weights $\vphi$ and $\mu$ as follows. 
Let $\mI_4(\alpha_0,\dots,\alpha_m)$ denote
the set of perturbations which are incompatible with at least 4 of the perturbations $\alpha_i$, let
\[
  \hatp(\om) = \begin{cases} 0, &\text{if } \om \in \mI_4(\alpha_1,\dots,\alpha_m), \\
    \frac{1}{9}, &\text{if } \om \in \{\alpha_1,\dots,\alpha_m\},\\
    \vphi(\om), &\text{otherwise}, \end{cases}
\]
\[
  \hatm(\om) = \begin{cases} 0, &\text{if } \om \in \mI_4(\alpha_1,\dots,\alpha_m), \\
    \frac{1}{3}, &\text{if } \om \in \{\alpha_1,\dots,\alpha_m\},\\
    3\mu(\om), &\text{otherwise}, \end{cases}
\]
and let $\hatp_u^{\rm T}(\cdot)$ be the corresponding weight for clusters, as in \eqref{phiT}.
From the definition of minimal path, for each $X \in \mG(\alpha_0,\dots,\alpha_m)$
and $\om \in X$ there are at most 
3 values $i$ for which $\om \not\sim \alpha_i$, so $X \cap \mI_4(\alpha_0,\dots,\alpha_m) = \phi$.  
Hence for every perturbation $\om$, by \eqref{omega},
\begin{align} \label{CTbound}
  \sum_{\om' \not\sim \om} \hatm(\om') &\leq 3 \cdot \frac{1}{3} + 3\sum_{\om' \not\sim \om} \mu(\om')
    \leq 1 + \frac{1}{4}s^{1/2} |\olg^{ext}_\om|.
\end{align}
Since $s\le2(3k+3)^{-4}$, $|\olg^{ext}_\om| \leq (3k+3)^2/4$ and $k \geq 8$, we see that 
the right side of \eqref{CTbound} is bounded by $\log 3$.
It is then easily checked that the condition \eqref{i1} is satisfied in the Convergence Theorem
for $\hatp$ and $\hatm$.  Therefore using \eqref{Xksum}, \eqref{split} 
and the Convergence Theorem we have
\begin{align} \label{hatnohat}
  \sum_{X \in \mG(\alpha_0,\dots,\alpha_m)} &X(\alpha_0) \left| \vphi_u^{\rm T}(X) \right| \notag \\
  &= \sum_{X \in \mG(\alpha_0,\dots,\alpha_m)} 
    X(\alpha_0) \frac{ \prod_{i=0}^m \left| \vphi_u^{\rm T}(X_i) \right| }
    { \prod_{i=1}^m \vphi(\alpha_i) } \notag \\
  &= \sum_{X \in \mG(\alpha_0,\dots,\alpha_m)}  X(\alpha_0) \left| \hatp_u^{\rm T}(X_0) \right| 
    \frac{ \prod_{i=1}^m (9\vphi(\alpha_i))^{X(\alpha_i)+1} \left| \hatp_u^{\rm T}(X_i) \right| }
    { \prod_{i=1}^m \vphi(\alpha_i) } \notag \\
  &\leq  \left( \prod_{i=1}^m 81\vphi(\alpha_i) \right) \sum_{X \in \mG(\alpha_0,\dots,\alpha_m)}
     X(\alpha_0) \prod_{i=0}^m  \left| \hatp_u^{\rm T}(X_i) \right| \notag \\
   &\leq \left( \prod_{i=1}^m 81\vphi(\alpha_i) \right) 
     \left( \sum_{X:X \ni \alpha_0} X(\alpha_0) \left| \hatp_u^{\rm T}(X) \right| \right)
     \prod_{i=1}^m \left( \sum_{X:X \ni \alpha_i} \left| \hatp_u^{\rm T}(X) \right| \right) \notag \\
   &\leq 81^m \left( e^{\hatm(\alpha_0)}-1 \right) \prod_{i=1}^m \vphi(\alpha_i)\hatm(\alpha_i) \notag \\
   &\leq 3^{3m+2} \mu(\alpha_0) \prod_{i=1}^m \vphi(\alpha_i).
\end{align}
This proves \eqref{factorexp}.  The proof of \eqref{factor} is essentially the same, 
with the $X(\alpha_0)$ factors removed.

(ii)  In part (i), the path assumption was used only to create the bead description, and to ensure 
that for relevant clusters $X$ and $\om \in X$, there are at most 3 values $i$ for which 
$\om \not\sim \alpha_i$.  Neither of these considerations is needed for $m=1$ so the same
proof applies.

(iii) We wish to use (ii).  We define a new polymer system for which the set $\mP^*$ of polymers
consists of the set $\mP_{el}$ of elementary perturbations and 
one additional polymer $\alpha^*$.  This $\alpha^*$ is a ``copy of $\alpha$'' in the sense that we define 
the weight $\varphi^*$ on $\mP$ by $\varphi^*=\varphi$ on $\mP_{el}$ 
and $\varphi^*(\alpha^*) = \varphi(\alpha)$, 
and define $\alpha^*$ to be compatible with the same 
elementary perturbations as $\alpha$.  $\mu^*$ is defined analogously, and
the corresponding truncated function $(\varphi^*)_u^{\rm T}$ given by the analog of 
\eqref{phiT}.
There is a natural projection $Q$ from clusters in $\mP^*$ to 
clusters in $\mP_{el}$ defined by replacing each copy of 
$\alpha^*$ with a copy of $\alpha$, and for $Y$ a cluster in $\mP^*$ we have
\[
  (\varphi^*)_u^{\rm T}(Y) = {Y(\alpha)+Y(\alpha^*) \choose Y(\alpha)} \varphi_u^{\rm T}(Q(Y)),
\]
because the sum over graphs in \eqref{aT} is the same in $Y$ as in $Q(Y)$.
Therefore for each cluster $X$ in $\mP_{el}$ with $X(\alpha) \geq 2$,
\[
  \sum_{Y:Q(Y)=X \atop {\alpha,\alpha^* \in Y}} |(\varphi^*)_u^{\rm T}(Y)| = (2^{X(\alpha)}-2) |\varphi_u^{\rm T}(X)|
    \geq  |\varphi_u^{\rm T}(X)|.
  \]
We would like to conclude from (ii) (with $\alpha_0=\alpha^*,\alpha_1=\alpha$) that
\[
  \sum_{X:X(\alpha)\geq 2} \left| \vphi_u^{\rm T}(X) \right| \leq
    \sum_{Y:\alpha,\alpha^* \in Y} |(\varphi^*)_u^{\rm T}(Y)|
    \leq 3^4 \mu^*(\alpha^*)\varphi^*(\alpha).
  \]
Since the proof of (ii) uses the Convergence Theorem, we need to know that \eqref{i2} remains valid 
for $\hat\varphi$ and $\hat\mu$ when one more term
$\hat\mu(\alpha^*) = 3\mu(\alpha)$ is added to the sum there.  But this follows from the fact that 
\eqref{CTbound} remains valid with the extra term.  Thus we can indeed apply (ii), proving \eqref{factor4}.
The proof of \eqref{factorexp2} is similar, using the fact that $Y(\alpha) + Y(\alpha^*) = X(\alpha)$ when 
$Q(Y)=X$.
\end{proof}

Lemma \ref{pertcost}(i) can be used to help control the contribution to $V_h(t,u)$ from 
clusters of large diameter (at least $16h$.)  We say that an elementary perturbation
$\om$ is \emph{simple} if $\om$ consists of a single downward cylinder $\ga$ with $L(\ga)=1$.
We say a cluster $X$ is \emph{simple} if every elementary perturbation in $X$ is simple.
For $A \subset \ZZ^2$ we define the site boundary $\p_sA = \{y \in A: d(y,A^c)=1\}$.

\begin{Lem} \label{chains}
Let $u,t$ be as in Lemma \ref{fkLemma}, let $s=s(t)=te^{t^{1/4}}$ and let $h \geq 0$. 

(i) For all $0 \neq x \in \ZZ^2$, 
\beq \label{linkcost}
  \sum_{X:0,x \in \Supp(X)} \left| \vphi_u^{\rm T}(X) \right| \leq 10(180s)^{|x|+2},
  \eeq
where the sum is over clusters from $\C_k^{el}(\ZZ^2,h)$.  

(ii) For $\zeta$ an elementary perturbation and $x \in \ZZ^2$ with $d(\Supp \zeta,x) = r \geq 1$,
\beq \label{linkcost1}
  \sum_{X:\zeta\in X,\atop{x \in \Supp(X)}} \left| \vphi_u^{\rm T}(X) \right| \leq 3^9r\mu(\zeta)(180s)^{r+1},
  \eeq
and
\beq \label{linkcost1exp}
  \sum_{X:\zeta\in X,\atop{x \in \Supp(X)}} X(\zeta) \left| \vphi_u^{\rm T}(X) \right| \leq 3^{10}r\mu(\zeta)(180s)^{r+1}.
  \eeq
\end{Lem}
\begin{proof}
We consider (i) first.
We have $0,x \in \Supp(X)$ if and only if there exists a minimal incompatibility path 
$(\alpha_0,\dots,\alpha_m)$ in $g(X)$ with $0 \in \olg_{\alpha_0}^{ext}, x \in \olg_{\alpha_m}^{ext}$.
Given $X$, let $U_m(X)$ denote the set of all $x$ for which such a minimal path exists,
with length $m \geq 0$.  Then
\beq \label{msum}
  \sum_{X:0,x \in \Supp(X)} \left| \vphi_u^{\rm T}(X) \right| = 
    \sum_{m=0}^\infty \sum_{X:x \in U_m(X)} \left| \vphi_u^{\rm T}(X) \right|,
  \eeq
and from Lemma \ref{pertcost}(i) and a slight modification of \eqref{geom2} we have for $m \geq 0$,
\begin{align} \label{simple}
  \sum_{X:x \in U_m(X)} \left| \vphi_u^{\rm T}(X) \right| &\leq 
    \sum_{{(\alpha_0,\dots,\alpha_m)} \atop {0 \in \olg_{\alpha_0}^{ext},x \in \olg_{\alpha_m}^{ext}}}
    \sum_{X \in \mG(\alpha_0,\dots,\alpha_m)} \left| \vphi_u^{\rm T}(X) \right| \notag \\
  &\leq \sum_{{(\alpha_0,\dots,\alpha_m)} \atop {0 \in \olg_{\alpha_0}^{ext},x \in \olg_{\alpha_m}^{ext}}}
    3^{3m+1} \mu(\alpha_0) \prod_{i=1}^m \vphi(\alpha_i) \notag \\
  &\leq 3^{3m+1} \sum_{{{(\alpha_0,\dots,\alpha_m)} \atop {0 \in \olg_{\alpha_0}^{ext},x \in \olg_{\alpha_m}^{ext}}} 
    \atop {\text{all $\alpha_i$ simple}}} \prod_{i=0}^m 
    \left( \sum_{ \om: \ti\ga_ \om^{ext} = \ti\ga_{\alpha_i}^{ext} }\mu( \om) \right) \notag \\
  &\leq 3^{4m+2} \sum_{{{(\alpha_0,\dots,\alpha_m)} \atop {0 \in \olg_{\alpha_0}^{ext},x \in \olg_{\alpha_m}^{ext}}} 
    \atop {\text{all $\alpha_i$ simple}}} \prod_{i=0}^m \mu(\alpha_i).
\end{align}
In these sums, $(\alpha_0,\dots,\alpha_m)$ represents a sequence of perturbations which
form a minimal path in some cluster $X$; when such an $X$ exists, one such $X$ 
consists of the perturbations $\alpha_0,\dots,\alpha_m$ with multiplicity 1 each.
In the third line of \eqref{simple} we identify a simple $\alpha_i$ and its 
unique cylinder, in a mild abuse of notation.
We will show by induction on $m$ that
\beq \label{hypoth}
  \sum_{{{(\alpha_0,\dots,\alpha_m)} \atop {0 \in \olg_{\alpha_0}^{ext},x \in \olg_{\alpha_m}^{ext}}} 
    \atop {\text{all $\alpha_i$ simple}}} \prod_{i=0}^m \mu(\alpha_i) \leq (180s)^{|x|+m+2}.
\eeq
For $m=0$ this is a simple Peierls-type bound: a cylinder $\alpha$ with $0,x \in \ola$ 
must have 
\[
  |\tilde{\alpha}| \geq (2|x|+4) \vee 4|\ola|^{1/2},
  \]
and there are at most $3^l$ base perimeters $\ti\ga$ of length $l$ through a given site, so
using \eqref{Peierlssum},
\begin{align} \label{PeierlsA}
  \sum_{{\alpha \text{ simple}} \atop {0,x \in \ola}} \mu(\alpha) &\leq
    \sum_{l \geq 2|x|+4} \left( \frac{l}{4} \right)^2 3^l s^{l/2}
    \leq \frac{5}{64} (2|x|+4)^2 (3s^{1/2})^{2|x|+4}
    \leq (18s)^{|x|+2}.
\end{align}
Now suppose \eqref{hypoth} holds for $m=0,\dots,j-1$ for some $j \geq 1$.  
For $x \in \ZZ^2$ let $R(x)$ be the (possibly degenerate) rectangle with opposite corners 0 and $x$; then for all $y$,
\[
  |x-y| = |x| - |y| + 2\dist(y,R(x)),
  \]
where dist denotes $\ell^1$ distance.  Given 
$(\alpha_0,\dots,\alpha_j)$ with $0 \in \olg_{\alpha_0}^{ext}, x \in \olg_{\alpha_j}^{ext}$, 
there must exist $y,y' \in \ZZ^2$,
either equal or adjacent, with $y \in \olg_{\alpha_{j-1}}^{ext}, y' \in \olg_{\alpha_j}^{ext}$.  
Therefore using \eqref{PeierlsA},
\begin{align} \label{increase}
  \sum_{{{(\alpha_0,\dots,\alpha_j)} \atop {0 \in \olg_{\alpha_0}^{ext},x \in \olg_{\alpha_j}^{ext}}} 
    \atop {\text{all $\alpha_i$ simple}}} \prod_{i=0}^j \mu(\alpha_i)
    &\leq \sum_{y \in \ZZ^2}\ \sum_{{y' \in \ZZ^2} \atop {|y' - y| \leq 1}}
     \sum_{{ {(\alpha_0,\dots,\alpha_{j-1})} \atop {0 \in \olg_{\alpha_0}^{ext},x \in \olg_{\alpha_{j-1}}^{ext}} } 
     \atop {\text{all $\alpha_i$ simple}}}
     \sum_{{\alpha_j \text{ simple}} \atop {y',x \in \olg_{\alpha_j}^{ext}}}
     \prod_{i=0}^j \mu(\alpha_i) \notag \\
   &\leq  \sum_{y \in \ZZ^2}\ \sum_{{y' \in \ZZ^2} \atop {|y' - y| \leq 1}}
     (180s)^{|y|+j+1} (18s)^{|x-y'|+2} \notag \\
   &\leq 18s(180s)^{j+1} \sum_{y \in \ZZ^2} 5(180s)^{|y|} (18s)^{|x-y|} \notag \\
   &\leq 5(18s)^{|x|+1}(180s)^{j+1} \sum_{y \in \ZZ^2} 
     10^{|y|} (18s)^{2\dist(y,R(x))}.
\end{align}
Now 
\begin{align} \label{increase2}
  \sum_{y \in \ZZ^2} 10^{|y|} (18s)^{2\dist(y,R(x))}
    &\leq \sum_{j=1}^\infty j10^{|x|-j+1} + \sum_{y \in \ZZ^2 \bs R(x)} 
    10^{|y|} (18s)^{2\dist(y,R(x))} \notag \\
  &\leq \left( \frac{10}{9} \right)^2 10^{|x|} + 10^{|x|}\cdot 8000s^2 \notag \\
  &\leq 2 \cdot 10^{|x|}.
\end{align}
This and \eqref{increase} show that \eqref{hypoth} holds for $j$, and the induction is complete.  
Since $s \leq 2 \cdot 3^{-12}$, with \eqref{msum} and \eqref{simple} this shows that
\beq\label{finalsum}
  \sum_{X:0,x \in \Supp(X)} \left| \vphi_u^{\rm T}(X) \right| 
  \leq \sum_{m=0}^\infty 3^{4m+2}(180s)^{|x|+m+2} \leq 10(180s)^{|x|+2}.
\eeq
This completes the proof of (i).

For (ii) we may assume $x=0$.  If $0 \in \Supp \zeta$, \eqref{linkcost1} follows from 
the Convergence Theorem, so we assume $r \geq 1$.  Given $X$ let $\ti U_m(X,\zeta)$ 
denote the set of all $y \in \ZZ^2$ with $d(y,\Supp \zeta) \leq 1$ for 
which there exists %$y' \in \ZZ^2$ with $|y-y'| \leq 1$, and
a minimal incompatibility path $(\alpha_0,\dots,\alpha_m)$ in $g(X)$, with $\alpha_0=\zeta$,
$y \in \Supp \alpha_1$ and $0 \in \Supp \alpha_m$.  Then using Lemma \ref{pertcost}(i) as in \eqref{simple}, 
and using \eqref{hypoth},
\begin{align}\label{onefixed}
  \sum_{X:\zeta\in X,\atop{0 \in \Supp(X)}} \left| \vphi_u^{\rm T}(X) \right| &\leq
    \sum_{y: d(y,\Supp \zeta) \leq 1} \sum_{m=1}^\infty \sum_{X:y \in \ti U_m(X,\zeta)}
    \left| \vphi_u^{\rm T}(X) \right| \notag\\
  &\leq \sum_{y: d(y,\Supp \zeta) \leq 1} \sum_{m=1}^\infty \sum_{(\alpha_0,\dots,\alpha_m) \atop
    {{\alpha_0=\zeta,0 \in \Supp \alpha_m,} \atop{y \in \Supp \alpha_1}}} 
    \sum_{X \in \mG(\alpha_0,\dots,\alpha_m)} \left| \vphi_u^{\rm T}(X) \right| \notag\\
  &\leq \mu(\zeta) \sum_{y: d(y,\Supp \zeta) \leq 1} \sum_{m=1}^\infty
    3^{4m+2} \sum_{{{(\alpha_1,\dots,\alpha_m)} \atop {y \in \olg_{\alpha_1}^{ext},0 \in \olg_{\alpha_m}^{ext}}} 
    \atop {\text{all $\alpha_i$ simple}}} \prod_{i=1}^m \mu(\alpha_i) \notag\\
  &\leq \mu(\zeta) \sum_{y: |y| \geq r-1} \sum_{m=1}^\infty
    3^{4m+2} (180s)^{|y|+m+1} \notag \\
  &\leq 3^7\mu(\zeta) \sum_{y: |y| \geq r-1} (180s)^{|y|+2} \notag\\
  &\leq 3^9r\mu(\zeta)(180s)^{r+1}.
  \end{align}
This proves \eqref{linkcost1}.  The proof of \eqref{linkcost1exp} is similar, using \eqref{factorexp} in place of
\eqref{factor} in Lemma \ref{pertcost}(i).
\end{proof}

We say a cylinder $\ga$ in an elementary perturbation is \emph{touching} if $I(\ga)=0$.
A touching cylinder is \emph{small} if $|\tga| \leq 6$ and \emph{big} if $|\tga| \geq 8$.
We write $\om: h \to 0$, and say $\om$ is \emph{touching},
to designate that $\om$ is an elementary perturbation from
height $h$ containing a touching cylinder, 
and we write $X: h \to 0$ to designate that the cluster $X$ contains such an elementary perturbation.  
We say that a touching elementary perturbation $\om: h \to 0$ is \emph{multi-touching} if $\om$ includes two or more
touching cylinders (possibly with nested interiors); otherwise $\om$ is \emph{single-touching}.
We write $\C_k^{si}(\Lambda,h)$ for the set of all single-touching elementary perturbations 
$\om: h \to 0$.
For cylinders $\ga_1,\ga_2$ we write $\ga_1 \prec \ga_2$ to mean that $\olg_1$ is a 
proper subset of $\olg_2$, $E(\ga_1) = I(\ga_2)$ and $\ga_1,\ga_2$ are not separated by any cylinder.  
%%In other words, $\ga_1$ is maximal, in the partial ordering by inclusion of interiors, among those cylinders 
%%inside the base perimeter of $\ga_2$.

For each $x \in \mathbb{Z}^2$ and each $\om:h \to 0$ with 
$\bar{\om} \ni x$ and $x$ at height 0 in $\om$, there exists a unique maximal compatible family 
$\mT_x(\om) = \{\ga_1\prec\dots\prec\ga_r\}$  with 
\beq \label{Mconds}
  x \in \olg_1,\ I(\ga_1)=0 \quad \text{and} \quad \ga_r=\ga_{\om}^{\rm ext}.
  \eeq
This family also makes up an elementary perturbation; we call this special nested
type of perturbation a \emph{tornado} above $x$.  
We write $\ga_1^{\om}$ for the innermost cylinder in $\mT_x(\om)$, and denote by $T_h^x$ the set of all
tornadoes (for a given $h$) above $x$.

We say a tornado $\om:h \to 0, \om = \{\ga_1\prec\dots\prec\ga_r\}$, is 
\emph{semi-monotone} if $I(\ga_r) > \dots > I(\ga_1)$.   
The tornado is \emph{fully monotone} if also $h > I(\ga_r)$.
Given a tornado $\alpha:h \to 0, \alpha = \{\ga_1\prec\dots\prec\ga_r = \ga_\alpha^{ext}\}$, we can construct a 
semi-monotone tornado $\mM(\alpha):h \to 0$ from $\alpha$ 
as follows.  Let 
$i_1 < \dots < i_m = r$ be the indices $i$ for which either $i=r$ or $I(\ga_i) < \min_{j>i} I(\ga_j)$.
For $l \leq m-1$ define the cylinder $\zeta_l = (\tilde{\ga}_{i_l},I(\ga_{i_{l+1}}),I(\ga_{i_l}))$, 
and let $\zeta_m = \ga_{i_m}$. 
Note that $\zeta_l$ is obtained by truncating the top of the cylinder $\ga_{i_l}$ from its height
$E(\ga_{i_l})$ to the 
possibly lower height $I(\ga_{i_{l+1}})$, that is, we retain only the portion of $\ga_{i_l}$ which 
projects below all cylinders which are larger, in the ordering of $\alpha$ by $\prec$.
Finally let $\mM(\alpha) = \{\zeta_1 \prec \cdots \prec \zeta_m\}$ denote the resulting 
monotone tornado, 
and for an elementary perturbation $\om:h \to 0$ and a site $x$ at height 0 in $\om$ let 
$\mM_x(\om) = \mM(\mT_x(\om))$.  We denote by $M_h^x$ and $F_h^x$ the sets of all semi-monotone
and fully monotone tornadoes, respectively, above $x$, for a given $h$.

Let $\Omega_{(a)}$ and $\Omega_{(c)}$ denote the sets of elementary perturbations of 
types (a), (c), respectively, in Figure 6, at arbitrary location.  (The fixed
height $h$ is suppressed in the notation but should be clear from the context.)  

\begin{Prop} \label{unifh}
There exists $C$ as follows.  Let $k \geq 8$, $t<t_1(k)$ and $s=s(t)=te^{t^{1/4}}$.  Then for all $h \geq 1$,
\beq \label{PhQh}
  |P_h(t)| \leq 3^6s^{2h+1}, \quad |Q_h(t)| \leq 3^{16}s^{3h+1},    
  \eeq
and for all $h \geq 0$ and $u<t^{1/2}$,
\beq \label{Vh}
  |V_h(t,u)| \leq Cs^{3h+4}.
  \eeq
\end{Prop}

\begin{proof}
As in \eqref{geom} we have for all $h \geq 1$,
\begin{align} \label{Phfirst}
  |P_h(t)| &\leq \sum_{ {\om \in M_h^0} \atop {|\olg_1^{\om}| = 1,\om \notin \Omega_{(a)}} }\ 
    \sum_{ {\alpha: h \to 0} \atop {\mM(\alpha) = \om} }\ 
    \sum_{ {X:h \to 0} \atop {\alpha \in X} } \left| \vphi_u^{\rm T}(X) \right| \notag \\
  &\leq \sum_{ {\om \in M_h^0} \atop {|\olg_1^{\om}| = 1,\om \notin \Omega_{(a)}} }\ 
    \sum_{ {\alpha: h \to 0} \atop {\mM(\alpha) = \om} } \mu(\alpha) \notag \\
  &\leq \sum_{ {\om \in M_h^0} \atop {|\olg_1^{\om}| = 1,\om \notin \Omega_{(a)}} } 
    \mu(\om) \exp(9s^{1/2}|\olg_{\om}^{ext}|) \notag \\
  &\leq 25 \sum_{ {\om \in M_h^0} \atop {|\olg_1^{\om}| = 1,\om \notin \Omega_{(a)}} } \mu(\om),
  \end{align}
where the third inequality follows from the Convergence Theorem.
Similarly for all $h \geq 1$ we have
\beq \label{Qhfirst}
  |Q_h(t)| \leq 25 \sum_{ {\om \in M_h^0} \atop {|\tga_1^{\om}| = 6,\om \notin \Omega_{(c)} } } \mu(\om).
  \eeq
Before bounding the right sides of \eqref{Phfirst} and \eqref{Qhfirst}, we will obtain 
bounds for (parts of) $|V_h(t,u)|$ also in terms of sums of weights $\mu(\cdot)$.
We will then bound these sums of weights in a somewhat unified way.  
  
For $V_h$ with $h \geq 1$,
we split the corresponding sum over clusters into several parts, according to the 
external height ($h$ or $h+1$) and according to five types.
For height $h$:
\begin{itemize}
\item[(i)] Type 1 consists of those clusters $X$ in 
which some perturbation contains a big touching cylinder.
\item[(ii)] Type 2 consists of those clusters, not of type 1, in which some perturbation is
multi-touching.
\item[(iii)] Type 3 consists of those clusters, not of types 1 and 2, of diameter at least $16h$,
in which some perturbation is touching.
\item[(iv)] Type 4 consists of those clusters, not of types 1, 2, 3, 
 in which there are two or more (possibly equal) touching
perturbations. 
\end{itemize}
Note that outside of these four types, the other clusters $X:h \to 0$ 
are types (a)--(d) in Figure 6, and for types 3 and 4, 
all touching perturbations must be of these types (a)--(d).
For height $h+1$, there is a fifth type:  we can take a cluster $X:h \to 0$ of Types 1--4, ``lift'' it everywhere by one
height unit to create a cluster $X^+:h+1 \to 1$, and then extend it downward to height 0 by
either (A) lowering the interior height of a single small cylinder in $X^+$ from 1 to 0, or 
(B) adding a new compatible small cylinder $\ga$ to $X^+$ with $E(\ga) = 1, I(\ga) = 0$.
Thus for height $h+1$:
\begin{itemize}
\item[(v)] Type 5 consists of clusters created by extension, as in (A) or (B) above,
of a cluster $X:h \to 0$ of types 1--4.
\end{itemize}
For $h \geq 1$ and $i=1,2,3,4$ we write $W_{h,i}(t,u)$ for the contribution to the sum $V_h(t,u)$
from clusters $X:h \to 0$ of type $i$, and $W_{h+1,5}(t,u)$ for the contribution to $V_h(t,u)$
from clusters $X:h+1 \to 0$ of type 5.  Then
\beq \label{sum7}
  V_h(t,u) = \sum_{i=1}^4 \big( W_{h,i}(t,u) - W_{h,i}(t,0) - W_{h+1,i}(t,u) \big) - W_{h+1,5}(t,u).
  \eeq

We begin with Type 1.  As in \eqref{Phfirst}, using the argument of \eqref{geom} 
to bound the effect of summing over $\alpha$, we have for all $h \geq 1$:
\begin{align} \label{type1}
  |W_{h,1}(t,u)| &\leq \sum_{ {\om \in M_h^0} \atop {|\tga_1^{\om}| \geq 8} }\ 
    \sum_{ {\alpha: h \to 0} \atop {\mM_0(\alpha) = \om} }\ 
    \sum_{ {X:h \to 0} \atop {\alpha \in X} } \frac{1}{|\olg_1^{\om}|} \left| \vphi_u^{\rm T}(X) \right| \notag \\
  &\leq 25 \sum_{ {\om \in M_h^0} \atop {|\tga_1^{\om}| \geq 8} }\ 
    \frac{1}{|\olg_1^{\om}|} \mu(\om).
  \end{align}

Turning to Type 2, let $\om$ be a multi-touching perturbation with no big touching cylinders, 
so all touching cylinders in $\om$ have disjoint interiors,
and suppose 0 and some site $x \neq 0$ are in the interiors of some such disjoint touching cylinders. 
Then $|x| \geq 2$.  In $\om$ there must 
exist an innermost cylinder $\ga$ with both $0,x \in \olg$, and  we have $|\tga| \geq 2(|x|+2) \geq 8$ and 
$I(\ga)=j$ for some $j \geq 1$.
The tornadoes $\mT_0(\om)$ and $\mT_x(\om)$ must both include $\ga$, 
and the cylinders larger than $\ga$ in the ordering of the tornado are the same in both tornadoes.
Thus there is a unique triple $\mT_{0,x}(\om) = (\eta,\zeta,\xi)$ associated to $\om$ and $x$, in which 
$\zeta,\xi$ are tornadoes:
\[
  \zeta = \{\zeta_1 \prec \dots \prec \zeta_r\} \text{ with } 0 \in \bar{\zeta}_1, 
    |\tilde{\zeta}_1| \leq 6,
    I(\zeta_i) \geq 1 \text{ for all } i \geq 2, \text{ and } \zeta:j \to 0,
  \]
\[
  \xi = \{\xi_1 \prec \dots \prec \xi_q\} \text{with } x \in \bar{\xi}_1,
    |\tilde{\xi}_1| \leq 6,
    I(\xi_i) \geq 1 \text{ for all } i \geq 2, \text{ and } \xi:j \to 0,
  \]
and
\begin{align} \label{cloud}
  \eta = \{\eta_1 \prec \dots \prec \eta_m\}, \text{with } |\tilde{\eta}_1| \geq 2(|x|+2), 
     I(\eta_i) \geq 1 \text{ for all } i \geq 1, \text{ and } \eta:h \to j,
  \end{align}
with $\ga = \eta_1$, $\bar{\zeta}_r \cap \bar{\xi}_q = \phi$, $\zeta_r \prec \eta_1$ and $\xi_q \prec \eta_1$.
Here $\eta:h \to j$ means the innermost cylinder $\eta_1$ has interior height $j$.
We write $T_h^{0,x}$ for the set of all such triples $(\eta,\zeta,\xi)$ (with $\ga,j$ arbitrary.)
We call a perturbation $\eta$ as in \eqref{cloud} a \emph{cloud} of height $j$,
and write $C_k^{cl}(\Lambda,h,j)$ for the set of all such clouds.  
Note that the cylinders in $\eta,\zeta,\xi$ together make up an elementary perturbation,
so our earlier definition gives the weight $\mu(\eta,\zeta,\xi) = \mu(\eta)\mu(\zeta)\mu(\xi)$.
The \emph{total depth} of a tornado or cloud $\eta = \{\eta_1 \prec \dots \prec \eta_m\}$ is
\[
  D(\eta) = \sum_{i=1}^m L(\eta_i).
  \]
Again analogously to \eqref{Phfirst}
we have for all $h \geq 1$:
\begin{align} \label{type2}
  |W_{h,2}(t,u)| &\leq \sum_{x:2 \leq |x| \leq 3k+3}\ \sum_{(\eta,\zeta,\xi) \in T_h^{0,x}}
    \sum_{{\alpha:h \to 0} \atop {\mT_{0,x}(\alpha) = (\eta,\zeta,\xi)}}
    \sum_{ {X:h \to 0} \atop {\alpha \in X} } \left| \vphi_u^{\rm T}(X) \right| \notag \\
  &\leq \sum_{x:2 \leq |x| \leq 3k+3}\ \sum_{(\eta,\zeta,\xi) \in T_h^{0,x}}
    \sum_{{\alpha:h \to 0} \atop {\mT_{0,x}(\alpha) = (\eta,\zeta,\xi)}} \mu(\alpha) \notag \\
  &\leq 25 \sum_{x:2 \leq |x| \leq 3k+3}\ \sum_{(\eta,\zeta,\xi) \in T_h^{0,x}}
    \mu(\eta)\mu(\zeta)\mu(\xi) \notag \\
  &\leq 25 \sum_{x:2 \leq |x| \leq 3k+3}\ \sum_{j \geq 1} 
    \left( \sum_{{\eta \in C_k^{cl}(\Lambda,h,j)} \atop {0 \in \bar\eta_1,|\tilde{\eta}_1| \geq 2(|x|+2)}} \mu(\eta) \right)
    \left( \sum_{{\zeta \in T_j^0} \atop {|\tilde{\zeta}_1| \leq 6}} \mu(\zeta) \right)^2 \notag \\
  &\leq 25^3 \sum_{x:2 \leq |x| \leq 3k+3}\ \sum_{j \geq 1} 
    \left( \sum_{{\eta \in C_k^{cl}(\Lambda,h,j)} \atop {0 \in \bar\eta_1,|\tilde{\eta}_1| \geq 2(|x|+2)}} \mu(\eta) \right)
    \left( \sum_{{\zeta \in M_j^0} \atop {|\tilde{\zeta}_1| \leq 6}} \mu(\zeta) \right)^2.
\end{align}
Here in the last inequality, we have bounded the sum over $T_j^0$ in terms of the sum
over $M_j^0$, using the argument of \eqref{geom}.

Next consider Type 3.  We conclude from Lemma \ref{chains}(i) that for all $h \geq 1$,
\begin{align} \label{type3}
  |W_{h,3}(t,u)| &\leq \sum_{x:|x| = 8h-2} 
    \sum_{X:0,x \in \Supp(X)} \left| \vphi_u^{\rm T}(X) \right| \notag \\
  &\leq 10 \sum_{x:|x| = 8h-2} (180s)^{|x|+2} \notag \\
  &\leq 40(8h-2)(180s)^{8h} \notag \\
  &\leq s^{3h+4},
\end{align}
where in the last inequality we used $s \leq s_8 = 2 \cdot 3^{-12}$.  Note that 
\eqref{type3} (excluding the first inequality) is also valid for $h=0$, if we replace $8h-2$ with 8.

Now we consider Type 4.  We use Lemma \ref{pertcost}(ii,iii) and once more
reason analogously to \eqref{Phfirst} to obtain that, for all $h \geq 1$,
\begin{align} \label{type4}
  |W_{h,4}&(t,u)| \notag \\
  &\leq \sum_{{x \in {\bf Z}^2} \atop {|x| \leq 16h}}
    \sum_{{\om \in M_h^0} \atop {|\tilde{\ga}_1^{\om}| \leq 6}}
    \sum_{{\alpha \in M_h^x} \atop {|\tilde{\ga}_1^{\alpha}| \leq 6}}
    \sum_{{\zeta:h \to 0} \atop {\mM_0(\zeta)=\om}}  
    \sum_{{\xi:h \to 0} \atop {\mM_x(\xi)=\alpha}}
    \sum_{{X:h \to 0} \atop {\zeta,\xi \in X}}
    \left| \vphi_u^{\rm T}(X) \right| \notag \\
  &\leq 3^6 \sum_{x \in {\bf Z}^2:|x|\leq 16h} 
    \left( \sum_{{\om \in M_h^0} \atop {|\tilde{\ga}_1^{\om}| \leq 6}}
      \sum_{{\zeta:h \to 0} \atop {\mM_0(\zeta)=\om}}
      \mu(\zeta) \right) 
      \left( \sum_{{\alpha \in M_h^x} \atop {|\tilde{\ga}_1^{\alpha}| \leq 6}}
      \sum_{{\xi:h \to 0} \atop {\mM_x(\xi)=\alpha}} \vphi(\xi) \right) \notag \\
   &\leq 3^6 \cdot 2(16h+1)^2
     \left( 25 \sum_{{\om \in M_h^0} \atop {|\olg_1^{\om}| \leq 2}} \mu(\om) \right)^2 \notag \\
   &\leq 3^{17}h^2 \left( \sum_{{\om \in M_h^0} \atop {|\olg_1^{\om}| \leq 2}} \mu(\om) \right)^2.
\end{align}
Here on the right side of the first inequality, for terms with $\zeta=\xi$ we interpret $\zeta,\xi \in X$ as meaning $X(\zeta)\geq 2$ and use Lemma \ref{pertcost}(iii).

Last, for Type 5 we observe that when we create a cluster by extension, lifting a cluster $X$ and 
adding a small cylinder to some elementary perturbation $\alpha \in X$, in the factor $a^{\rm T}$ from \eqref{aT},
$X!$ is increased by a factor of $X(\alpha)$ and (as in the proof of Lemma \ref{pertcost}(iii)) the sum over subgraphs is
unchanged.  With this observation, we can decompose the sum $W_{h+1,5}(t,u)$ into sums $W_{h+1,5,i}(t,u), i=1,2,3,4$, 
according to the type $i$ of the elementary perturbation that is extended, and then apply a modified version of
the first inequality in each of \eqref{type1}, \eqref{type2}, \eqref{type3} and \eqref{type4} in which each term 
$|\vphi_u^{\rm T}(X)|$ is multiplied by $X(\alpha)|\bar\ga_1^\alpha|(t^2+4t^3)$ 
(or similar with $\omega$ in place of $\alpha$), as follows.  First, similarly to \eqref{type1}, for $h \geq 1$,
\begin{align} \label{type5-1}
  |W_{h+1,5,1}(t,u)| &\leq (t^2+4t^3) \sum_{ {\om \in M_h^0} \atop {|\tga_1^{\om}| \geq 8} }\ 
    \sum_{ {\alpha: h \to 0} \atop {\mM_0(\alpha) = \om} }\ 
    \sum_{ {X:h \to 0} \atop {\alpha \in X} } X(\alpha) \left| \vphi_u^{\rm T}(X) \right| \notag \\
  &\leq s^2 \sum_{ {\om \in M_h^0} \atop {|\tga_1^{\om}| \geq 8} }\ 
    \sum_{ {\alpha: h \to 0} \atop {\mM_0(\alpha) = \om} } (e^{\mu(\alpha)}-1) \notag \\
  &\leq 3^3 s^2  \sum_{ {\om \in M_h^0} \atop {|\tga_1^{\om}| \geq 8} } \mu(\om).
  \end{align}
Second, similarly to \eqref{type2}, for $h \geq 1$,
\begin{align} \label{type5-2}
  |W_{h+1,5,2}&(t,u)| \notag \\
  &\leq (t^2+4t^3)\sum_{x:2 \leq |x| \leq 3k+3}\ \sum_{(\eta,\zeta,\xi) \in T_h^{0,x}}
    \sum_{{\alpha:h \to 0} \atop {\mT_{0,x}(\alpha) = (\eta,\zeta,\xi)}}
    \sum_{ {X:h \to 0} \atop {\alpha \in X} } |\bar\ga_1^\zeta| X(\alpha) \left| \vphi_u^{\rm T}(X) \right| \notag \\
  &\leq 2s^2\sum_{x:2 \leq |x| \leq 3k+3}\ \sum_{(\eta,\zeta,\xi) \in T_h^{0,x}}
    \sum_{{\alpha:h \to 0} \atop {\mT_{0,x}(\alpha) = (\eta,\zeta,\xi)}} (e^{\mu(\alpha)}-1) \notag \\
  &\leq 3^4s^2 \sum_{x:2 \leq |x| \leq 3k+3}\ \sum_{(\eta,\zeta,\xi) \in T_h^{0,x}}
    \mu(\eta)\mu(\zeta)\mu(\xi) \notag \\
  &\leq 3^{10}s^2 \sum_{x:2 \leq |x| \leq 3k+3}\ \sum_{j \geq 1} 
    \left( \sum_{{\eta \in C_k^{cl}(\Lambda,h,j)} \atop {0 \in \bar\eta_1,|\tilde{\eta}_1| \geq 2(|x|+2)}} \mu(\eta) \right)
    \left( \sum_{{\zeta \in M_j^0} \atop {|\tilde{\zeta}_1| \leq 6}} \mu(\zeta) \right)^2.
\end{align}
Third, for an elementary perturbation $\om$, let $r(\om) = \max\{d(x,0): x \in \Supp(\om)\}$.
We claim that for all $h,l \geq 1$,
\beq\label{claim}
  \sum_{ {\om \in M_h^0} \atop {|\tga_1^{\om}| \leq 6 \atop{r(\om)=l} } } \mu(\om)
    \leq 3^{10}l^2(9s)^l s^{2h}.
\eeq
Assuming this claim, similarly to \eqref{Phfirst}, \eqref{Qhfirst} and \eqref{type3}, using Lemma \ref{chains}(ii) 
and $s \leq s_8$ we obtain that for $h \geq 1$,
\begin{align} \label{type5-3}
  |&W_{h+1,5,3}(t,u)| \notag \\
    &\leq (t^2+4t^3) \Bigg( \sum_{l=0}^{8h-3}\ \sum_{ {\om \in M_h^0} \atop {|\tga_1^{\om}| \leq 6
    \atop{r(\om)=l} } }\ 
    \sum_{ {\alpha: h \to 0} \atop {\mM_0(\alpha) = \om} }\ \sum_{x:|x| = 8h-2}\ 
    \sum_{ {X:h \to 0} \atop {\alpha \in X \atop {x \in \Supp(X)} } } |\bar\ga_1^\alpha|\ 
    X(\alpha) \left| \vphi_u^{\rm T}(X) \right| \notag \\
  &\qquad \qquad \qquad + \sum_{ {\om \in M_h^0} \atop {|\tga_1^{\om}| \leq 6
    \atop{r(\om) \geq 8h-2} } }\ 
    \sum_{ {\alpha: h \to 0} \atop {\mM_0(\alpha) = \om} }\ 
    \sum_{ {X:h \to 0} \atop {\alpha \in X} } |\bar\ga_1^\alpha| 
    X(\alpha) \left| \vphi_u^{\rm T}(X) \right| \Bigg) \\
  &\leq s^2 \Bigg( \sum_{l=0}^{8h-3}\ \sum_{ {\om \in M_h^0} \atop {|\tga_1^{\om}| \leq 6
    \atop{r(\om)=l} } } 
    \sum_{ {\alpha: h \to 0} \atop {\mM_0(\alpha) = \om} } 4(8h-2)3^{10}(8h-2-l)\mu(\alpha)(180s)^{8h-1-l} \notag \\
  &\qquad \qquad + 3 \sum_{ {\om \in M_h^0} \atop {|\tga_1^{\om}| \leq 6
    \atop{r(\om) \geq 8h-2} } }\ \sum_{ {\alpha: h \to 0} \atop {\mM_0(\alpha) = \om} }\ \mu(\alpha) \Bigg) \notag \\
  &\leq 25s^2 \Bigg[ 256\cdot 3^{10}h^2 \sum_{l=0}^{8h-3}\ 
    \Bigg( \sum_{ {\om \in M_h^0} \atop {|\tga_1^{\om}| \leq 6 \atop{r(\om)=l} } } \mu(\om) \Bigg)
    (180s)^{8h-1-l} +  3\sum_{ {\om \in M_h^0} \atop {|\tga_1^{\om}| \leq 6
    \atop{r(\om)\geq 8h-2} } } \mu(\om) \Bigg] \notag \\
  &\leq 25s^2 \Bigg[ 256\cdot 3^{10}h^2 (180s)^{8h-1}s^{2h} \sum_{l=0}^{8h-3} l^2 20^{-l} 
    + 4\cdot 3^{10}(8h-2)^2 (9s)^{8h-2}s^{2h} \Bigg] \notag \\ 
  %\frac{180}{179}\cdot 6400Ch^2(180s)^{10h+1} + \frac{1}{1-s}\cdot 75Cs^{10h} \notag\\
  &\leq 16000h^2 180^{-2h}(180s)^{10h} \notag \\
  &\leq (180s)^{10h}. \notag
\end{align}
Fourth, similarly to \eqref{type4}, using Lemma \ref{pertcost}(ii,iii), for $h \geq 1$,
\begin{align} \label{type5-4}
  |W_{h+1,5,4}&(t,u)| \notag \\
  &\leq (t^2+4t^3)\sum_{{x \in {\bf Z}^2} \atop {|x| \leq 16h}}
    \sum_{{\om \in M_h^0} \atop {|\tilde{\ga}_1^{\om}| \leq 6}}
    \sum_{{\alpha \in M_h^x} \atop {|\tilde{\ga}_1^{\alpha}| \leq 6}}
    \sum_{{\zeta:h \to 0} \atop {\mM_0(\zeta)=\om}}  
    \sum_{{\xi:h \to 0} \atop {\mM_x(\xi)=\alpha}}
    \sum_{{X:h \to 0} \atop {\zeta,\xi \in X}} |\bar\ga_1^\zeta|\ 
    X(\zeta) \left| \vphi_u^{\rm T}(X) \right| \notag \\  
  &\leq 2s^2 \sum_{{x \in {\bf Z}^2} \atop {|x| \leq 16h}}
    \sum_{{\om \in M_h^0} \atop {|\tilde{\ga}_1^{\om}| \leq 6}}
    \sum_{{\alpha \in M_h^x} \atop {|\tilde{\ga}_1^{\alpha}| \leq 6}}
    \sum_{{\zeta:h \to 0} \atop {\mM_0(\zeta)=\om}}  
    \sum_{{\xi:h \to 0} \atop {\mM_x(\xi)=\alpha}} 2\cdot 3^5\mu(\zeta)\varphi(\xi)
    \notag \\
  &\leq 4\cdot 3^5 s^2 \sum_{{x \in {\bf Z}^2} \atop {|x| \leq 16h}} \left( \sum_{{\om \in M_h^0} \atop {|\tilde{\ga}_1^{\om}| \leq 6}}
      \sum_{{\zeta:h \to 0} \atop {\mM_0(\zeta)=\om}}
      \mu(\zeta) \right) 
      \left( \sum_{{\alpha \in M_h^x} \atop {|\tilde{\ga}_1^{\alpha}| \leq 6}}
      \sum_{{\xi:h \to 0} \atop {\mM_x(\xi)=\alpha}} \vphi(\xi) \right) \notag \\
   &\leq 8\cdot 3^5(16h+1)^2 s^2
     \left( 25  \sum_{{\om \in M_h^0} \atop {|\tilde{\ga}_1^{\om}| \leq 6}} \mu(\om) \right)^2 \notag \\
   &\leq 3^{18}h^2 s^2 \left(  \sum_{{\om \in M_h^0} \atop {|\tilde{\ga}_1^{\om}| \leq 6}} \mu(\om) \right)^2.
\end{align}

We now establish bounds for the right sides of \eqref{Phfirst}, \eqref{Qhfirst}, \eqref{type1}, \eqref{type2} 
\eqref{type4}, \eqref{type5-1}, \eqref{type5-2}, and \eqref{type5-4}, and establish the claim \eqref{claim}.  
We first do this with $M_h^0$ replaced by $F_h^0$.
In fact we claim that for all $h \geq 1,q \geq 3$,
\beq \label{hyp1}
  S_1(h,2q) \equiv
  \sum_{ {\om \in F_h^0} \atop {|\olg_1^{\om}| = 1,\om \notin \Omega_{(a)} \atop{ |\ti\ga_{\om}^{ext}| \geq 2q} } } \mu(\om) \leq 
    6q^2 (9s)^q s^{2h-2} \prod_{j=1}^{h-1} \left( 1 + (3^{11}s)^{j-1} \right),
  \eeq
and when $q \geq 4$,
\beq \label{hyp2}
  S_2(h,2q ) \equiv
  \sum_{ {\om \in F_h^0} \atop {|\tga_1^{\om}| = 6,\om \notin \Omega_{(c)} \atop{ |\ti\ga_{\om}^{ext}| \geq 2q} } } \mu(\om)
    \leq 4q^2(9s)^q s^{3h-3} \prod_{j=1}^{h-1} \left( 1 + (3^{11}s)^{j-1} \right)
  \eeq
and
\beq \label{hyp3}
  S_3(h,2q) \equiv
  \sum_{ {\om \in F_h^0} \atop {|\tga_1^{\om}| \geq 8 \atop{ |\ti\ga_{\om}^{ext}| \geq 2q} } }\ \mu(\om) \leq q^2(9s)^q(3^{11}s^4)^{h-1}.
  \eeq
Of course the products in \eqref{hyp1} and \eqref{hyp2} are bounded in $h$
(bounded by 12, in fact, since $s \leq s_8$),
so they can be replaced by constants,
but their presence simplifies the induction.  
%We write $S_1(h),S_2(h),S_3(h)$ for the sums on the left sides 
%of \eqref{hyp1}, \eqref{hyp2}, \eqref{hyp3}, respectively.
For $h=1$ we have 
\[
  S_1(1,2q) = S_2(1,2q) = 0 \quad \text{for all } q \geq 3,
  \]
and we have the Peierls bound from \eqref{Peierlssum}:
\beq \label{Peierls5}
  S_3(1,2q) \leq \sum_{l \geq 2q} \left( \frac{l}{4} \right)^2 3^l s^{l/2} \leq q^2(9s)^q \quad \text{for all } q \geq 4.
  \eeq
Now suppose that \eqref{hyp1}---\eqref{hyp3} are valid for $h=1,\dots,m$, and all given values of $q$,
for some $m \geq 1$, and consider $h=m+1$.
Considering the effect of removing the lowest layer of cubes from a fully monotone tornado in
$F_{m+1}^h$ we see that
\begin{align} \label{hyp1a}
  S_1&(m+1,2q) \notag \\
  &\leq s^2\left(S_1(m,2q) + 4s^{3m}\delta_{\{q=3\}} + S_2(m,2q) + S_3(m,2q)\right) \notag \\
  &\leq 6q^2(9s)^qs^{2m}  
    \left( 1 + \frac{2}{3^9}s^{m-1} + \frac{2}{3}s^{m-1} +\frac{1}{6}(3^{11}s^2)^{m-1} \right)
    \prod_{j=1}^{m-1} \left( 1 + (3^{11}s)^{j-1} \right) \notag \\
  &\leq 6q^2(9s)^qs^{2m}  
    \prod_{j=1}^m \left( 1 + (3^{11}s)^{j-1} \right),
  \end{align}
where we used $s \leq s_8$.  Similarly,
\begin{align} \label{hyp2a}
  S_2(m+1,2q) &\leq s^3(S_2(m) + 4S_3(m)) \notag \\
  &\leq 4q^2(9s)^qs^{3m-3} \left( 1 + (3^{11}s)^{m-1} \right) 
    \prod_{j=1}^{m-1} \left( 1 + (3^{11}s)^{j-1} \right) \notag \\
  &= 4q^2(9s)^qs^{3m-3} 
    \prod_{j=1}^m \left( 1 + (3^{11}s)^{j-1} \right)
  \end{align}
and
\begin{align} \label{hyp3a}
  S_3(m+1,2q) &\leq S_3(1,8) S_3(m,2q) \notag \\
  &\leq 16(9s)^4 \cdot q^2(9s)^q (3^{11} s^4)^{m-1} \notag \\
  &\leq q^2(9s)^q (3^{11} s^4)^m,
  \end{align}
so \eqref{hyp1}--\eqref{hyp3} are valid for $h=m+1$ as well, establishing the claim.

We would now like to replace $F_h^0$ with $M_h^0$ in \eqref{hyp1}--\eqref{hyp3}.
If $\om:h \to 0, \om = \{\ga_1 \prec \dots \prec \ga_k = \ga_{\om}^{ext}\}$,
is a semi-monotone tornado for some $h \geq 1$, and $I(\ga_{\om}^{ext})=h+j$
for some $j \geq 1$, 
then $\{\ga_1 \prec \dots \prec \ga_{k-1}\}$ is a fully monotone tornado from $h+j$ to 0.  Therefore for $h \geq 1$
and $q \geq 3$,
using \eqref{Peierlssum} and \eqref{hyp1}, and using $S_1(h+j,6) \leq 2s^{2(h+j)}$ for $s \leq s_8$,
\begin{align} \label{hyp1b}
  \sum_{ {\om \in M_h^0} \atop {|\olg_1^{\om}| = 1,\om \notin \Omega_{(a)} 
    \atop{ |\ti\ga_{\om}^{ext}| \geq 2q} } } \mu(\om) &\leq S_1(h,2q) + 
    \sum_{j \geq 1} \left(s^{2(h+j)} + S_1(h+j,6)\right) \sum_{l \geq 2q} \left( \frac{l}{4} \right)^2 3^l s^{jl/2} \notag \\
  &\leq S_1(h,2q) + 3s^{2h} \sum_{l \geq 2q} \left( \frac{l}{4} \right)^2 3^l \sum_{j \geq 1} s^{j(l+4)/2} \notag \\
  &\leq 72q^2(9s)^q s^{2h-2} + q^2(9s)^q s^{2h+2} \notag \\
  &\leq 3^4q^2(9s)^q s^{2h-2}.
  \end{align}
Similarly, for $h \geq 1$ and $q \geq 4$, using $S_2(h+j,8) \leq 19s^{3(h+j)}$ for $s \leq s_8$,
\begin{align} \label{hyp2b}
  \sum_{ {\om \in M_h^0} \atop {|\tga_1^{\om}| = 6,\om \notin \Omega_{(c)} 
    \atop{ |\ti\ga_{\om}^{ext}| \geq 2q} } } \mu(\om) &\leq S_2(h,2q) + 
    \sum_{j \geq 1} (s^{3(h+j)} + S_2(h+j,8)) \sum_{l \geq 2q} \left( \frac{l}{4} \right)^2 3^l s^{jl/2} \notag \\
  &\leq S_2(h,2q) + 20s^{3h} \sum_{l \geq 2q} \left( \frac{l}{4} \right)^2 3^l \sum_{j \geq 1} s^{j(l+6)/2} \notag \\
  &\leq 48q^2(9s)^q s^{3h-3} + 7q^2(9s)^q s^{3h+3} \notag \\
  &\leq 3^4q^2(9s)^q s^{3h-3}
  \end{align}
and using $S_3(h+j,8) \leq (3^{11}s^4)^{h+j}$,
\begin{align} \label{hyp3b}
  \sum_{ {\om \in M_h^0} \atop {|\tga_1^{\om}| \geq 8 \atop{ |\ti\ga_{\om}^{ext}| \geq 2q} } } \mu(\om) &\leq S_3(h,2q) + 
    \sum_{j \geq 1} S_3(h+j,8) \sum_{l \geq 2q} \left( \frac{l}{4} \right)^2 3^l s^{jl/2} \notag \\
  &\leq S_3(h,2q) + (3^{11}s^4)^h \sum_{l \geq 2q} \left( \frac{l}{4} \right)^2 3^l
    \sum_{j \geq 1} (3^{11} s^{(l+8)/2} )^j \notag \\
  &\leq q^2(9s)^q(3^{11}s^4)^{h-1} + q^2(9s)^q(3^{11}s^4)^{h+1} \notag \\
  &\leq 2q^2(9s)^q(3^{11}s^4)^{h-1}.
  \end{align}
From \eqref{Phfirst}, \eqref{Qhfirst}, \eqref{hyp1b} with $q=3$ and \eqref{hyp2b} 
with $q=4$, we obtain that for $h \geq 1$,
\beq \label{PQhsecond}
  |P_h(t)| \leq 3^{12}s^{2h+1}, \quad |Q_h(t)| \leq 3^{15}s^{3h+1}.
  \eeq
From \eqref{type1} and \eqref{hyp3b} with $q=4$, we obtain that for $h \geq 4$,
\beq \label{type1a}
  |W_{h,1}(t,u)| \leq 30(3^{11}s^4)^h \leq 3^{48}s^{3h+4}.
  \eeq
Note that the first inequality in \eqref{type1a}, but not the second, is valid for $h=1,2,3$.
From \eqref{type4}, \eqref{hyp1b} with $q=3$ and \eqref{hyp2b} 
with $q=4$, we obtain that for $h \geq 4$,
\beq \label{type4a}
  |W_{h,4}(t,u)| \leq 3^{17}h^2 \left( s^{2h} + 3^{12}s^{2h+1} + s^{3h} + 16\cdot 3^{12}s^{3h+1} \right)^2
    \leq 3^{22}s^{3h+4}.
  \eeq
  
To control $|W_{h,2}(t,u)|$ we need a bound for the sum over clouds $\eta$ on the right side of 
\eqref{type2}.  We proceed by induction on the total depth $D(\eta)$.  
We claim that for all $d \geq 1, |x| \geq 2$ and all $h,j \geq 1$ with $d \geq |h-j| \vee 1$, 
\beq \label{hypcloud}
  \sum_{{\eta \in C_k^{cl}(\Lambda,h,j)} \atop {0 \in \bar\eta_1,|\tilde{\eta}_1| \geq 2(|x|+2),D(\eta)=d}} \mu(\eta)
    \leq (18s)^{d(|x|+2)}.
  \eeq
Write $V(d,h,j,|x|)$ for the sum in \eqref{hypcloud}.
For $d=1$, \eqref{hypcloud} is a consequence of a Peierls argument and \eqref{Peierlssum}:  
for $h \geq 1, j = h \pm 1$ and $|x| \geq 2$, similarly to \eqref{PeierlsA} we have
\beq \label{Peierls4}
  V(1,h,j,|x|) \leq \sum_{l \geq 2(|x|+2)} 
    2\left( \frac{l}{4} \right)^2 3^l s^{l/2} \leq \frac{5}{8}(|x|+2)^2 (9s)^{|x|+2} \leq (18s)^{|x|+2}.
  \eeq
Now let $m \geq 1$ and suppose \eqref{hypcloud} is valid whenever $d \leq m$.  
A cloud $\eta$ of depth $m+1$ can be obtained by adding a cylinder of length 1 to 
a cloud of depth $m$, or by extending the length of the innermost cylinder $\eta_1$ by 1.
Therefore using \eqref{Peierls4}, for $h,j \geq 1$,
\begin{align} \label{induccloud}
  V(m&+1,h,j,|x|) \notag \\
  &\leq \big( V(m,h,j-1,|x|) + V(m,h,j+1,|x|) \big) \sum_{l \geq 2(|x|+2)} 
    \left( \frac{l}{4} \right)^2 3^l s^{l/2} \notag \\
  &\leq (18s)^{m(|x|+2)} (18s)^{|x|+2} \notag \\
  &= (18s)^{(m+1)(|x|+2)},
  \end{align}
so \eqref{hypcloud} is valid for $d=m+1$ as well, establishing the claim.

Consider now $W_{h,2}(t,u)$.  
Observe that the quantity in parentheses in \eqref{type4a} is actually a bound for the
quantity in parentheses on the right side of \eqref{type4}, and hence it is also a bound
for the quantity which is squared on the right side of \eqref{type2}.  With this fact,
summing \eqref{hypcloud} over $d$ with $d \geq |h-j| \vee 1$, and plugging into \eqref{type2},
we obtain analogously to \eqref{type3} that for $h \geq 4$,
\begin{align} \label{type2a}
  |&W_{h,2}(t,u)| \notag \\
  &\leq 25^3 \sum_{x:2 \leq |x| \leq 3k+3}\ \sum_{j \geq 1} 
    2(18s)^{(|h-j| \vee 1)(|x|+2)} \notag \\
  &\qquad\qquad\qquad\qquad\qquad \cdot
    \left( s^{2j} + 3^{12}s^{2j+1} + s^{3j} + 16\cdot 3^{12}s^{3j+1} \right)^2 \notag \\
  &\leq 10 \cdot 25^3 \sum_{j \geq 1}s^{4j} \sum_{x:|x| \geq 2} (18s)^{(|h-j| \vee 1)(|x|+2)} \notag \\
  &\leq 3^4 \cdot 25^3 \sum_{j \geq 1}s^{4j} (18s)^{4(|h-j| \vee 1)} \notag \\
  &\leq \frac{18}{17} \cdot 3^4 \cdot 25^3 (18s)^{4(h-1)}s^4 \notag \\
  &\leq 3^{42} s^{3h+4}.
  \end{align}
  
Turning to the sums $W_{h+1,5,i}(t,u)$, similarly to the bound \eqref{type1a} on \eqref{type1}, 
for $h \geq 4$ \eqref{type5-1} leads to
\beq \label{type5-1a}
  |W_{h+1,5,1}(t,u)| \leq 3^{48}s^{3h+6},
  \eeq
and similarly to the bound \eqref{type2a} on \eqref{type2}, \eqref{type5-2} leads to 
\beq \label{type5-2a}
  |W_{h+1,5,2}(t,u)| \leq 3^{43}s^{3h+6}.
  \eeq
For $W_{h+1,5,3}(t,u)$ with $h \geq 1$, let us establish the claim \eqref{claim} to complete the proof of \eqref{type5-3}.  
For $l=0$ the sum on the left side of \eqref{claim} has only one term, $s^{2h}$, so \eqref{claim} holds.  
For $l=1$, by \eqref{hyp1} the left side of \eqref{claim} is bounded by 
\[
  S_1(h,6) + 4s^{3h} \leq 3^{12} s^{2h+1}.
  \]
For $l \geq 2$, by \eqref{hyp1} and \eqref{hyp2} the left side of \eqref{claim} is bounded by
\begin{align} \label{claim1}
  S_1(h,2l+4) + S_2(h,2l+4) &\leq 2 \cdot 3^9 l^2(9s)^ls^{2h} + 3^9l^2(9s)^l s^{3h-1} \notag \\
   &\leq 3^{10}l^2(9s)^ls^{2h},
  \end{align}
which establishes the claim.  Finally, using the bound \eqref{type4a} for the right side of \eqref{type4}, we see 
from \eqref{type5-4} that 
\beq \label{type5-4a}
  |W_{h+1,5,4}(t,u)| \leq 3^{23}s^{3h+6}.
  \eeq
Combining \eqref{type5-3}, \eqref{type5-1a}, \eqref{type5-2a}, and \eqref{type5-4a} and using $s \leq s_8$ 
yields that for $h \geq 4$,
\beq \label{type5a}
  |W_{h+1,5}(t,u)| \leq 3^{26} s^{3h+4}.
  \eeq
Then from \eqref{sum7}, \eqref{type3}, \eqref{type1a}, \eqref{type4a}, \eqref{type2a} and \eqref{type5a}, again for $h \geq 4$,
\beq \label{Vhbound}
  |V_h(t,u)| \leq 3^{49}s^{3h+4},
  \eeq
which with \eqref{PQhsecond} proves the proposition for these $h$.

To deal with $h=1,2,3$ we first observe that the bounds \eqref{type3}, \eqref{claim}, \eqref{type5-3} and \eqref{PQhsecond}
remain valid for these $h$, so we may 
restrict attention to elementary perturbations and clusters whose 
support is contained in $\{x: |x|<8h\}$, and we need only establish 
\eqref{Vhbound}. Similarly, for $h=0$ we 
may assume all supports are in $\{x:|x|<8\}$.  Let $\C_h$ be the set of all clusters consisting only of 
perturbations with such supports.  For a cluster $X$ let $m(X)$ denote the number of horizontal plaquettes in $X$
in contact with the wall in the case $h \geq 1$, and the negative of the number not in contact with the 
wall, in the case $h=0$.  Let $\|X\|$ denote half the number of vertical plaquettes in $X$, so that $\vphi_u^{\rm T}(X) = 
a^T(X)t^{\|X\|}e^{um(X)}$.  It is easily seen that $|m(X)| \leq 4\|X\|$ for all
$0 \leq h \leq 3$ and $X \in \C_h$.  Note now that the absolute convergence of the cluster expansion, 
established in Lemma \ref{fkLemma}, means that provided $k \geq 8$, for all 
$0 \leq h \leq 3$ we have, using (\ref{i2}) and (\ref{8k}),
\beq
t\le t_1(8), u\le t^{1/2} \implies 
\sum_{X \in \C_h} |a^T(X)|t^{\|X\|} e^{um(X)} <% \infty
3000 s^2 \cdot 2\bigl(8(h\vee1)\bigr)^2
\eeq
and hence
\beqa 
  \sum_{X \in \C_h\atop\|X\| \geq 3h+4} |a^T(X)| t^{\|X\|} e^{um(X)}  
&\le&\Bigl(\frac{t}{t_1(8)}\Bigr)^{3h+4}
  \sum_{X \in \C_h\atop\|X\| \geq 3h+4} |a^T(X)| t_1(8)^{\|X\|}e^{um(X)} \cr
&<&\Bigl(\frac{t}{t_1(8)}\Bigr)^{3h+4}
 6000\, s(t_1(8))^2 \cdot \bigl(8(h\vee1)\bigr)^2\cr
&<&Ct^{3h+4},
 \label{tail} \eeqa
with $C=t_1(8)^{-11}\times7000\times24^2$.
This shows that \eqref{Vhbound} holds,
with a different constant, for $0 \leq h \leq 3$,
and completes the proof.
\end{proof}

For later use, we note that by combining \eqref{Phfirst}, \eqref{hyp1b}, \eqref{Qhfirst}, \eqref{hyp2b},
\eqref{Vhbound} and \eqref{tail}, we see that there exist numbers
$C,K_1$ such that for $h \geq 1, k \geq 8, t < t_1(k)$ and $u \leq \sqrt{t}$,
\begin{align} \label{fullsum}
  \sum_{{X:h \to 0} \atop {\Supp X\ni 0}} &\frac{1}{|\Supp X|} \left| \varphi_u^{\rm T}(X) \right| \notag \\
  &\leq (t^{2h} + 25 \cdot 3^{12}s^{2h+1})e^u + (t^{3h} + 25 \cdot 3^{12} s^{3h+1}) e^{2u} + Cs^{3h+4}
    \notag \\
  &\leq 2t^{2h} + K_1s^{2h+1}.
\end{align}
Here we use the fact that the bound \eqref{Vhbound} is obtained by adding a bound for the
sum of the absolute values of all terms corresponding to clusters $X:h \to 0$ 
to a similar bound for clusters $X:h+1 \to 0$. In the case $h=0$, in place of \eqref{fullsum}
we have using \eqref{type3} (modified for $h=0$ in the manner noted there) and \eqref{tail} that
for some $K_2$,
\beq \label{fullsum0}
  \sum_{X:\Supp X\ni 0} \frac{1}{|\Supp X|} \left| \varphi_u^{\rm T}(X) \right|
    \leq t^2e^{-u} + 2t^3e^{-2u} + Ct^4 + s^4
    \leq t^2 + K_2t^3.
  \eeq

We continue with the proof of Proposition \ref{Prop1}.
We first assume $u=t^2+O(t^4)$, in agreement with (\ref{Ass}) for
$n\ge2$. Then $e^{mu}-1-mt^2e^u = O(t^4)$ for $m=2,3,4$, and $R_h(t,u) = O(t^{3h+4})$ for $h = 1,2,3$.
Taking (\ref{PP})--\eqref{Fig3A1} into account, it follows using Proposition \ref{unifh} that for $h \geq 1$,
\begin{align}
f_k(&h+1)-f_k(h) \nonumber\\ 
&\quad=(t^{2h}+P_h(t))\Big(\exp\big(u+\ln(1-t^2)\big)-1\Big)-2t^{3h+3}+O(s^{3h+4}),
\label{P}
\end{align}
with the $O(s^{3h+4})$ uniform in $h \geq 1$ and $k \geq 8$, provided we restrict to $t<t_1(k)$.

If we assume that $n \geq 2$ and the first inequality in equation (\ref{Ass}) holds,
then as a consequence of equation (\ref{P}) we have
$\forall h\ge1$,
\beq
f_k(h+1)-f_k(h)\ge (2+a)t^{2h+n+3}-2t^{3h+3}+O(s^{3h+4})+O(s^{2h+n+4})
\label{P0}  
\eeq
and in particular,
\beq
f_k(n+1)-f_k(n) \ge at^{3n+3}+O(s^{3n+4}) \label{P1} 
\eeq
and
\beq
f_k(h+1)-f_k(h)\ge (2+a)t^{2h+n+3}+O(s^{2h+n+4}),\quad\forall h\ge n+1.
\label{P2}
\eeq
Here all $O(\cdot)$ terms are uniform in $h \geq 1, n \geq 1, k \geq 8$ provided we restrict to 
$t < t_1(k)/2$.  Summing these increments we obtain
\beq \label{P2a}
  f_k(n) \leq f_k(h)-at^{3n+3}+O(s^{3n+4}),\quad \hbox{uniformly in } 
  k \geq 8, n \geq 1, h\ge n+1.
\eeq
Restricting to $n \leq k$ allows us to replace $O(s^{3n+4})$ with $O(t^{3n+4})$, by \eqref{svst}, 
so \eqref{Pr1} is proved.

If instead we assume that the second inequality in equation (\ref{Ass}) holds with $n \geq 2$, then 
again from \eqref{P}, $\forall h\ge1$,
\beq
f_k(h+1)-f_k(h) 
\le (2-b)t^{2h+n+2}-2t^{3h+3}+O(s^{3h+4})+O(s^{2h+n+3}),
\eeq
and in particular, for $n\ge2$,
\beq
f_k(n)-f_k(n-1)\le -bt^{3n}+O(s^{3n+1}) \label{P3}
\eeq
and 
\beq
f_k(h)-f_k(h-1)\le -2t^{3h}+O(s^{3h+1}),\quad\forall\ 2\le h\le n-1, \label{P4}
\eeq
again with uniformity in $h \geq 1, n \geq 2, k \geq 8$ for $O(\cdot)$ terms provided we restrict to 
$t < t_1(k)$.  From \eqref{10} we also have
\beq \label{ht1}
  f_k(1) - f_k(0) \leq -2t^3 + O(t^4),
  \eeq
uniformly in $k \geq 8$ for these same $t$ values.  Summing the increments we obtain
\beq \label{P22}
  f_k(n) \leq f_k(h)-2t^{3h+3}+O(s^{3h+4}),\quad \hbox{uniformly in } 
  k \geq 8, n \geq 2, 0 \leq h \leq n-2
\eeq
and
\beq \label{P2b}
  f_k(n) \leq f_k(n-1)-bt^{3n}+O(s^{3n+1}),\quad \hbox{uniformly in } 
  k \geq 8, n \geq 2.
\eeq
As with \eqref{P2a}, restricting to $n \leq k$ makes \eqref{Pr2} a consequence of 
\eqref{P22} and \eqref{P2b}.
 
In the $n=1$ case, we have $u=t^2+O(t^3)$, $e^{mu}-1-mt^2e^u = m(u+\log(1-t^2))+O(t^4)$ for 
$m=2,3,4$, and $R_h(t,u) = O(t^{3h+4})$ for $h = 1,2,3$.  Therefore by \eqref{PP}--\eqref{Fig3A1},
for $h \geq 1$, in place of \eqref{P} we have
\begin{align}
f_k(&h+1)-f_k(h) \nonumber\\ 
&\quad=(t^{2h}+O(t^{2h+1}))\left( u+\ln(1-t^2) \right) -2t^{3h+3}+O(s^{3h+4}),
\label{Pn1}
\end{align}
with the $O(\cdot)$ uniform in $h \geq 2$ and $k \geq 8$.
Assuming the first inequality in \eqref{Ass} it follows that
\beq \label{n1}
f_k(h+1)-f_k(h)\ge (2+a)t^{2h+4}+O(t^{2h+5}),
\qquad\forall\ h\ge2,
\eeq
and 
\beq \label{n1a}
f_k(2)-f_k(1)\ge at^6+O(t^7).
\eeq
Now \eqref{Pr1} follows from \eqref{n1} and \eqref{n1a}.  From \eqref{10} we obtain 
\[
  f_k(1) - f_k(0) \leq -bt^3 + O(t^4),
  \]
which proves \eqref{Pr2}.

It remains only to consider the case $n=0$.  Here the assumption is $-\ln(1-t^2) + (2+a)t^3 \leq u < \sqrt{t}$.
Let us first assume that $u$ is at most of the order $t^2$. Then 
\beq 
f_k(1)-f_k(0)= u-t^2-2t^3+O(t^4),  
\eeq
which implies that 
\[
  f_k(1) - f_k(0) \geq at^3+O(t^4).
  \]
But in fact, $u = O(t^2)$ is stronger than necessary. 
We notice that, from (\ref{10}) and Proposition \ref{unifh}, 
uniformly in $u<t^{1/2}$,
\beq
f_k(1)-f_k(0)= u(1-2t^2) - t^2 +O(t^3) + O(t^2u^3),
\eeq 
which is greater than $u/2$ provided $3t^2 < u < t^{1/2}$ and $t$ is sufficiently small.
This ends the proof of Proposition \ref{Prop1}.

%\newpage

\section{Proof of Theorem \ref{Th1}}\label{ProofTh}

We consider the model in a rectangular box $\Lambda$, 
with sides parallel to the axes, under the constant boundary 
condition $\olp_x=n$, for any given integer $n\ge0$. 
First, we are going to rewrite the partition function $\Xi(\Lambda,n)$, 
expression (\ref{Xi1}), using some grouping of
the nonelementary cylinders.
This will allow us to describe the model as a polymer system.
The next task will be to study the cluster expansion that 
can be obtained by this method.

For this purpose, given $\Gamma\in \mathcal{C}(\Lambda,n)$, we define 
$\Gamma_k^{L}\subset \Gamma$ (the subset of {\it large\/} cylinders), 
as the set of cylinders obtained by removing
from $\Gamma$ all those cylinders which are elementary. 
In all this section, we take $k\ge \max(2n,8)$. 
$\Gamma_k^{L}$ is still a (possibly empty) compatible set of cylinders,
because the operation of removing from $\Gamma$
a cylinder together with all the other cylinders contained in it
does not spoil the compatibility. We set
$$
\C_k^{L}(\Lambda,n)=\{\Gamma\in\mathcal{C}(\Lambda,n) : \Gamma=\Gamma_k^{L} \}.
$$
Then, according to (\ref{Xi1}), we can write
\beqa 
\Xi(\Lambda,n)&=&e^{u\delta(n)|\Lambda|}\sum_{\Ga\in\C(\La,n)}
\prod_{\ga\in\Ga}\vphi(\ga) \nonumber \\
&=&e^{u\delta(n)|\Lambda|}\sum_{\Ga\in\C^{L}_k(\La,n)}
\prod_{\ga\in\Ga}\vphi(\ga)
\sum_{\Ga'\in\C_k^{el}(\La,n)\atop\Ga\cup\Ga'\in\C(\La,n)}
\prod_{\ga'\in\Ga'}\vphi(\ga') .
\label{Xi2}\eeqa

We define a {\it contour} as a compatible set 
$\Gamma\in\C_k^{L}(\Lambda,n)$ of nonelementary cylinders, such that:

\smallskip

{\parindent=0pt
\begin{tabular}{ll}
$(1)$ there exists a unique cylinder which is external in $\Gamma$, \\
$(2)$ if $\gamma\in\Gamma$ and $I(\gamma)=n$, then there is no other 
$\gamma'\in\Gamma$ such that ${\olg}'\subset{\olg}$.
\end{tabular} }

\smallskip

Condition (2) is equivalent to

\smallskip

{\parindent=6pt
 $(3)$ if $\gamma\in\Gamma$ is not external, then $E(\gamma) \neq 
E(\Gamma) = n$.}

\smallskip

Any given contour $\Gamma$ can be written as a set 
$\Gamma=\{\gamma^{ext},\gamma_i,\gamma^{int}_j\}$ 
where $\gamma^{ext}$ is the unique external cylinder in $\Ga$, 
$\gamma^{int}_j$ are the cylinders that satisfy $I(\ga_j^{int})=n$, 
and $\gamma_i$ are the remaining cylinders in $\Gamma$.

With these notations, 
we define 
\beqa
{\Supp}(\Gamma)&=& {\olg}^{ext}\setminus
\big(\cup_{j}{\olg}^{int}_j\big),\label{suppp}\\  
{\Supp}^i(\Gamma)&=& {\olg}_i\setminus
\big(\cup_{{\olg}\subset{\olg}_i}{\olg}\big), 
\label{suppi} \\ 
{\Supp}^{ext}(\Gamma)&=& {\olg}^{ext}\setminus
\big(\cup_{\gamma\ne\gamma^{ext}}{\olg}\big). \label{suppe}
\eeqa 

The set $\Supp(\Gamma)$ is called the {\it support} of $\Gamma$. 
The sets (\ref{suppi}), associated to the cylinders $\gamma_i$, 
together with the set (\ref{suppe}), associated to $\gamma^{ext}$, 
are mutually disjoint subsets of $\Supp(\Gamma)$ 
and their union coincides with this support. 

As an explanation for these definitions let us notice that
$E(\Gamma)=E(\gamma^{ext})=I(\gamma_j^{int})=n$. 
On the other hand, if we consider the interface  
that contains only the contour $\Gamma$,  
then the associated configuration $\phi_\Lambda$ 
has $\phi_x=n$ if $x\not\in\Supp(\Gamma)$ and $\phi_x\ne n$ if
$x\in\Supp(\Gamma)$.
It takes constant values in $\Supp^{ext}(\Gamma)$, 
namely $\phi_x=I(\gamma^{ext})$, and in each $\Supp^i(\Gamma)$, 
where we have $\phi_x=I(\gamma_i)$.

We say that 
two contours $\Ga$ and $\Ga'$ are {\it compatible} if their supports do
not intersect and $\Ga\cup\Ga'\in\C^{L}_k(\La,n)$.
This last condition enters only when the boundaries of the 
supports of $\Ga$ and $\Ga'$ have a non-empty intersection,
otherwise it is already satisfied.

Notice that, from definition (\ref{suppp}),
$\partial\,\Supp(\Gamma)={\tga}^{ext}\cup
\big(\cup_j{\tga}^{int}_j)$, so the compatibility
condition $\Ga\cup\Ga'\in\C^{L}_k(\La,n)$ concerns only 
these particular cylinders.
Since, among the compatibility conditions (see Section \ref{CM}),
condition (3) is already satisfied,
they have only to agree with the signs 
according to conditions (1) and (2).
More precisely, any two cylinders $\gamma$ and $\gamma'$,
not necessarily compatible, such that 
$\tga\cap\tga'\ne\emptyset$,
satisfy the sign condition if,
\beq
S(\gamma)=S(\gamma'),\hbox{ if }\olg\subset\olg',
\hbox{ and }
S(\gamma)=-S(\gamma'),\hbox{ if }\olg\subset
\Lambda\setminus\olg'.
\label{sign}\eeq

This leads to the following equivalent definition.

\medskip

Two contours $\Ga$ and $\Ga'$, such  that $E(\Ga)=E(\Ga')$, 
are {\it compatible} if their 
supports do not intersect and if, in addition, we have: 
$S(\gamma^{ext})=-S(\gamma^{'ext})$, if
${\tga}^{ext}\cap{\tga}^{'ext}\ne\emptyset$, 
$S(\gamma^{ext})=S(\gamma^{'int}_j)$, if
${\tga}^{ext}\cap{\tga}^{'int}_j\ne\emptyset$ 
and
$S(\gamma^{'ext})=S(\gamma^{int}_j)$, if
${\tga}^{'ext}\cap{\tga}^{int}_j\ne\emptyset$. 

\medskip

The set of contours $\{\Gamma_i\}$ is a {\it compatible set},
if any two contours in it are compatible.
Let $\Gamma\in\C_k^{L}(\Lambda,n)$. 
Then it is possible to write $\Gamma$ as the disjoint union
\beq
\Gamma = \Gamma_1 \cup \ldots \cup \Gamma_r
\label{decompo}\eeq
in such a way that, for each $i=1,\dots,r$, $\Gamma_i$ is a contour.
The decomposition (\ref{decompo}) is unique.

\medskip

At this point we can write the first sum in expression (\ref{Xi2})
as a sum over compatible sets of contours $\{\Gamma_i\}$.
The sum over the elementary contours in (\ref{Xi2}) decomposes
then into a product of sums associated to the different regions
of $\Lambda$ determined by the contours, 
and gives rise to some partition functions that we are going
to discuss.

Some remarks on the geometry of contours will be helpful.
We notice that the interiors of the cylinders in $\Gamma$ are partially
ordered by inclusion.
The maximal one in the sense of this partial order is $\gamma^{ext}$.
The $\gamma^{int}_j$ are minimal elements,
but in general not the only ones.
Given a cylinder in $\Gamma$, denoted $\gamma_0$,
that is not a minimal element, there are other cylinders
$\gamma_r\in\Gamma$, different from $\gamma_0$, 
that are maximal elements in
the set of cylinders contained in $\olg_0$.
This is the situation specified by our notation $\gamma_r\prec\gamma_0$.

%\begin{Rem}aaa\end{Rem}

The cylinder $\gamma_0$ can be either $\gamma^{ext}$ or 
one of the $\gamma_i$.
Let $\Lambda'$ be one of the sets in definitions
(\ref{suppi}) or (\ref{suppe}), accordingly. 
Then we may write 
$\Lambda'=\olg_0\setminus\big(\cup_r\olg_r\big)$,
and we have
$\partial\Lambda' \subset \tga_0\cup
\big(\cup_r\tga_r\big)$, where the $\gamma_r$
are, as before, the cylinders such that 
$\gamma_r\prec\gamma_0$.
We denote by $(\Lambda',n')$ the set of elementary
perturbations $\omega$ belonging to $\mathcal{C}^{el}_k(\Lambda',n')$,
hence with support contained in $\Lambda'$
and satisfying $E(\omega)=n'$.
We denote by $(\Lambda',n')^*$ the subset of $(\Lambda',n')$ 
consisting of the elementary perturbations $\om$ for
which in addition $\ga_\om^{ext}$ satisfies the sign condition (\ref{sign})
with respect to all the cylinders $\gamma_0$ and $\gamma_r$
for which the base perimeters have a non-empty intersection
with $\tga^{ext}_\omega$.

%\medskip

%\centerline{$S(\gamma^{ext}_\omega)=S(\gamma_0)$, if 
%${\tga}^{ext}_\omega\cap{\tga}_0\ne\emptyset$,
%and 
%$S(\gamma^{ext}_\omega)=-S(\gamma_r)$, if 
%${\tga}^{ext}_\omega\cap{\tga_r}\ne\emptyset$.} 

%\medskip

We introduce the partition function 
\beq
Z_k^*(\Lambda',n')=
\sum_{\{\omega_j\}\subset(\Lambda',n')^*}
\prod_j \varphi(\omega_j) ,
\label{Zeta4}\eeq
where the sum runs over all compatible sets 
of elementary perturbations $\{\omega_j\}$,
whose elements belong to $(\Lambda',n')^*$.
Here $k$ is the value used
in the definition (\ref{elementary}) of the elementary cylinders.
We notice that if $n'=I(\gamma_0)$, 
%which is the ``correct'' level
%of $\Lambda'$ according to definitions (\ref{suppi}) and (\ref{suppe}),
then $\Gamma\cup\{\omega_j\}$ is a compatible set of cylinders, i.e.,
$\Gamma\cup\{\omega_j\}\in\mathcal{C}(\Lambda,n)$.
But the partition functions (\ref{Zeta4}) are well defined,
and will be used, also when $n'\ne I(\gamma_0)$.
%in the case in which $n'$ is not the correct level.

In the case $\Lambda'=\Supp(\Gamma)$, definition (\ref{suppp}),
we also introduce the partition function $Z_k^*(\Lambda',n')$ 
by the same formula (\ref{Zeta4}).
Since
$\partial\Supp(\Gamma)=
\tga^{ext}\cup\big(\cup_j\tga_j^{int})$, 
now the cylinders $\gamma^{ext}$ and $\gamma_j^{int}$ 
play the role of $\gamma_0$ and $\gamma_r$,
in the sign condition, when $\tga^{ext}_\omega$
has a non-empty intersection with $\partial\Supp(\Gamma)$.
%Notice that in this case, $\Lambda'$ has not 
%a determinate level, in general, and 
%$\gamma^{ext}$ and $\gamma_j^{int}$ 
%are not necessarily compatible cylinders.

Given a compatible set $\{\Gamma_i\}$ of contours,
we introduce also the partition function $Z_k^*(\Lambda',n)$
associated to the complementary region
$\Lambda'=\Lambda\setminus\big(\cup_i\Supp(\Gamma_i)\big)$,
again by equation (\ref{Zeta4}).
In this case
the sign condition for an elementary perturbation $\om$ has to be 
verified with respect to all
%satisfied by all elementary perturbations such that 
%$\tga^{ext}_\omega$ has a non-empty intersection  with one of 
the boundaries $\partial\Supp(\Gamma_i)$.
%More precisely with the cylinders whose perimeters
%form the boundaries of the $\Gamma_i$.
%Because $n$ is the correct level of $\Lambda'$,
%it follows that 
Since $E(\om)=n$, the set $\omega\cup(\cup_i\Gamma_i)$ is a compatible set of 
cylinders, i.e., $\omega\cup(\cup_i\Gamma_i)\in\mathcal{C}(\Lambda,n)$.

\medskip

Now, for each contour $\Gamma=\{\gamma^{ext},\gamma_i, \gamma^{int}_j\}$ in
$\C_k^{L}(\Lambda,n)$, we define the statistical weight  
\beqa
\varphi(\Gamma)&=& \varphi(\gamma^{ext})
\Big(\prod_i \varphi(\gamma_i)\Big) 
\Big(\prod_j \varphi(\gamma^{int}_j)\Big) \nonumber \\
&\times&
\frac{Z_k^*\big(\Supp^{ext}(\Gamma),I(\gamma^{ext})\big)
\prod_i Z_k^*\big(\Supp^i(\Gamma),I(\gamma_i)\big)}
{Z_k^*\big(\Supp(\Gamma),n\big)} .
\label{weight}\eeqa
For the convergence of \eqref{contourconv} in Lemma \ref{LemLast} below, it is essential
to have the denominator present in \eqref{weight}, so that the energy-entropy benefit of 
choosing any given interface heights for $\Supp^{ext}(\Ga)$ and the $\Supp^i(\Ga)$
is measured relative to the benefit of
leaving these regions at height $n$.

\begin{Lem} \label{LXi5}
We have 
\beqa 
&&\Xi(\Lambda,n)=e^{u\delta(n)|\Lambda|} \nonumber \\
&&\times\sum_{\{\Gamma_i\}}
Z_k^*\Big(\Lambda\setminus(\cup_i\Supp(\Gamma_i)),n\Big)
\prod_i \Big(\varphi(\Gamma_i)Z_k^*\big(\Supp(\Gamma_i),n\big)\Big).
\label{Xi5}\eeqa
The sum runs over all compatible sets $\{\Gamma_i\}$ of contours
contained in $\Lambda$ such that $E(\Gamma_i)=n$, for all $i$. 

\end{Lem}
\begin{proof}
The proof follows from formula (\ref{Xi2}) and 
the above definitions and remarks. 
This lemma corresponds to lemma 2.4 of ref. \cite{DM}.
\end{proof}

Note that the height $n$ appears in the partition function 
$Z_k^*\big(\Supp(\Gamma_i),n\big)$ on the right side of 
\eqref{Xi5}, rather than the actual height $I(\ga_i)$ of the base 
of $\Gamma_i$, because the factor $\varphi(\Gamma_i)$ contains the
ratio of partition functions that appears in \eqref{weight}.

We still need the following definition of compatibility.
We say that an elementary perturbation $\omega$,
with external cylinder denoted $\gamma^{ext}_\omega$ and external level
$E(\omega)=n$,
is {\it compatible} with the contour
$\Gamma=(\gamma^{ext},\gamma_i,\gamma^{int}_j)$, of external level
$E(\Gamma)=n$, if either
\begin{itemize}
\item[(i)] $\olg^{ext}_\omega\subset\Supp(\Gamma)$,
and then:
$S(\gamma^{ext}_\omega)=S(\gamma^{ext})$ if 
${\tga}^{ext}_\omega\cap{\tga}^{ext}\ne\emptyset$, 
$S(\gamma^{ext}_\omega)=-S(\gamma^{int}_j)$ if 
${\tga}^{ext}_\omega\cap{\tga^{int}_j}\ne\emptyset$, or,
\item[(ii)]
$\olg^{ext}_\omega\subset\Lambda\setminus\Supp(\Gamma)$,
and then:
$S(\gamma^{ext}_\omega)=-S(\gamma^{ext})$ if
${\tga}^{ext}_\omega\cap{\tga}^{ext}\ne\emptyset,$ \\ 
and 
$S(\gamma^{ext}_\omega)=S(\gamma^{int}_j)$ if
${\tga}^{ext}_\omega\cap{\tga^{int}_j}\ne\emptyset$.
\end{itemize}

\smallskip

We say that a set $\{\Gamma_i,\omega_j\}$ is a 
{\it compatible set of contours and elementary perturbations} if 
any two contours in the set are compatible, as well as
any two elementary perturbations in it and,
also, any elementary perturbation in the set is compatible with
any contour in the set.

\begin{Lem}
We have 
\beq 
\Xi(\Lambda,n)=e^{u\delta(n)|\Lambda|}\sum_{\{\Ga_i,\omega_j\}}
\prod_i\vphi(\Gamma_i)\prod_j\varphi(\omega_j). 
\label{polymer}\eeq
The sum runs over all compatible sets $\{\Gamma_i,\omega_j\}$ 
of contours and elementary perturbations
contained in $\Lambda$ such that 
$E(\Gamma_i)=E(\omega_j)=n$, for all $i,j$. 
\label{contourpert}
\end{Lem}
\begin{proof}
The lemma follows from the above definition of compatibility
between contours and elementary perturbations. 
The set of contours $\{\Gamma_i\}$ being fixed,
the sum over all compatible sets $\{\omega_j\}$ 
of elementary perturbations, such that also
${\{\Ga_i,\omega_j\}}$ is a compatible set, gives
the product of the partition functions in equation (\ref{Xi5}).
\end{proof}

As a consequence of the lemma we see that the system of contours 
and elementary perturbations has become a polymer model.

Our next task is to find a suitable estimate for the 
statistical weight of a contour.
To this end we first derive, from equation (\ref{weight}),
a more convenient expression for this weight.
We introduce the partition functions
\beq \label{tildepart}
{\tilde Z}_k(\Lambda',n')=e^{u\delta(n')|\Lambda'|}Z^*_k(\Lambda',n')
\eeq
with the ``correct'' factor $e^{u\delta(n')|\Lambda'|}$, where 
$Z^*_k(\Lambda',n')$ is from \eqref{Zeta4}.
They differ from the $Z_k(\Lambda',n')$
considered before, definition (\ref{Zeta}) and (\ref{Zeta3}), by the
compatibility condition on the elementary
perturbations that touch the boundary 
$\partial\Lambda'$, assumed in the definition of $Z^*_k$.
Only boundary terms are therefore affected by this condition.

For a contour $\Gamma=\{\gamma^{ext},\gamma_i,\gamma^{int}_j\}$
we define $\|\Ga\|$ to be half the number of vertical plaquettes in $\Ga$, that is,
\[
  2\|\Ga\|= |\ti\ga^{ext}|L(\ga^{ext})+\sum_i|\ti\ga_i|L(\ga_i)
    +\sum_j|\ti\ga_j^{int}|L(\ga_j^{int}).
  \]

\begin{Lem}
The statistical weight of a contour can also be written as
\beqa
\varphi(\Ga)&=& t^{\|\Ga\|}
{{{\tilde Z}_k\big(\Supp^{ext}(\Ga),I(\ga^{ext})\big)
\prod_i 
{\tilde Z}_k\big(\Supp^i(\Ga),I(\ga_i)\big)}
\over{{\tilde Z}_k\big(\Supp(\Ga),n\big)}}.
\label{weight2}
\eeqa
\label{weightTwo}
\end{Lem}

\begin{proof}
In expression (\ref{weight}) replace each $\vphi(\gamma)$ by its value
(\ref{phicyl}). 
Then the claim in the Lemma follows, except for a factor $e^{uG}$ with
\beqa
G&=&(\del(I(\ga^{ext}))-\del(n))|\olg^{ext}|
+\sum_i(\del(I(\ga_i))-\del(E(\ga_i)))|\olg_i|\cr
&+&\sum_j(\del(n)-\del(E(\ga_j^{int})))|\olg_j^{int}|
-\del(I(\ga^{ext}))|{\Supp}^{ext}(\Ga)|\cr
&-&\sum_i\del(I(\ga_i))|{\Supp}^i(\Ga)|
+\del(n)|\Supp(\Ga)| .
\eeqa
Then the identity $G=0$ will complete the proof.
First, we see that the coefficient of $\del(n)$ is zero because
$|\Supp(\Ga)|=|\olg^{ext}|-\sum_j|\olg_j^{int}|$.
Next, consider the terms where the cylinders $\ga^{ext}$, 
$\ga_i\prec\ga^{ext}$ and $\ga_j^{int}\prec\ga^{ext}$,
appear.
We have then $I(\ga^{ext})=E(\ga_i)=E(\ga_j^{int}))$ and
$$
|\olg^{ext}|-\sum_{i}|\olg_i|
-\sum_{j}|\olg_j^{int}| =|{\Supp}^{ext}(\Ga)| . 
$$
This implies that, in $G$, the terms concerning $\ga^{ext}$
and the considered $\ga_j^{int}$ cancel, and, for each of the 
considered $\ga_i$, only the term $\del(I(\ga_i))|\olg_i|$ remains.
In the next step,
we apply the same arguments to each of these $\ga_i$.
After a certain number of similar steps we arrive at the
minimal $\ga_i$ that do not contain any other cylinder of the
contour, and thus we obtain $G=0$, completing the proof of 
Lemma~\ref{weightTwo}.
\end{proof}

\medskip

In order to estimate the statistical weight of a contour, 
given by equation (\ref{weight2}), we first observe that
\beqa
{{{\tilde Z}_k\big(\Supp^{ext}(\Gamma),I(\gamma^{ext})\big)\prod_i 
{\tilde Z}_k\big(\Supp_i(\Gamma),I(\gamma_i)\big)}
\over{{\tilde Z}_k\big(\Supp(\Gamma),n\big)}}&\le& \nonumber \\
{{{\tilde Z}_k\big(\Supp^{ext}(\Gamma),I(\gamma^{ext})\big)}
\over{{\tilde Z}_k\big(\Supp^{ext}(\Gamma),n\big)}}
\prod_i {{{\tilde Z}_k\big(\Supp_i(\Gamma),I(\gamma_i)\big)}
\over{{\tilde Z}_k\big(\Supp_i(\Gamma),n\big)}} .&& 
\label{QQ}\eeqa 
We have then, for this purpose, to estimate the quotient
$$
\frac{ {\tilde Z}_k(\Lambda',h) }{ {\tilde Z}_k(\Lambda',n) } \quad\hbox{when } 
h\ne n .
$$
From the definition of a contour we always have the level
$n=E(\Gamma)$ in the denominator of (\ref{QQ}), while 
only levels $h\ne n$ appear in the numerator.
By the definition of ${\tilde Z}_k$ we may write
\beq
\ln \frac{ {\tilde Z}_k(\Lambda',h) }{ {\tilde Z}_k(\Lambda',n) } =
-|\Lambda'|\big(f_k(h)-f_k(n)\big)
+\hbox{``boundary terms''} .
\label{bound}\eeq

Now, assume that inequalities (\ref{Ass}) in Proposition \ref{Prop1}
are satisfied, with $a=b=\epsilon$ in the hypotheses, 
so that the restricted ensemble at level $n$ is a dominant state.
%%We'll denote $t_2(k,n,\ep,\ep)$ simply as $t_2(k,n,\ep)$.

Concerning the ``boundary terms'' in the right hand side of (\ref{bound}), 
let $\mD(\La',h)$ denote the set of all clusters $X$
satisfying $(\Supp X) \cap \La' \neq \phi$ for which either 
(i) $(\Supp X) \cap (\La')^c \neq \phi$,
or (ii) $\Supp X \subset \La'$ and the sign condition is violated by 
some $\om \in X$.  
From \eqref{lnZ} and \eqref{fk}, for $h \geq 0$,
\begin{align} \label{boundary4}
  \ln {\tilde Z}_k(\Lambda',h) + |\Lambda'| f_k(h) &= -\sum_{x \in \La'}\ \sum_{X \in \mD(\La',h):\Supp X \ni x}
    \frac{1}{ |\La'\cap\Supp X| } \varphi_u^{\rm T}(X).
  \end{align}
Let $h \geq 1$ and recall that $s=te^{t^{1/4}}$.  For $x \in \La'$ with $d(x,(\La')^c) \geq 8h$ we use \eqref{type3} (excluding the first inequality there)
and discard the factor $1/|\La'\cap\Supp X|$:
assuming $t<t_1(k)$,
\begin{align} \label{nearbdry}
  \sum_{X \in \mD(\La',h):\Supp X \ni x} \frac{1}{ |\La'\cap\Supp  X| } \left| \varphi_u^{\rm T}(X) \right|
    &\leq \sum_{y \in \partial \La'}\ \sum_{{X:h \to 0} \atop {x,y \in \Supp X}} \left| \varphi_u^{\rm T}(X) \right| 
    \leq s^{3h+4}.
\end{align}
For $x \in \La'$ with $d(x,(\La')^c) < 8h$ we use \eqref{fullsum}: assuming $s<1/K_1$,
\beq \label{farbdry}
  \sum_{X \in \mD(\La',h):\Supp X \ni x} \frac{1}{ |\La'\cap\Supp  X| } \left| \varphi_u^{\rm T}(X) \right|
    \leq 2t^{2h} + K_1s^{2h+1} \leq 3s^{2h}.
\eeq
Combining these we obtain that for $h \geq 1$,
\beq \label{bound5}
  \left| \ln {\tilde Z}_k(\Lambda',h) + |\Lambda'| f_k(h) \right| \leq 384h^2 |\partial\La'| s^{2h} + |\La'|s^{3h+4}.
  \eeq
  
For $h=0$, \eqref{type3} is valid with $8h-2$ replaced by 8, as noted there, so \eqref{nearbdry} 
is valid for $x$ with $d(x,(\La')^c) \geq 9$.  For $x$ with $d(x,(\La')^c) \leq 8$ we 
use \eqref{fullsum0}: assuming $t<1/K_2$,
\beq \label{nearbdry0}
   \sum_{X \in \mD(\La',h):\Supp X \ni x} \frac{1}{ |\La'\cap\Supp  X| } \left| \varphi_u^{\rm T}(X) \right|
    \leq t^2 + K_2t^3 \leq 2t^2.
\eeq
Hence in place of \eqref{bound5} we have for $h=0$:
\beq \label{bound50}
  \left| \ln {\tilde Z}_k(\Lambda',h) + |\Lambda'| f_k(h) \right| \leq 384 |\partial\La'| t^2 + |\La'|s^4.
\eeq
Write $h_1$ for $\max(h,1)$ and $n_1$ for $\max(n,1)$.  
By \eqref{Pr1} and \eqref{Pr2} in Proposition \ref{Prop1}, for some $K_3$ 
we have 
\begin{equation} \label{freegap}
  f_k(h) - f_k(n) \geq \begin{cases} 2t^{3h+3} - K_3t^{3h+4} &\text{for } h \leq n-2, \\
  \ep t^{3h+3} - K_3t^{3h+4} &\text{for } h =n-1, \\
  \ep t^{3n+3} - K_3t^{3n+4} &\text{for } h \geq n+1. \end{cases}
  \end{equation}
Provided $t \leq \min(t_1(k),(\ep\wedge 1)/2(K_3+2))$, this and 
\eqref{bound5}, \eqref{bound50} show that for some $K_4$, for $h \leq n-2$, 
\begin{align} \label{ratio}
   \log \frac{ {\tilde Z}_k(\Lambda',h) }{ {\tilde Z}_k(\Lambda',n) } &\leq 
     -|\Lambda'|\big( f_k(h) - f_k(n) \big) + K_4h_1^2|\partial \La'| s^{2h_1} + 2|\La'|s^{3h+4} \notag \\
   &\leq -|\Lambda'| t^{3h+3} + K_4h_1^2|\partial \La'|s^{2h_1},
\end{align}
and similarly, considering $h=n-1$, 
\begin{align} \label{ratio2}
   \log \frac{ {\tilde Z}_k(\Lambda',n-1) }{ {\tilde Z}_k(\Lambda',n) } &\leq 
     -|\Lambda'|\big( f_k(h) - f_k(n) \big) + K_4h_1^2|\partial \La'| s^{2h_1} + 2|\La'|s^{3h+4} \notag \\
   &\leq -|\Lambda'| \frac{\epsilon}{2} t^{3h+3} + K_4h_1^2|\partial \La'|s^{2h_1},
\end{align}
while for $h \geq n+1$, 
\begin{align} \label{ratio3}
   \log \frac{ {\tilde Z}_k(\Lambda',h) }{ {\tilde Z}_k(\Lambda',n) } &\leq 
     -|\Lambda'|\big( f_k(h) - f_k(n) \big) + K_4n_1^2|\partial \La'| s^{2n_1} + 2|\La'|s^{3n+4} \notag \\
   &\leq -|\Lambda'| \frac{\epsilon}{2} t^{3n+3} + K_4n_1^2|\partial \La'|s^{2n_1}.
\end{align}
Here we have used the fact that $s^{3h+4} \leq 2t^{3h+4}$ for $t \leq t_1(k)$ and $h \leq n \leq k/2$.

Thus we obtain the following lemma.

\begin{Lem}
Assume that inequalities (\ref{Ass1}) in Theorem \ref{Th1}
%are satisfied, then, there exists a value $t_1(n,\epsilon)$
%such that, for $t\le t_1(n,\epsilon)$ and $h\ne n$, we have
are satisfied for some $\ep>0$.  Then there exist $K_5,K_6$ such that
for $n,h\ge0$, $h\ne n$, $k\ge\max(8,2n)$ 
and $t\le t_2(k,\epsilon) \equiv \min(t_1(k),K_5(\ep \wedge 1))$, we have
\beq
\frac{ {\tilde Z}_k(\Lambda',h) }{ {\tilde Z}_k(\Lambda',n) } \le \exp\Big(
-|\Lambda'| \frac{\epsilon}{2} t^{3n+3}+|\partial\Lambda'| K_6s^{2}\Big) ,
\label{bound2}\eeq
where $s=te^{t^{1/4}}$.
\label{Lem3}
\end{Lem}

This lemma is the analog of Lemma~2.5 in ref.~\cite{DM}, 
though the definitions are not exactly the same.
The factor $\frac{\epsilon}{2} t^{3n+3}$ represents the cost per unit area
for the region $\La'$ to be at a suboptimal height $h \neq n$.

%\newpage
As a consequence of Lemma \ref{weightTwo}, \eqref{QQ} and Lemma \ref{Lem3} 
we obtain the following estimate
on the statistical weight of a contour:
\begin{align}
\varphi(\Ga) &\le t^{ \|\Ga\| } \exp\Bigg( |\partial\,\Supp^{ext}(\Gamma)|K_6s^2 + 
  \sum_i|\partial\,\Supp^{i}(\Gamma)|K_6s^2 \notag \\
&\qquad\qquad\qquad - \bigg(|\Supp^{ext}(\Gamma)| \frac{\ep}{2} t^{3n+3} 
+\sum_i|\Supp^{i}(\Gamma)|\frac{\ep}{2} t^{3n+3}\bigg) \Bigg) \notag\\
&\leq t^{\|\Ga\|} \exp\left(\Big(|\tga^{ext}|+2\sum_i|\tga_i|
+\sum_j|\tga_j^{int}|\Big)K_6s^2 - |\Supp(\Gamma)| \frac{\ep}{2} t^{3n+3} \right).
\label{phigamma}\end{align}

We now consider the convergence of the 
cluster expansion for the logarithm of $\Xi(\Lambda,n)$, 
as written in expression (\ref{polymer}).
That is, we consider the convergence of the following expansion of the surface tension:
\beqa
\tau^{\scriptscriptstyle WB}-2(J_{WA}+J_{AB}) &=& -\lim_{\Lambda\to{\bf Z}^2}\frac{1}{\beta|\Lambda|}
\ln \Xi(\Lambda,n) \nonumber \\ 
&=& -\frac{u}{\beta}\delta(n)-\frac{1}{\beta}\sum_{X\ni 0}
\frac{1}{|\Supp X|}\varphi_u^{\rm T}(X). 
\label{expansion}\eeqa
The sum runs over the clusters of contours $\Gamma$ 
and elementary perturbations $\omega$ such that 
$E(\Gamma)=E(\omega)=n$
and we use the notation 
$$
\Supp X=(\cup_{\Gamma: X(\Gamma)\ge1}\Supp(\Gamma))\cup
(\cup_{\omega: X(\omega)\ge1}\Supp\omega) .
$$
We will apply the Convergence Theorem in Section \ref{ProofProp} with
\beq
\mu(\Ga) = s^{\|\Ga\|}
e^{|\Supp\Ga|(16t)^{3k+4}} 
{{{\tilde Z}_k\big(\Supp^{ext}(\Ga),I(\ga^{ext})\big)
\prod_i 
{\tilde Z}_k\big(\Supp^i(\Ga),I(\ga_i)\big)}
\over{{\tilde Z}_k\big(\Supp(\Ga),n\big)}}.
\label{weight2mu}
\eeq
%The bound (\ref{phigamma}) is then changed into
Then (\ref{QQ}) and Lemma \ref{Lem3} yield
\beq
\mu(\Ga) \le  s^{\|\Ga\|} \exp\left(
\Big(|\tga^{ext}|+2\sum_i|\tga_i|+\sum_j|\tga_j^{int}|\Big)K_6s^2
-|\Supp(\Gamma)| \frac{\ep}{4} t^{3n+3} \right),
\label{phigammamu}\eeq
provided 
\beq \label{tepsilon}
(16t)^{3k+4} < \frac{\ep}{4}t^{3n+3},
\eeq
which is satisfied as soon as $k\ge\max(8,2n)$ and  
$t\le \ep^{1/12}/2000$. 

The statistical weight (\ref{phigamma}), (\ref{phigammamu}) is given in
terms of factors $t^{(1/2)|\tga|L(\gamma)}$ associated to the cylinders 
of the contour, but in contrast to the situation for clusters, 
the set of perimeters of these cylinders is  
not connected. 
The following two lemmas show how one can nevertheless include  
the set of perimeters in a connected structure, in such a way that 
the weights become summable, as needed  
in condition (\ref{i1}).

\begin{Lem}
\label{treeLem}
Let $k$ and $n$ be integers with $n \geq 0$ and 
$k \geq \max(2n,8)$.
Consider rooted trees $T$, oriented from the root $r(T)$ outwards.
Let each vertex $v$ carry an integer label $n_v$, and each edge $g$
an integer label $p_g$. Denote by $d(v)=d(v,T)$ the number of edges starting
from vertex $v$, and let these edges be ordered, 1 through $d(v)$.
Let $\mT$ be the set of all such labeled trees in which the labels satisfy
\beq 
2n_v \geq d(v)\,,\quad n_v\ge 3k+4\,,\quad p_g\ge 0.
\eeq
Let $c>0,\epsilon_1>0,t>0$ and $s=te^{t^{1/4}}$ with 
\beq\label{t7}
cs < \frac{1}{4},\qquad t\leq(3k+3)^{-4}\,,\qquad t^{3(k-n)-1}\leq\epsilon_1.
\eeq 
Assign to every tree $T$ a weight 
\beq \label{treeweight}
\mu(T)=\prod_{v\in T}\bigg({{2n_v}\atop{d(v)}}\bigg)(cs)^{n_v}
\prod_{g\in T}\exp\big(-\epsilon_1p_g t^{3n+3}\big) ,
\eeq
Then for $m \geq 3k+4$,
\beq
\sum_{T \in \mT: n_{r(T)}=m}\mu(T)\le
\Bigl( cse^{3\epsilon_1^{-1}t^{3(k-n)+1}} \Bigr)^m.
\label{tree}\eeq
\end{Lem}

\begin{proof}
This lemma corresponds to Lemma~2.13 of ref. \cite{DM}.
%but our statement is simpler.
The proof proceeds by induction on the number $|T|$ of vertices of $T$,
starting with
\beq
\sum_{T: n_{r(T)}=m\atop|T|=1} \mu(T)= (cs)^m
\eeq
and then assuming
\beq
\sum_{T: n_{r(T)}=m \atop|T|\le q} \mu(T)\le 
\Bigl( cse^{3\epsilon_1^{-1}t^{3(k-n)+1}} \Bigr)^m
\eeq
for all $m$, for some $q\ge1$, so that also
\begin{align}
  \sum_{T:|T| \leq q} \mu(T) &\leq \frac{ \Bigl( cse^{3\epsilon_1^{-1}t^{3(k-n)+1}} \Bigr)^{3k+4} }
    {1 - cse^{3\epsilon_1^{-1}t^{3(k-n)+1}} } 
    \leq e^{2cs} \Bigl( cse^{3\epsilon_1^{-1}t^{3(k-n)+1}} \Bigr)^{3k+4}.
\end{align}
Any tree with $|T|=q+1$ can be
decomposed into its root $v_0$, edges $g_1,\dots,g_f$ which lead to the
vertices $v_1,\dots,v_f$, and subtrees $T_1,\dots,T_f$, 
such that $|T_i|\le q$ and $r(T_i)=v_i$. We have the formal identity
\begin{align}
\sum_{T: n_{r(T)}=m} &\mu(T) \nonumber \\ 
&= (cs)^m \sum_{f=0}^{2m} \bigg({{2m}\atop{f}}\bigg)
\prod_{i=1}^f\bigg(\sum_{p_i=0}^\infty \exp\big(-\epsilon_1p_i t^{3n+3}\big)
\sum_{T_i \in \mT}\mu(T_i)\bigg)%\hskip1cm 
\end{align}
and the related induction bound
\beqa
&&\sum_{T: n_{r(T)}=m \atop|T|\le q+1} \mu(T) \nonumber \\ 
&&\le
(cs)^m\sum_{f=0}^{2m} \bigg({{2m}\atop{f}}\bigg)
\prod_{i=1}^f\bigg(\sum_{p_i=0}^\infty \exp\big(-\epsilon_1p_i t^{3n+3}\big)
\sum_{T_i: |T_i|\le q}\mu(T_i)\bigg) \nonumber \\ 
&&\le
(cs)^m \sum_{f=0}^{2m} \bigg({{2m}\atop{f}}\bigg)
\bigg({{1}\over{1-\exp(-\epsilon_1t^{3n+3})}}
e^{2cs}\Bigl(cse^{3\epsilon_1^{-1}t^{3(k-n)+1}}\Bigr)^{3k+4}\bigg)^{f} 
\nonumber \\  
&&\le
(cs)^m  \bigg(1+(\epsilon_1t^{3n+3})^{-1}
e^{2cs}\Bigl(cse^{3\epsilon_1^{-1}t^{3(k-n)+1}}\Bigr)^{3k+4}\bigg)^{2m} 
\nonumber \\   
&&=
(cs)^m \bigg(1+\epsilon_1^{-1}e^{2cs}t^{3(k-n)+1}
\Bigl(e^{3\epsilon_1^{-1}t^{3(k-n)+1}}\Bigr)^{3k+4}\bigg)^{2m} 
\nonumber \\ 
&&\le
\Bigl( cse^{3\epsilon_1^{-1}t^{3(k-n)+1}} \Bigr)^m,
\eeqa
where (\ref{t7}) was used in the last inequality, concluding the proof
of the Lemma. 
\end{proof}

% ======== {\bf On peut egalement borner $\sum_{T: r(T)=v_0} (n_0/2)^2 \mu(T)$}

\begin{Lem}
Assume that for all contours $\Ga$,
\beq\label{phigammamu3}
\mu(\Ga) \le  s^{\|\Ga\|} \exp\left(
2K_6\Big(|\tga^{ext}|+\sum_i|\tga_i|+\sum_j|\tga_j^{int}|\Big)s^2 
  -|\Supp(\Gamma)| \ep_1 t^{3n+3} \right),
\eeq
with $\ep_1>0$, $k\ge\max(2n,8)$, $K_6$ from Lemma \ref{Lem3}, $s=te^{t^{1/4}}$ and
\beq\label{t8}
t<(4K_6)^{-1} \wedge (3k+3)^{-4}\,,\qquad t^{3(k-n)-1}\leq\epsilon_1/4\,.
\eeq 
Then
\beq\label{contourconv}
\sum_{\Gamma:\,\Supp(\Gamma)\ni0}\mu(\Gamma)\le (16\,t)^{3k+4}.
\eeq
\label{LemLast}
\end{Lem}

\begin{proof}
For each base perimeter $\ti\ga$ of a cylinder $\ga \in \Ga$, $\ga \neq \ga^{ext}$,
we select as a designated point $z(\ti\ga) \in \ti\ga$ the lowest among the leftmost points in $\ti\ga$.
From each $z(\ti\ga)$, draw a horizontal line segment 
(necessarily a lattice path in $\ZZ^2 + (\frac{1}{2},\frac{1}{2})$)
leftward from $z(\ti\ga)$ until it reaches the base perimeter of some 
other cylinder in $\Ga$.  The line segment may have length 0.
Reverse the orientation of this segment so it points rightward.  
Note these segments are necessarily disjoint for distinct $\ti\ga$ from $\Ga$, and 
contained in $\Supp(\Ga)$.
We obtain a rooted labeled tree as in Lemma \ref{treeLem}, connecting all the $\ti\ga\in\Ga$, by taking each 
$\ti\ga$ from $\Ga$ to be a vertex $v$, with label $2n_v=|\ti\ga|$ and with $\ti\ga^{ext}$ as root, and 
taking each line segment to be an oriented edge $g$, with label $p_g$ equal to
its length; the ordering of the edges emanating from each $v$ is given by counterclockwise 
ordering of the segments emanating from $\ti\ga$, starting from $z(\ti\ga)$.

We may then identify all those rooted labeled trees with vertices $\ti\ga$ which are
isomorphic as rooted trees and which have the same values $n_v,p_g$, taking into account the ordering
of edges;
the resulting structure defines a rooted labeled tree $T=T(\Ga)$ of the type in Lemma \ref{treeLem}.
%with root $v_0=\ti\ga^{ext}(\Ga)$, 
%connecting all the $\ti\ga\in\Ga$. 
In order to apply Lemma~\ref{treeLem}, we need
\beq
\sum_{\Ga: T(\Ga)=T \atop \Supp(\Ga) \ni 0}\mu(\Ga)\le\mu(T).
\eeq
Given $T$, to specify a contour $\Ga$ with $T(\Ga)=T$, we need first to choose for the 
root $r(T)$ a cylinder $\ga$ with $|\ti\ga|=2n_{r(T)}$ and $0\in\bar\ga$, 
then choose $d(r(T))$ sites along $\ti\ga$ to be the starting points of the line segments,
and then choose a length $p_g$ for each such line segment $g$ emanating from $\ti\ga$.
Second we must choose for each nonroot vertex $v$ a cylinder $\ga$ 
with $|\ti\ga|=2n_v$ for which $\ti\ga$ 
passes through a fixed site (say $(\frac{1}{2},\frac{1}{2})$), then (as with the root) 
choose $d(v)$ sites along 
$\ti\ga$ and a length $p_g$ for each segment $g$ emanating from $\ti\ga$.
Thus the factor $\big({{2n_v}\atop{d(v)}}\big)$ in \eqref{treeweight} is the number of choices of
starting points of $d(v)$ segments starting from $\ti\ga$ with
$|\ti\ga|=2n_v$. The factor with $p_g$ is obtained from \eqref{phigammamu3} using
\beq
\Supp(\Ga)\ge\sum_{g\in T}p_g
\eeq
and the power of $cs$ is obtained using \eqref{phigammamu3} and the bound
\beqa
\sum_{\gamma,|\tga|=2n_v \atop \Supp(\ga) \ni 0}s^{{1\over2}|\ti\ga|L(\ga)}e^{2|\ti\ga|K_6t^2} &<&
2n_v^2(1-s^{n_v})^{-1}(3^2se^{4K_6t^2})^{n_v} \nonumber \\
&<& 3n_v^2(10s)^{n_v} \nonumber\\
&<& (14\,t)^{n_v}.
\eeqa
We see that we fit into the conditions of Lemma~\ref{treeLem}, with 
$c=14$. Then the right-hand side of (\ref{tree}) is bounded by 
$(15t)^m$.  Summing over $m \geq 3k+4$ proves Lemma~\ref{LemLast}. 
\end{proof}

\medskip

Applying Lemma~\ref{LemLast} with $\ep_1=\ep/4$, we obtain
\beq
\sum_{\Gamma':\,\Gamma'\not\sim\Gamma}\mu(\Gamma')\le\sum_{x\in\Supp(\Gamma)}
\sum_{\Gamma':\,\Supp(\Gamma')\ni x}\mu(\Gamma')
\le |\Supp(\Gamma)|(16\,t)^{3k+4}.
\label{B1}\eeq
Similarly,
\beq
\sum_{\Gamma':\,\Gamma'\not\sim\omega}
\mu(\Gamma')\le\sum_{x\in\tga^{ext}_\omega}
\sum_{\Gamma':\,\partial\Supp(\Gamma')\ni x}\mu(\Gamma')
\le |\tga^{ext}_\omega|(16\,t)^{3k+4}.
\label{B2}\eeq

Recall $t_2(k,\ep)$ from Lemma \ref{Lem3}, and recall the sufficient 
condition $t\le \ep^{1/12}/2000$ for \eqref{tepsilon}.  Note that since $k-n \geq k/2$, to satisfy the second inequality 
in \eqref{t8} with $\ep_1 = \ep/4$, it suffices that $t \leq (\ep/16)^{2/3k}/12$, and this in turn 
follows from the condition $t \leq \ep^{1/12}/2000$.  After reducing $K_5$ if necessary, this last condition 
follows from the condition $t \leq K_5(\ep \vee 1)$ in Lemma \ref{Lem3}.

\begin{Lem}\label{cvgCluster}
Let the level $n$ be given and 
let $k\ge\max(8,2n)$.
Assume that the inequalities (\ref{Ass1}) in Theorem \ref{Th1} are satisfied.  
Then for $t\le t_5(k,\epsilon) \equiv \min(t_1(k)/12,K_5(\ep \vee 1))$,
the cluster expansion (\ref{expansion}) of the surface tension converges.  
Here $K_5$ is from Lemma \ref{Lem3}.
\end{Lem}

\begin{proof}
We must check (\ref{i1}) with $\vphi(\om)=\vphi_{t,u}(\om)$,
$\mu(\om)=\vphi_{s,0}(\om)$, $s=te^{t^{1/4}}$, with $\vphi(\Ga)$ given by
(\ref{weight}) and (\ref{weight2}), bounded as in (\ref{phigamma}), and with $\mu(\Ga)$
given by (\ref{weight2mu}), bounded as in (\ref{phigammamu}).

In the case associated to an elementary perturbation $\omega$, (\ref{i1})
takes the form
\beq\label{i1omGa}
|\vphi(\om)|\le\mu(\om)
e^{-\sum_{\om'\not\sim\om}\mu(\om')-\sum_{\Ga'\not\sim\om}\mu(\Ga')}.
%\label{i1}
\eeq
Since $\om$ is elementary we have both $|\olg^{ext}_\om| \leq (3k+3)^2/4$ and 
$|\olg^{ext}_\om| \leq |\ti\ga^{ext}_\omega|^2/4$, 
so also $|\olg^{ext}_\om| \leq (3k+3)|\ti\ga^{ext}_\omega|/4$.
Hence from \eqref{omegaA} and (\ref{B2}) we have
\begin{align}
s^{1/2} |\olg^{ext}_\om| &+
  \sum_{\omega':\,\omega'\not\sim\omega}\mu(\omega)+
  \sum_{\Gamma':\,\Gamma'\not\sim\omega}\mu(\Gamma') \notag \\
&\leq 2s^{1/2}|\olg^{ext}_\om|+(16\,t)^{3k+4}|\ti\ga^{ext}_\omega| \notag \\
&\leq \bigl(2^{-{7/4}}s^{1/4}+(16\,t)^{3k+4}\bigr)|\ti\ga^{ext}_\omega|.
\end{align}
Then as in \eqref{ts4} and \eqref{ts}, 
(\ref{i1omGa}) will be satisfied if for every elementary cylinder
$\ga$ of length one,
\beq
t^{{1\over2}|\ti\ga|}<s^{{1\over2}|\ti\ga|}
e^{-(2^{-{7/4}}s^{1/4}+(16\,t)^{3k+4})|\ti\ga|},
\eeq
or, since $s=te^{t^{1/4}}$,
\beq
\frac{1}{2}t^{1/4}>2^{-{7/4}}s^{1/4}+(16\,t)^{3k+4},
\eeq
which is easily verified for $t<1/32$ with $k\ge8$.

In the case associated to a contour
$\Gamma=(\gamma^{ext},\gamma_i,\gamma^{int}_j)$,   
(\ref{i1}) takes the form
\beq\label{i1Ga}
|\vphi(\Ga)|\le\mu(\Ga)
e^{-\sum_{\om'\not\sim\Ga}\mu(\om')-\sum_{\Ga'\not\sim\Ga}\mu(\Ga')}.
%\label{i1}
\eeq
From (\ref{8k}), a slight variant of the first inequality in
(\ref{omega}), and (\ref{B1}), taking into account
the compatibility condition between elementary perturbations and
contours given before Lemma \ref{contourpert},  we have
\begin{align}
\sum_{\omega':\,\omega'\not\sim\Gamma} &\mu(\omega')+
\sum_{\Gamma':\,\Gamma'\not\sim\Gamma}\mu(\Gamma') \notag \\
&\le 3000s^2\, \Big( 2|\ti\ga^{ext}|+\sum_j (|\ti\ga_j^{int}| + 4) \Big)
+|\Supp(\Gamma)|(16\,t)^{3k+4} \notag \\
&\le 6000s^2\, \Big( |\ti\ga^{ext}|+\sum_j |\ti\ga_j^{int}| \Big)
+|\Supp(\Gamma)|(16\,t)^{3k+4}.
\end{align}
From (\ref{weight2}) and (\ref{weight2mu}) we have
\beq
{\vphi(\Ga)\over\mu(\Ga)}=
\Bigl({t\over s}\Bigr)^{{1\over2}\|\Ga\|} 
e^{-|\Supp\Ga|(16t)^{3k+4}},
\eeq
and hence a sufficient condition for \eqref{i1Ga} is
\begin{align}
t^{{1\over2}\|\Ga\|} \leq s^{{1\over2}\|\Ga\|} 
  e^{-10000s^2\, \big(|\ti\ga^{ext}|+\sum_j|\ti\ga_j^{int}|\big) },
\label{phigammamu2}
\end{align}
which is satisfied for $s=te^{t^{1/4}}$, $t<t_1(k)$ and $k\ge\max(8,2n)$.
Then the Convergence Theorem gives (\ref{i2}), that is,
\beq
\sum_{X\ni\om}|\varphi_u^{\rm T}(X)|\le\mu(\om)\,,\qquad 
\sum_{X\ni\Ga}|\varphi_u^{\rm T}(X)|\le\mu(\Ga)\, .
%\label{i2}
\eeq
As in \eqref{newsum}, these together with \eqref{8k} and Lemma \ref{LemLast}
yield convergence of the cluster expansion \eqref{expansion} for the surface
tension.  We also obtain convergence of cluster expansions 
for the correlation functions, see e.g. (14) in \cite{M}.
\end{proof}

For $t_5(k,\ep)$ from Lemma \ref{cvgCluster}, after a further reduction of $K_5$ if necessary, we have 
\[
  t_5(\max(2n,8),\ep) \geq \min\left( \frac{1}{(6n+3)^4},K_5(\ep \vee 1) \right).
  \]
Therefore (1) in Theorem \ref{Th1} with
\beq
t_0(n,\ep)=\min\left( \frac{1}{(6n+3)^4},K_5(\ep \vee 1) \right)
\eeq
follows as a corollary, in the same way as done in \cite{LM} for the model with an
external field. 

We turn to the proof of (2) in Theorem \ref{Th1}, which we obtain from
convergence of the cluster expansion in a similar way as was done in \cite{LM}.
As in the case of the SOS model in an external field, FKG inequalities
can be used to reduce the problem to constant boundary conditions. 
We restrict attention to rectangular $\La$.
Let $h\ne n$ be taken as a constant boundary condition. 
This can be obtained from the $n$-boundary condition by (i) requiring the
presence of a contour $\Ga_0=\{\ga^{ext},\ga_i,\ga^{int}_j\}$ such that 
\begin{itemize}
\item $\bar\ga^{ext}=\La$
\item $I(\ga^{ext})=h$
\item $\Supp^{ext}(\Ga_0)\supset\{x\in\La:d(x,\La^c)=1\}$,
\end{itemize}
where $d(\cdot)$ is the $\ell^\infty$-distance (or the euclidean distance), and (ii)
replacing the weight $\varphi(\Ga_0)$ from \eqref{weight} with a weight 
$\vphi^*(\Ga_0)$ similar to \eqref{weight} except that the partition function 
$Z_k^*\big(\Supp^{ext}(\Gamma),I(\gamma^{ext})\big)$ excludes elementary perturbations 
with support intersecting $\{x\in\La:d(x,\La^c)=1\}$.
Then Lemma \ref{LXi5} and Lemma \ref{contourpert} will hold for the
corresponding partition function, denoted $\Xi(\Lambda,n,h)$, with
summations including a $\Ga_0$ as above:
\beq 
\Xi(\La,n,h)=e^{u\delta(n)|\La|}\sum_{\Ga_0}\vphi^*(\Ga_0)
\sum_{\{\Ga_i,\omega_j\}}\prod_i\vphi(\Ga_i)\prod_j\varphi(\omega_j). 
\label{polymer0}\eeq
The sum runs over all compatible sets $\{\Ga_0,\{\Gamma_i,\omega_j\}\}$ 
of contours and elementary perturbations contained in $\Lambda$ such that 
$E(\Ga_0)=E(\Gamma_i)=E(\omega_j)=n$, for all $i,j$. 
The probability that the configuration includes a given $\Ga_0$ is
%in agreement with (\ref{polymer}),
\beq\label{PGa0}
\mu_\La(\Ga_0|n,h)=\Xi(\La,n,h)^{-1}e^{u\delta(n)|\La|}\vphi^*(\Ga_0)
\sum_{\{\Ga_i,\omega_j\}\sim\Ga_0}\prod_i\vphi(\Ga_i)\prod_j\varphi(\omega_j)
\eeq
where the sum is over compatible families $\{\Ga_i,\omega_j\}$ compatible 
with $\Ga_0$. 

\begin{Lem}\label{LemGa0}
Given $\ep>0$ and $n \geq 0$, there exists $t_6(\ep,n)$ as follows.
Let $t<t_6(\ep,n)$ and let $c_{n,h,\ep,t}=9|n-h|\ep^{-1}t^{-3n-3}\ln t^{-1}$. Then for
contours $\Ga_0$ as in (i) above,
\beq\label{SGa0}
\sum_{\Ga_0:|\Supp\Ga_0|>c_{n,h,\ep,t}|\p\La|}\mu_\La(\Ga_0|n,h)
<t^{|n-h||\p\La|}.
\eeq
\end{Lem}
\begin{proof}
The sum over $\{\Ga_i,\omega_j\}$ in (\ref{PGa0}) is a partition
function which can be exponentiated, with a cluster expansion obeying the
same bounds as the full partition function, only with fewer terms:
\beq
\sum_{\{\Ga_i,\omega_j\}\sim\Ga_0}\prod_i\vphi(\Ga_i)\prod_j\varphi(\omega_j)
=\exp\Bigl(\sum_{X\sim\Ga_0}\vphi_u^{\rm T}(X)\Bigr)
\eeq
where the clusters $X$ are clusters of polymers contained in $\La$,
compatible with $\Ga_0$. A particular $\Ga_0$ consisting of two cylinders is
$\Ga_{00}=\{\ga^{ext},\ga^{int}\}$ with $\ga^{ext}$ as in (i) and
$\bar\ga^{int}=\{x\in\La:d(x,\La^c)>1\}, I(\ga^{int})=n$. Then for each $\Ga_0$,	
\beqa\label{Ga0Ga00}
{\mu_\La(\Ga_0|n,h)\over\mu_\La(\Ga_{00}|n,h)}
&=&{\vphi^*(\Ga_0)\over\vphi^*(\Ga_{00})}
\exp\Bigl(-\sum_{X\not\sim\Ga_0\atop X\sim\Ga_{00}}\vphi_u^{\rm T}(X)\Bigr)\cr
&<&{\vphi^*(\Ga_0)\over\vphi^*(\Ga_{00})}.
\eeqa
Indeed the sum in (\ref{Ga0Ga00}) may be written as
\beq
\sum_{X\not\sim\Ga_0\atop X\sim\Ga_{00}}
  \vphi_u^{\rm T}(X)=\sum_{x\in\La}\sum_{ {X\not\sim\Ga_0\atop X\sim\Ga_{00}} \atop {\Supp X \ni x}}
  \frac{\vphi_u^{\rm T}(X)}{|\Supp X|}.
\eeq
For each $x$, analogously to Lemma \ref{cvgCluster} we have a convergent cluster expansion, 
uniformly in $n$ and $h$, with
leading term $t^2$ corresponding to a unit cube excitation
up or down. The conditions in the sum over $X$ does not allow both up and down
excitations. Given the hypotheses over $t$, the remainder
$O(t^3)$ can be uniformly bounded by $t^2/2$. Therefore the sum
in (\ref{Ga0Ga00}) is positive. 

The same argument, starting from (\ref{weight2}) shows that
\begin{equation}
  \vphi^*(\Ga_{00})>t^{|n-h||\p\La|}e^{-3t^2|\p\La|} \quad \text{and} \quad 
    \vphi^*(\Ga_0) < e^{3t^2|\p\La|} \vphi(\Ga_0).
  \end{equation}
This and $\mu_\La(\Ga_{00}|n,h)<1$ in  (\ref{Ga0Ga00}) give
\beq\label{muGa0}
\mu_\La(\Ga_0|n,h)<t^{-|n-h||\p\La|}e^{6t^2|\p\La|}\vphi(\Ga_0).
\eeq
From (\ref{phigamma}) and $t<s$, we see that 
$e^{{\ep\over4}t^{3n+3}|\Supp\Ga|}\vphi(\Ga)$
obeys the hypotheses of Lemma \ref{LemLast} with $\ep_1=\ep/4$, so that
\beq
\sum_{\Gamma:\,\Supp(\Gamma)\ni0}e^{{\ep\over4}t^{3n+3}|\Supp\Ga|}\vphi(\Gamma)
\le (16\,t)^{3k+4}
\eeq%\label{LemLast}
and therefore
\beq
\sum_{\Gamma:\,\Supp(\Gamma)\ni0\atop|\Supp(\Gamma)|>c_{n,h,\ep,t}|\p\La|}
\vphi(\Gamma)\le (16\,t)^{3k+4} \exp\left( -{\ep\over4}t^{3n+3}c_{n,h,\ep,t}|\p\La| \right).
\eeq
Hence, with (\ref{muGa0}),
\begin{align}
&\sum_{|\Supp\Ga_0|>c_{n,h,\ep,t}|\p\La|} \mu_\La(\Ga_0|n,h) \notag\\
&\qquad\qquad<|\La|(16\,t)^{3k+4}t^{-|n-h||\p\La|}\exp\left(
6t^2|\p\La|-{\ep\over4}t^{3n+3}c_{n,h,\ep,t}|\p\La| \right),
\end{align}
which implies (\ref{SGa0}).
\end{proof}

Lemma \ref{LemGa0} shows that a translation invariant Gibbs state
obtained from the $h$ boundary condition is the same as that obtained from
the $n$-boundary condition.
We continue the proof of (2) in Theorem \ref{Th1},
going back to the $n$-boundary condition and estimating:
\begin{align*}
\mu_n&(\{\phi_0\ne n\}) \\
&<\Xi(\La,n)^{-1}e^{u\delta(n)|\La|}\biggl[
\sum_{\Supp\om\ni0}\vphi(\om)
\sum_{\{\Ga_i,\omega_j\}\sim\om}\prod_i\vphi(\Ga_i)\prod_j\varphi(\omega_j)\\
&\hskip 5cm \sum_{\Supp\Ga\ni0}\vphi(\Ga)
\sum_{\{\Ga_i,\omega_j\}\sim\Ga}\prod_i\vphi(\Ga_i)\prod_j\varphi(\omega_j)
\biggr] \\
&=\Xi(\La,n)^{-1}e^{u\delta(n)|\La|}\biggl[
\sum_{\Supp\om\ni0}\vphi(\om)e^{\sum_{X\sim\om}\vphi_u^{\rm T}(X)}+
\sum_{\Supp\Ga\ni0}\vphi(\Ga)e^{\sum_{X\sim\Ga}\vphi_u^{\rm T}(X)}
\biggr]\\
&=\sum_{\Supp\om\ni0}\vphi(\om)e^{-\sum_{X\not\sim\om}\vphi_u^{\rm T}(X)}+
\sum_{\Supp\Ga\ni0}\vphi(\Ga)e^{-\sum_{X\not\sim\Ga}\vphi_u^{\rm T}(X)}\hskip3cm\\
&< 3000t^2e^{2t^{1/4}}+(16t)^{3k+4},
\end{align*}
using (\ref{8k}) and Lemma \ref{LemLast}.

Similarly, for the proof of (3) in Theorem \ref{Th1}, using a cylinder $\om$
with $|\ti\om|=4$ and $L(\om)=n$, $n\ne0$ for definiteness,
\beqa
\rho_0=\mu_n(\{\phi_0=0\})&>&\Xi(\La,n)^{-1}e^{u\delta(n)|\La|}\vphi(\om)
\sum_{\{\Ga_i,\omega_j\}\sim\om}\prod_i\vphi(\Ga_i)\prod_j\varphi(\omega_j)\cr
&=&\Xi(\La,n)^{-1}e^{u\delta(n)|\La|}\vphi(\om)e^{\sum_{X\sim\om}\vphi_u^{\rm T}(X)}\cr
&=&\vphi(\om)e^{-\sum_{X\not\sim\om}\vphi_u^{\rm T}(X)}\cr
&=&t^{2n}e^{O(t^2)},
\eeqa
completing the proof of Theorem \ref{Th1}.
%\newpage
\bigskip

{\bf Acknowledgments. }
One of the authors (S.\ M-S.) wishes to thank 
the Isaac Newton Institute for Mathematical Sciences, University of Cambridge, 
for warm hospitality and support during % to attend 
the programme on Combinatorics and Statistical Mechanics (January--June 2008),
% during the elaboration of this work. 
where part of this work was carried out. 
%He is grateful to 
%the organizers of the programme for their kind invitation.
%Two of us, 
F.\ D. and S.\ M-S., remember with gratitude the kind hospitality 
of Mrs. Yvonne Bodin at Pordic, la Fosse Argent, during studious summer
vacations. 
The possibility to meet offered to the authors by the 
Laboratoire de Physique Th{\'e}orique et Modelisation, 
Universit{\'e} de Cergy-Pontoise, and 
the Centre de Physique Th{\'e}orique, CNRS, Marseille, 
is also acknowledged.  The research of K.\ A. was supported by NSF grants
DMS-0405915 and DMS-0804934.

%\newpage

\end{document}